\documentclass[final,1p,times]{elsarticle}

\usepackage{amssymb}
\journal{Computer Methods in Applied Mechanics and Engineering}

\usepackage{amssymb}
\usepackage{amsthm,amsmath}
\usepackage{subcaption}
\usepackage{algorithm,algorithmic}
\usepackage{color}
\usepackage[notrig]{physics}
\usepackage{hyperref}
\usepackage[T1]{fontenc}
\usepackage[german,french,english]{babel}

\newtheorem{thm}{Theorem}[section]
\newtheorem{prop}[thm]{Proposition}
\newtheorem{lem}[thm]{Lemma}

\theoremstyle{definition}

\newtheorem{example}[thm]{Example}

\newcommand\R{\mathbb{R}}
\newcommand\N{\mathbb{N}}
\newcommand\C{\mathbb{C}}
\newcommand\Z{\mathbb{Z}}

    \definecolor{darkgreen}{RGB}{0,127,0}
    
\newcommand{\Dev}{\ensuremath\mathop{\mathrm{Dev}}}

\newcommand{\paramfineh}{\rho_h}
\newcommand{\paramunifh}{t_h}

\begin{document}

\begin{frontmatter}



\title{Data--Driven Games in Computational Mechanics}


\author{
K.~Weinberg${}^a$, L.~Stainier${}^b$, S.~Conti${}^c$ and M.~Ortiz${}^{c,d}$
}
%\author{
%K.~Weinberg${}^1$, L.~Stainier${}^2$, S.~Conti${}^3$ and M.~Ortiz${}^{3,4}$
%}
%\address
%{
%${}^1$Chair of Solid Mechanics, University of Siegen, Paul-Bonatz-Str. 9-11, 57076 Siegen, Germany
%\\
%${}^2$Nantes Universit\'e, Ecole Centrale Nantes, CNRS, GeM, UMR 6183, F-44000 Nantes, France
%\\
%${}^3$Institut f\"ur Angewandte Mathematik and Hausdorff Center for Mathematics, Universit\"at Bonn, Endenicher Allee 60, 53115 Bonn, Germany.
%\\
%${}^4$California Institute of Technology, Engineering and Applied Science Division, Pasadena CA, 91125, USA
%}
\affiliation{organization={Chair of Solid Mechanics, Department of Mechanical Engineering, Universit\"at Siegen},
            addressline={Paul-Bonatz-Str. 9-11},
            city={Siegen},
            postcode={57076},
            country={Germany}}
\affiliation{organization={Nantes Universit\'e, Ecole Centrale Nantes, CNRS, GeM, UMR 6183},
            city={Nantes},
            postcode={F-44000},
            country={France}}
\affiliation{organization={Institut f{\"u}r Angewandte Mathematik and Hausdorff Center for Mathematics, Universit\"at Bonn}, %Department and Organization
            addressline={Endenicher Allee 60},
            city={Bonn},
            postcode={53115},
            country={Germany}}
\affiliation{organization={California Institute of Technology, Engineering and Applied Science Division},             city={Pasadena},
            postcode={91125},
            state={CA},
            country={USA}}

\begin{abstract}
We resort to game theory in order to formulate Data-Driven methods for solid mechanics in which stress and strain players pursue different objectives. The objective of the stress player is to minimize the discrepancy to a material data set, whereas the objective of the strain player is to ensure the admissibility of the mechanical state, in the sense of compatibility and equilibrium. We show that, unlike the cooperative Data-Driven games proposed in the past, the new non-cooperative Data-Driven games identify an effective material law from the data and reduce to conventional displacement boundary-value problems, which facilitates their practical implementation. However, unlike supervised machine learning methods, the proposed non-cooperative Data-Driven games are unsupervised, {\sl ansatz}--free and parameter--free. In particular, the effective material law is learned from the data directly, without recourse to regression to a parameterized class of functions such as neural networks. We present analysis that elucidates sufficient conditions for convergence of the Data-Driven solutions with respect to the data. We also present selected examples of implementation and application that demonstrate the range and versatility of the approach.
\end{abstract}


%\item Game theory is concerned with scenarios involving several players seeking strategies that strive to minimize their respective costs. It has been of foundational importance in a variety of fields, but applied sparingly to mechanical problems. Here we bring game-theoretical concepts to computational solid mechanics.
%\item From a mechanics perspective, deterministic game-theoretical problems may be regarded as instances of coupled problems with a particular variational structure.
%\item We resort to game theory in order to formulate Data-Driven methods for solid mechanics in which stress and strain players pursue different objectives.
%\item The new non-cooperative Data-Driven games identify an effective material law from the data and reduce to conventional displacement boundary-value problems, which facilitates their practical implementation.
%\item The new non-cooperative Data-Driven games are unsupervised, ansatz-free, parameter-free and follow in the vein of prior cooperative Data-Driven games by striving to effect a direct, unsupervised, model-free connection between data and prediction.
%\item We present selected examples of implementation and application that demonstrate the range and versatility of the approach.


\end{frontmatter}

%% main text %%%%%%%%%%%%%%%%%%%%%%%%%%%%%%%%%%%%%%%%%%%%%%%%%%%%%%%%%%%%%%%%%%%%%%%%%%%%%%%%%%%%%%%%%%%%%%%%%%%%%%%%%%%%%%%%%%%%%%%%%%%%%%%%%%%%%%%%%%%%%%%%%%%%%%%%
\section{Introduction}\label{sec:intro}

Game theory concerns itself with scenarios involving several players seeking strategies that strive to minimize their respective costs or, equivalently, maximize their respective payoffs. In general, the cost of one player's strategy depends on the strategies adopted by the remaining players and, in consequence, the optimal strategies of the players are coupled to each other. In one scenario, the players optimize their strategies {\sl cooperatively} by striving to minimize a joint cost computed as a weighted average of all costs, a condition known as {\sl Pareto optimality}. In another scenario, the players proceed {\sl non-cooperatively} by each seeking to minimize its own cost independently. Game theory was pioneered, {\sl inter aliis}, by economists Vilfredo Pareto \cite{Pareto:1909} and John Nash \cite{Nash:1951}, and mathematician John von Neumann \cite{Neumann:1928}, in seminal contributions and has been of foundational importance in a variety of fields, including economics, social sciences, evolutionary biology, computer science, and others.

Despite its phenomenal success in other fields, game theory has been applied sparingly to mechanics or not at all. From a mechanics perspective, deterministic game-theoretical problems may be regarded as instances of {\sl coupled problems} with a particular variational structure (cf.~Section~\ref{JPLu7K} for a brief review). Among these problems, {\sl inf-sup} problems may be regarded as {\sl zero-sum} games. These correspondences and others open up the opportunity of bringing game-theoretical concepts, tools and results to bear on a wide range of problems in mechanics, a potential that remains largely unattained at present.

In this work, we resort to game theory in order to formulate Data-Driven methods for solid mechanics in which {\sl stress} and {\sl strain players} pursue different objectives: the objective of the stress player is to minimize the discrepancy to a material data set, whereas the objective of the strain player is to ensure the admissibility of the mechanical state, in the sense of compatibility and equilibrium. We show that, unlike the cooperative Data-Driven games proposed in the past \cite{Kirchdoerfer:2016, Kirchdoerfer:2017, conti:2018, Prume:2023, ContiHoffmannOrtiz2023}, the new non-cooperative Data-Driven games identify an {\sl effective material law} from the data and reduce to conventional displacement boundary-value problems, which facilitates their practical implementation. In particular, the Data-Driven effective material law can be conveniently implemented as a standard user-supplied material in commercial finite-element software.

This change of mood notwithstanding, it bears emphasis that, unlike supervised machine learning methods, the proposed non-cooperative Data-Driven games are unsupervised, {\sl ansatz}--free and parameter--free. In particular, the effective material law is learned from the data {\sl directly}, without recourse to regression to a parameterized class of functions such as neural networks. In this sense, the new non-cooperative Data-Driven games follow in the vein of prior cooperative Data-Driven games \cite{Kirchdoerfer:2016, Kirchdoerfer:2017, conti:2018, Prume:2023, ContiHoffmannOrtiz2023} by striving to effect a direct, unsupervised and model--free connection between data and prediction.

By identifying stress and strain as players, the proposed non-cooperative Data-Driven games fall within the set-oriented formulation of mechanics problems, briefly reviewed in Section~\ref{M8mXkn0}. The connection between such problems and game theory is introduced in Section~\ref{2bzQRE} in the context of cooperative games, in which stress and strain strive to achieve a common objective of minimizing distance to a material set while satisfying the field equations of compatibility and equilibrium. This cooperative strategy reproduces---and provides a game-theoretical interpretation for---set-oriented Data-Driven methods proposed in \cite{Kirchdoerfer:2016, Kirchdoerfer:2017, conti:2018, Prume:2023, ContiHoffmannOrtiz2023}. The transition from cooperative to non-cooperative moods is presented in Section~\ref{7ATkpe} by regarding stress and strain as adversarial players, each pursuing its own objective. Evidently, this strategy is suboptimal with respect to the cooperative Data-Driven strategy, but it offers the significant practical advantage of reducing to a conventional and well-posed displacement problem, Section~\ref{jtA0LI}, amenable to approximation, Section~\ref{kEV3cF}.

A particularly important case concerns approximations based on empirical point-data sets, Section~\ref{u5RpCF}, e.~g., measured empirically or computed from micromechanics. The central question then concerns the elucidation of conditions on the data that ensure the convergence of the Data-Driven solutions to the solution of the underlying---and unknown---material law. We provide rigorous conditions for convergence with respect to the data for two different scenarios: i) {\sl Uniformly convergent data}, in which the sampling error decreases as data is added to the material-data set in a uniform manner controlled by strict upper bounds, Section~\ref{dQi63n}, and ii) {\sl noisy data with outliers}, in which the data concentrates around the limiting material law in a weak or average sense that allows for the presence of outliers, Section~\ref{kuP0lq}. In this second scenario, convergence requires regularization in the form of local data averages taken over carefully chosen local neighborhoods in order to mitigate the effect of the outliers, Section~\ref{9kYy4s}.

Finally, we present selected examples of implementation and application that demonstrate the range and versatility of the approach, Section~\ref{Nd4BgR}. In particular, the examples illustrate how the approach can be implemented within a standard displacement finite-element framework in any dimension, with and without regularization, and using interative solvers such as dynamic relaxation and Newton-Raphson iteration. The examples also bear out the type of convergence with respect to the data anticipated by the analysis.

\section{Elements of game theory} \label{JPLu7K}

{\sl Game theory} is a well-developed branch of mathematics (cf., e.~g., \cite{Roubicek:2020} for a general modern account), but it has not been extensively applied to solid mechanics and may, therefore, stand a brief review. We specifically collect basic elements of the theory required in subsequent developments.

For present purposes, it suffices to consider two-player, finite dimensional games (cf., e.~g., \cite{Leitmann:1974, Ferguson:2020}). Specifically, we consider two players seeking strategies $u \in \mathbb{R}^m$ and $v \in \mathbb{R}^n$ who strive to minimize their costs $F(u,v)$ and $G(u,v)$, respectively. They can do so {\sl cooperatively}, by minimizing a weighted average of their costs
\begin{equation}\label{Jw51yf}
    J(u,v) = \lambda_F F(u,v) + \lambda_G G(u,v) ,
    \quad
    \lambda_F \geq 0, \; \lambda_G \geq 0, \; \lambda_F + \lambda_G = 1,
\end{equation}
i.~e., by seeking a joint strategy $(u^*,v^*)$ such that
\begin{equation}
    J(u^*,v^*) \leq J(u,v),
    \quad
    \text{for all} \; u \in \mathbb{R}^m, \; v \in \mathbb{R}^n ,
\end{equation}
a condition known as {\sl Pareto optimality}; or they can do so {\sl non-cooperatively}, by each player seeking strategies such that
\begin{subequations}\label{eW0j3a}
\begin{align}
    &
    F(u^*,v^*) \leq F(u,v^*), \quad \text{for all} \; u \in \mathbb{R}^m ,
    \\ &
    G(u^*,v^*) \leq G(u^*,v), \quad \text{for all} \; v \in \mathbb{R}^n ,
\end{align}
\end{subequations}
a condition known as {\sl Nash equilibrium}.

We note that, in both (\ref{Jw51yf}) and  (\ref{eW0j3a}), the cost of each player depends on the strategy of the competitor, which they do not control. The players seek to minimize their own costs either jointly, as in (\ref{Jw51yf}) or without regard for the cost of the competitor, as in (\ref{eW0j3a}).

An important class of non-cooperative games is that of two-player zero--sum games. These are games in which the cost of one player is the negative of the other, i.~e., one player loses what the other player gains. Under these conditions, we have
\begin{equation}
    F(u,v) = L(u,v) ,
    \quad
    G(u,v) = -L(u,v) ,
\end{equation}
for some {\sl Lagrangian} $L(u,v)$, and the Nash equilibrium conditions (\ref{eW0j3a}) become
\begin{equation}
    L(u^*,v) \leq L(u^*,v^*) \leq L(u,v^*) ,
\end{equation}
which defines a {\sl saddle-point} or {\sl $\inf$--$\sup$} problem.

Problems (\ref{Jw51yf}) and (\ref{eW0j3a}) were introduced by economists Vilfredo Pareto \cite{Pareto:1909} and John Nash \cite{Nash:1951} in seminal contributions. From a mechanics perspective, problems (\ref{Jw51yf}) and (\ref{eW0j3a}) are instances of {\sl coupled problems} with a particular variational structure.

\begin{example}[Quadratic cost]
Suppose
\begin{subequations}
\begin{align}
    &
    F(u,v) = \frac{1}{2} A u \cdot u + C v \cdot u - f\cdot u,
    \\ &
    G(u,v) = \frac{1}{2} D v \cdot v + B u \cdot v - g\cdot v,
\end{align}
\end{subequations}
where $A \in \mathbb{R}^{m\times m}$, $C \in \mathbb{R}^{m\times n}$, $B \in \mathbb{R}^{n\times m}$, $D \in \mathbb{R}^{n\times n}$, $A=A^T$, $A>0$, $D=D^T$, $D>0$, $f\in \mathbb{R}^m$, $g \in \mathbb{R}^n$, $(\cdot)$ denotes the dot product and we write $C u\cdot v = (Cu) \cdot v$, {\sl et cetera}, for short. Then, the Nash-equilibrium equations are
\begin{subequations}
\begin{align}
    &
    A u + C v = f,
    \\ &
    B u + D v = g,
\end{align}
\end{subequations}
or, in matrix form,
\begin{equation}\label{dGm517}
    \left(
        \begin{array}{c|c}
            A & C \\ \hline B & D
        \end{array}
    \right)
    \left\{
        \begin{array}{c}
            u \\ \hline v
        \end{array}
    \right\}
    =
    \left\{
        \begin{array}{c}
            f \\ \hline g
        \end{array}
    \right\} ,
\end{equation}
which is a particular type of linear coupled problem characterized by symmetric and positive-definite diagonal blocks. We also note that $C \neq B^T$ in general, with the result that there is no joint minimum principle for both players to appeal to together. Evidently, a unique Nash equilibrium $(u^*,v^*)$ exists if and only if the matrix of the system (\ref{dGm517}) is non-singular.

An alternative form of the problem is
\begin{equation}
    a(y,z) = b(z) ,
    \quad
    \text{for all} \; z \in Z,
\end{equation}
where $a : Z \times Z \to \mathbb{R}$ and $b : Z \to \mathbb{R}$, $Z = \mathbb{R}^m \times \mathbb{R}^n$, defined as
\begin{equation}
    a(y,z)
    =
    (\alpha | \beta)
    \left(
        \begin{array}{c|c}
            A & C \\ \hline B & D
        \end{array}
    \right)
    \left\{
        \begin{array}{c}
            u \\ \hline v
        \end{array}
    \right\} ,
    \quad
    b(z)
    =
    (f | g)
    \left\{
        \begin{array}{c}
            u \\ \hline v
        \end{array}
    \right\} ,
\end{equation}
with $y=(\alpha,\beta)$ and $z=(u,v)$, are non-symmetric bilinear and linear forms, respectively. Then, by the Lax-Milgram theorem \cite{Evans:1998} a unique Nash equilibrium exist if and only if
\begin{equation}
    a(z,z) \geq \lambda \, \| z \|^2 ,
\end{equation}
for some $\lambda > 0$, i.~e., if $a(\cdot,\cdot)$ is {\sl coercive}. Suppose that the cost functions of the players are separately coercive, i.~e., there are $\lambda_A >0$ and $\lambda_D >0$ such that
\begin{equation}
    Au\cdot u \geq \lambda_A \| u \|^2 ,
    \quad
    Dv\cdot v \geq \lambda_D \| v \|^2 ,
\end{equation}
for all $u \in \mathbb{R}^m$ and $v \in \mathbb{R}^n$, respectively. Suppose, in addition, that there is $0 \leq \mu < 1$ such that
\begin{equation}\label{C9O6b6}
    \big| (B+C^T)u\cdot v \big|
    \leq
    \mu \,
    \Big( Au\cdot u + Dv\cdot v \Big) ,
\end{equation}
for all $u \in \mathbb{R}^m$ and $v \in \mathbb{R}^n$. Then,
\begin{equation}
\begin{split}
    a(z,z)
    & =
    Au\cdot u + Bu\cdot v+ Cv\cdot u + Dv\cdot v
    \\ & \geq
    \Big( Au\cdot u + Dv\cdot v\Big)
    -
    \big| (B+C^T)u\cdot v \big|
    \\ & \geq
    (1-\mu) \Big( Au\cdot u + Dv\cdot v\Big)
    \geq
    (1-\mu) \min\{\lambda_A,\lambda_D\} \| z \|^2 .
\end{split}
\end{equation}
Thus, the {\sl coercivity condition} (\ref{C9O6b6}) ensures that $a(\cdot,\cdot)$ be coercive with $\lambda = (1-\mu) \min\{\lambda_A,\lambda_D\}$ and, by Lax-Milgram, it ensures the existence of a unique Nash equilibrium. \hfill$\square$
\end{example}

\section{Set--oriented formulation of problems in mechanics}
\label{M8mXkn0}

We consider finite-dimensional mechanical systems comprising $m$ components, e.~g., structural members, material points, {\sl et similia}, whose state is characterized by two work-conjugate fields $\epsilon \equiv \{\epsilon_e \in \mathbb{R}^d,\ e=1,\dots,m\}$ and $\sigma \equiv \{\sigma_e \in \mathbb{R}^d,\ e=1,\dots,m\}$. We refer to the space of pairs $Z_e = \{z_e \equiv (\epsilon_e, \sigma_e) \in \mathbb{R}^d \times \mathbb{R}^d\}$ as the {\sl local phase space} of component $e$, and $Z = Z_1 \times \cdots \times Z_m = \mathbb{R}^{N} \times \mathbb{R}^{N}$, $N = m d$, as the {\sl global phase space} of the system. We suppose that a suitable norm is defined in $Z$, e.~g.,
\begin{equation}\label{Poyet2}
    \| z \|
    =
    \Big(
    \sum_{e=1}^m
        w_e
        \| z_e \|_e^2
    \Big)^{1/2}
    =
    \Big(
    \sum_{e=1}^m
        w_e
        \big(
            {\mathbb{C}}_e \epsilon_e \cdot \epsilon_e
            +
            {\mathbb{C}}_e^{-1} \sigma_e \cdot \sigma_e
        \big)
    \Big)^{1/2} ,
\end{equation}
where $w_e > 0$ are weights and $\mathbb C_e\in \mathbb R^{d\times d}_{{\rm sym},+}$ are positive-definite symmetric matrices, $e=1,\dots,m$.

\subsection{Classical solutions}

%\marginpar{Is $n$ here the same as $N$ above?} no KW
We begin by assuming linearized kinematics and compatibility and equilibrium constraints of the general form
\begin{subequations}\label{9qHWzU}
\begin{align}
    &
    \sum_{e=1}^m w_e B_e^T \sigma_e = f ,
    \\ &
    \epsilon_e = B_e u + g_e , \quad e = 1,\dots m ,
\end{align}
\end{subequations}
where $u \in \mathbb{R}^{n}$ is the array of degrees of freedom of the system, $w_e$ are positive weights, $B_e \in \mathbb{R}^{d \times n}$ is a discrete gradient operator, $B_e^T$ is a discrete divergence operator, $f \in \mathbb{R}^n$ is a force array resulting from distributed sources and Neumann boundary conditions and the arrays $g_e \in \mathbb{R}^{d}$ follow from Dirichlet boundary conditions. In terms of global arrays,
\begin{subequations}\label{nd8X6u}
\begin{align}
    &
     w B^T\sigma = f ,
    \\ &
    \epsilon = B u + g ,
\end{align}
\end{subequations}
where we write $w = {\rm diag}(w_1,\dots,w_m)$, $B = (B_1,\dots,B_m)$, $\epsilon=(\epsilon_1,\dots,\epsilon_m)$, $\sigma=(\sigma_1,\dots,\sigma_m)$ and $g = (g_1,\dots,g_m)$.

Classically, the problem is closed by assuming a material law of the form
\begin{equation}
    \sigma_e = \hat{\sigma}_e(\epsilon_e) ,
\end{equation}
or, in terms of global arrays,
\begin{equation}
    \sigma = \hat{\sigma}(\epsilon) = \{ \hat{\sigma}_1(\epsilon_1) ,\dots, \hat{\sigma}_m(\epsilon_m) \},
\end{equation}
where $\hat{\sigma}_e(\cdot)$ are material-specific functions. Existence and uniqueness of displacement solutions then follows under suitable restrictions on $B_e$ and $\hat{\sigma}_e(\cdot)$, cf.~Prop.~\ref{8CJa9S}.

\subsection{Set--oriented reformulation}
\label{M8mXkn}

An alternative set--oriented representation of the material law $\hat{\sigma}(\epsilon)$ is to view it as a graph $D$ in phase space $Z$, or {\sl material set}. In this representation, the material law is regarded as a material-specific $N$-dimensional manifold, or graph, in the $2N$-dimensional phase space $Z$. Similarly, the constraints (\ref{9qHWzU}) are material independent and define an affine subspace $E$ of $Z$, or {\sl constraint set}. The constraint set $E$ encodes all the data of the problem, including geometry, loading and boundary conditions. From elementary linear algebra considerations, it follows that the constraint set $E$ is an affine subspace of $Z$ of dimension $N$ and co-dimension $N$ \cite{ContiHoffmannOrtiz2023}. The actual states $z = (\epsilon,\sigma)$ of the system in phase space, if they exist, lie in the intersection ${D} \, \cap \, {E}$, i.~e., are the admissible states that are consistent with the material law, or, equivalently, the material states that are compatible and in equilibrium. Evidently, classical solutions exist if the sets $D$ and $E$ have a non-empty intersection, i.~e., if they are {\sl transversal} \cite{conti:2018}.

\section{Cooperative Data--Driven games in mechanics}\label{2bzQRE}

Suppose that, as is often the case in practice, the graph $D$ of the material law is not known in its entirety, but only through an approximating sequence of data sets. For instance, the sets $D$ may consist of increasing collections of points $(\epsilon,\sigma)$ in phase space $Z$ obtained, e.~g., by experimental measurement. In general, the intersection between the admissible set ${E}$ and the material data sets ${D}$ may be empty, in which case no classical approximating solution exists in the sense of Section~\ref{M8mXkn0}. One way to circumvent this excessive rigidity of the classical paradigm is to relax the notion of 'solution' and replace intersection by a regularized optimality criterion \cite{Kirchdoerfer:2016, conti:2018}. We regard the resulting paradigm as {\sl well-posed} if the corresponding approximate Data--Driven solutions converge to the exact solution as the data sets $D$ sample the exact material-law graph $D$ with increasing fidelity. In this section, we appeal to game-theoretical concepts in order to formulate well-posed Data-Driven approaches.

\subsection{Cooperative game-theoretical reformulation}
\label{D1nIRc}

Data--Driven problems of the type proposed in \cite{Kirchdoerfer:2016, Kirchdoerfer:2017, conti:2018, Prume:2023, ContiHoffmannOrtiz2023} can be interpreted as cooperative game problems in the sense of {\sl Pareto optimality}. Thus, suppose that the material behavior is characterized by a material law with graph $D$ in phase space $Z$ and that the equilibrium, compatibility, Dirichlet and Neumann constraints are represented by an affine subspace $E$ of $Z$ of dimension $N$ and co-dimension $N$. Let
\begin{equation}
    I_D(y)
    =
    \left\{
        \begin{array}{ll}
            0, & \text{if} \; y \in D , \\
            +\infty, & \text{otherwise} ,
        \end{array}
    \right.
\end{equation}
and
\begin{equation}\label{8kFSVk}
    I_E(z)
    =
    \left\{
        \begin{array}{ll}
            0, & \text{if} \; z \in E , \\
            +\infty, & \text{otherwise} ,
        \end{array}
    \right.
\end{equation}
be the corresponding {\sl indicator functions}. Suppose, in addition, that we are given a {\sl discrepancy function} $\Phi : Z \times Z \to \mathbb{R}$ with the properties: i) $\Phi$ is convex;  ii) $\Phi$ is non-negative; and iii) $\Phi(y,z) = 0$ iff $y=z$. We may then introduce the cost functions
\begin{subequations}
\begin{align}
    &
    F(y,z) = I_D(y) + \Phi(y,z) ,
    \\ &
    G(y,z) = I_E(z) + \Phi(y,z) .
\end{align}
\end{subequations}
Evidently, for given $z$ the function $F(\cdot,z)$ requires $y$ to be in $D$, whereupon it penalizes its distance to $E$; and for given $y$ the function $G(y,\cdot)$ requires $z$ to be in $E$, whereupon it penalizes its distance to $D$.

The Data--Driven solution $(y^*,z^*)$ is now the minimizer of the function
\begin{equation}\label{dVzu1W0}
    J(y,z)
    =
    w_F \, F(y,z) + w_G \, G(y,z),
\end{equation}
with $w_F > 0$, $w_G > 0$ and $w_F + w_G = 1$, i.~e., $(y^*,z^*)$ is the solution of the cooperative game
\begin{equation}\label{OPz86E}
    J(y^*,z^*) \leq J(y,z), \quad \text{for all} \; y, z \in Z ,
\end{equation}
in the sense of Pareto optimality. In this strategy, $y^*$ is the material state that is closest to being admissible and $z^*$ is the admissible state that is closest to being material, as desired. We note that, owing to the constraints (\ref{Jw51yf}) on the Pareto weights and the invariance of the indicator functions under multiplication by positive constants, the combined cost functional (\ref{dVzu1W0}) evaluates to
\begin{equation}\label{dVzu1W}
    J(y,z)
    =
    \Phi(y,z) + I_D(y) + I_E(z),
\end{equation}
independent of the Pareto weights, and the solution $(y^*,z^*)$ is also in Nash equilibrium.

An appeal to direct methods of the calculus of variations shows that the solutions exist if, in addition to the assumed convexity properties, $\Phi$ satisfies a growth condition of the form \cite{Kirchdoerfer:2016, conti:2018}: $\Phi(y,z) \to +\infty$ if $\|y\|+\|z\| \to +\infty$ with $y \in D$ and $z \in E$. This condition generalizes the intersection notion of transversality between the sets $D$ and $E$ to sets that may be discontinuous, e.~g., point sets, and may lack a proper intersection. Evidently, every classical solution in the intersection $D \cap E$ is also a Data--Driven solution.

% \marginpar{ADD: and the sets $D$ and $E$ are closed. SC.} please do

\subsection{The long and short games}

We can reduce the problem to the determination of admissible states $z$ by eliminating out the material state variable $y$ in the cooperative game (\ref{OPz86E}). By analogy to the game of golf, we may think of point data sets $D$ as a collection of holes in a links. The player $z$ plays the {\sl long game}, or game of approach to the holes. For a given outcome $z$ of an approach shot, the player $y$ then plays the {\sl short game} of putting the ball as close as possible to the nearest hole.

The {\sl short game} is, therefore, defined by the functional
\begin{equation}\label{DkNM6N}
    H(z)
    =
    \inf_{y\in Z} J(y,z)
    =
    I_E(z) + \varphi_D(z) ,
\end{equation}
where
\begin{equation}\label{eq:phiDinfPhi}
    \varphi_D(z) = \inf_{y\in D} \Phi(y,z)
\end{equation}
measures the deviation of $z$ from the data set $D$. The remaining {\sl long game} is
\begin{equation}\label{71VtKI}
    H(z^*) \leq H(z), \quad \text{for all} \; z \in Z ,
\end{equation}
which endeavors to minimize the deviation $\varphi_D(z)$ from $D$ among all admissible state $z \in E$. We note that the sequence of long and short games is equivalent to (\ref{OPz86E}), since the game is cooperative.

\subsection{Minimum--distance Data--Driven game}\label{y60Z7n}
The Data--Driven paradigm formulated in \cite{Kirchdoerfer:2016, conti:2018} may be rephrased as a cooperative game by setting
\begin{equation}\label{2Cu1UG}
    \Phi(y,z) = \| y - z \|^2 .
\end{equation}
Evidently, $\Phi$ is convex, non-negative and $\Phi(y,z)=0$ if $y=z$, as required. In addition, $\Phi$ satisfies the required growth condition iff $D$ and $E$ are transverse. In addition,
\begin{equation}\label{eq:phiDminDist}
    \varphi_D(z)
    =
    \inf_{y\in D} \|y-z\|^2
    =
    {\rm dist}^2(z,D) ,
\end{equation}
where ${\rm dist}(z,D)$ denotes the distance between $z$ and the set $D$. The resulting long game is, therefore
\begin{equation}
    z^* \in {\rm argmin} \{ {\rm dist}^2(z,D) \, : \, z \in E \} ,
\end{equation}
i.~e., finding the admissible state that is closest to the material set. If, for instance, the distance is computed from the norm (\ref{Poyet2}) and the material is elastic with
\begin{equation}
    D = \{ y = (\epsilon,\sigma) \, : \, \sigma = \mathbb{C} \epsilon \} ,
\end{equation}
then a straightforward calculation gives \cite{Kirchdoerfer:2016}
\begin{equation}\label{6ov8Ze}
    \varphi_D(z)
    =
    \frac{1}{2}
    \sum_{e=1}^m
        w_e
        \mathbb{C}_e^{-1}
        (\sigma_e-\mathbb{C}_e \epsilon_e)
        \cdot
        (\sigma_e-\mathbb{C}_e \epsilon_e)
    \equiv
    \frac{1}{2} \|\sigma - \mathbb{C} \epsilon\|^2 .
\end{equation}
Evidently, $H$ vanishes on $D$ and grows away from it, as required. Thus, $H(z)$ measures the distance from $z$ to the material data set $D$ and the long-game endeavors to minimize that distance for all admissible states $z \in E$.

\subsection{Convex Data--Driven games}\label{u9t9PX}
The preceding example is a particular case of a more general class of convex games. Suppose that the material behavior is characterized by a convex strain energy function $W(\epsilon)$ with well defined dual stress energy function $W^*(\sigma)$. Then, the material set is the graph
\begin{equation}\label{TF00eF}
\begin{split}
    D
    & =
    \{
        y
        =
        (\epsilon,\sigma) \in Z
        \, : \,
        \sigma = DW(\epsilon)
    \}
    \\ & =
    \{
        y
        =
        (\epsilon,\sigma) \in Z
        \, : \,
        \epsilon = DW^*(\sigma)
    \} .
\end{split}
\end{equation}
By convexity and the Fenchel-Taylor theorem \cite{Rockafellar:1997}, we have
\begin{equation}\label{kHp31i}
    f(\epsilon,\sigma)
    \equiv
    W(\epsilon)
    +
    W^*(\sigma)
    -
    \sigma \cdot \epsilon
    \geq 0 ,
\end{equation}
and
\begin{equation}
    f(\epsilon,\sigma)
    =
    0
    \; \text{iff} \;
    (\epsilon,\sigma) \in D .
\end{equation}
Hence, $f(\epsilon,\sigma)$ quantifies the deviation of $(\epsilon,\sigma)$ from $D$. Consider the discrepancy function
\begin{equation}
    \Phi(y,z)
    =
    f(\epsilon-\alpha,\sigma-\beta)
    =
    W(\epsilon-\alpha)
    +
    W^*(\sigma-\beta)
    -
    (\sigma-\beta) \cdot (\epsilon-\alpha),
\end{equation}
with $y=(\alpha,\beta)$ and $z=(\epsilon,\sigma)$. An appeal to duality and the Fenchel-Taylor theorem gives, with $w = (\xi,\eta)$,
\begin{equation}\nonumber
\begin{split}
    &
    \varphi_D(\epsilon,\sigma)
    =
    \inf_{y \in D}
    \Phi(y,z)
    =
    \inf_{y \in D}
    f(\epsilon-\alpha,\sigma-\beta)
    = \\ &
    \lim_{\lambda\to+\infty}
    \inf_{y \in Z}
    \Big(
        f(\epsilon-\alpha,\sigma-\beta)
        +
        \lambda \varphi_D(\alpha,\beta)
    \Big)
    = \\ &
    \lim_{\lambda\to+\infty}
    \inf_{y \in Z}
    \sup_{w\in Z}
    \Big(
        \eta\cdot(\epsilon-\alpha) - W^*(\eta)
        +
        \xi\cdot(\sigma-\beta) - W(\xi)
        - \\ & \qquad\qquad\qquad\qquad\qquad\qquad\qquad
        (\sigma-\beta) \cdot (\epsilon-\alpha)
        +
        \lambda \, \varphi_D(\alpha,\beta)
    \Big)
    =
\end{split}
\end{equation}
\begin{equation}
\begin{split}
    &
    \lim_{\lambda\to+\infty}
    \sup_{w\in Z}
    \inf_{y \in Z}
    \Big\{
        \eta \cdot \epsilon
        -
        W^*(\eta)
        +
        \sigma \cdot \xi
        -
        W(\xi)
        -
        \sigma \cdot \epsilon
        + \\ & \qquad
        \lambda \,
        \Big(
            \lambda^{-1}(\sigma - \eta) \cdot \alpha
            +
            \lambda^{-1}(\epsilon - \xi) \cdot \beta
            +
            W(\alpha) + W^*(\beta) - \beta\cdot\alpha
        \Big)
    \Big\}
    =
\end{split}
\end{equation}
\begin{equation}\nonumber
\begin{split}
    &
    \lim_{\lambda\to+\infty}
    \sup_{w\in Z}
    \Big\{
        \eta \cdot \epsilon
        -
        W^*(\eta)
        +
        \sigma \cdot \xi
        -
        W(\xi)
        -
        \sigma \cdot \epsilon
        + \\ & \qquad\qquad\qquad\qquad\qquad\qquad\qquad
        \lambda \, \varphi_D^*\big(\lambda^{-1}(\eta- \sigma), \lambda^{-1}(\xi - \epsilon) \big)
    \Big\}
    = \\ &
    \sup_{w\in Z}
    \Big\{
        \eta \cdot \epsilon
        -
        W^*(\eta)
        +
        \sigma \cdot \xi
        -
        W(\xi)
        -
        \sigma \cdot \epsilon
        + \\ & \qquad\qquad\qquad\qquad\qquad
        D_\eta \varphi_D^*\big(0,0) (\eta- \sigma)
        +
        D_\xi \varphi_D^*(0,0) (\xi - \epsilon)
    \Big\}
    = \\ &
    \sup_{w\in Z}
    \Big\{
        \eta \cdot \epsilon  + \sigma \cdot \xi
        -
        W^*(\eta)
        -
        W(\xi)
        -
        \sigma \cdot \epsilon
    \Big\}
    = \\ &
    W(\epsilon)
    +
    W^*(\sigma)
    -
    \sigma \cdot \epsilon
    =
    \varphi_D(\epsilon,\sigma) .
\end{split}
\end{equation}
Thus, it follows that the corresponding long game (\ref{71VtKI}) endeavors to minimize the deviation from $D$, as measured by $\varphi_D$, among all admissible states $z \in E$.

\section{Non-cooperative Data--Driven games}
\label{7ATkpe}

We now change moods and propose a new class of non-cooperative Data--Driven games as an alternative to the cooperative Data--Driven games described in the foregoing. We specifically focus on the long game (\ref{71VtKI}). We recall that this long game results from explicitly or implicitly minimizing out the material state $y$ over $D$, so that the resulting cost function $H(z)$ represents the deviation of an admissible state $z\in E$ from the material set $D$. The long game that remains is then to minimize such deviation among all admissible states.

\subsection{Data--Driven game with adversarial stresses and strains}

In the cooperative mood adopted in Section~\ref{2bzQRE}, the minimization of the cost $H(z)$, with $z=(\epsilon,\sigma)$, is pursued jointly in $\epsilon$ and $\sigma$. The compatibility constraint can be enforced constructively by setting $\epsilon = B u$, with $u$ a displacement field. In addition, the equilibrium constraint on $\sigma$ can be enforced by means of Lagrange multipliers. This implementation of the game results in two standard linear problems for the displacements $u$ and the Lagrange multipliers $v$, regardless of the nature of the material data \cite{Kirchdoerfer:2016}.

As an alternative, here we explore a reformulation of Data-Driven game in which the players $(\epsilon,\sigma)$ are adversarial: The objective of the stress player $\sigma$ is to minimize the discrepancy to the data set $D$ for fixed $\epsilon$; the objective of the strain player $\epsilon$ is to ensure the admissibility of the state $z$ for fixed $\sigma$. The corresponding non-cooperative game is
\begin{subequations}\label{eq:phiDphiE}
\begin{align}
    &
    \varphi_D(\epsilon^*,\sigma^*)
    \leq
    \varphi_D(\epsilon^*,\sigma),
    \quad \text{for all} \; \sigma \in \mathbb{R}^N ,
    \\ &
    \varphi_E(\epsilon^*,\sigma^*)
    \leq
    \varphi_E(\epsilon,\sigma^*),
    \quad \text{for all} \; \epsilon \in \mathbb{R}^N ,
\end{align}
\end{subequations}
where the function $\varphi_E$ penalizes deviations from $E$. For instance, proceeding as in the cooperative game (\ref{OPz86E}), we have $\varphi_E(\epsilon,\sigma) = I_E(z)$ with $z=(\epsilon,\sigma)$, or explicitly in stress-strain variables,
\begin{equation}\label{5nBaMU}
    \varphi_E(\epsilon,\sigma)
    =
    \left\{
        \begin{array}{ll}
            0, & \text{if } \epsilon = B u,\quad w B^T \sigma = f ,\\
            +\infty, & \text{otherwise} .
        \end{array}
    \right.
\end{equation}
The corresponding Nash-equilibrium conditions are
\begin{subequations}\label{NP77aA}
\begin{align}
    & \label{D6kzmk}
    \frac{\partial\varphi_D}{\partial\sigma} (\epsilon,\sigma) = 0 ,
    \\ & \label{Gy5P7U}
    \epsilon = B u , \qquad w B^T \sigma = f ,
\end{align}
\end{subequations}
and the optimal strategies $(\epsilon^*,\sigma^*)$ are the solutions of these equations, if any. Indeed, for $\varphi_E(\epsilon,\sigma)$ as in (\ref{5nBaMU}) to be finite we must have $z \in E$, whence (\ref{Gy5P7U}) follows.

An alternative means of deriving (\ref{Gy5P7U}) is by regularization of the indicator function and a subsequent passage to the limit. Thus, we may define the sequence of regularized functionals
\begin{equation}\label{eq:nonCoop:phiC:regularisiert:Stab}
    \varphi_{E,\delta}(\epsilon,\sigma)
    =
    \left\{
        \begin{array}{ll}
            \frac{\delta}{2}
            (B^Tw\mathbb{C}B) u \cdot u
            +
            w\sigma\cdot B u
            -
            f \cdot u ,
            & \text{if } \epsilon = B u, \\
            +\infty, & \text{otherwise} ,
        \end{array}
    \right.
\end{equation}
with $\delta \downarrow 0$ and the stiffness matrix $(B^Tw\mathbb{C}B)$ introduced for dimensional consistency. We note that compatibility can be enforced constructively, leading to the reduced functional
\begin{equation}
    \varphi_{E,\delta}(u,\sigma)
    =
    \frac{\delta}{2}
    (B^Tw\mathbb{C}B) u \cdot u
    +
    w\sigma\cdot B u - f \cdot u .
\end{equation}
The corresponding Nash equilibrium condition is
\begin{equation}
    \frac{\partial\varphi_{E,\delta}}{\partial u}(u,\sigma)
    =
    \delta
    (B^Tw\mathbb{C}B) u
    +
    w B^T\sigma
    -
    f
    =
    0 ,
\end{equation}
and (\ref{Gy5P7U}) is recovered by passing to the limit $\delta \to 0$.

\begin{example}[Distance deviation function]
With the deviation function (\ref{6ov8Ze}), the Nash equilibrium condition (\ref{D6kzmk}) specializes to
\begin{equation}
    2 w_e
    \mathbb{C}_e^{-1}
    (\sigma_e-\mathbb{C}_e \epsilon_e)
    =
    0 ,
    \quad e=1,\dots,m ,
\end{equation}
and the effective constitutive relation is given locally by Hooke's law
\begin{equation}
    \hat{\sigma}_e(\epsilon_e) = \mathbb{C}_e \epsilon_e ,
    \quad e=1,\dots,m ,
\end{equation}
as expected. In addition, we note that the Nash equilibrium conditions (\ref{NP77aA}) can be expressed jointly as
\begin{equation}
    \left(
        \begin{array}{c|c}
            0 & w B^T \\ \hline - w B & w \mathbb{C}^{-1}
        \end{array}
    \right)
    \left\{
        \begin{array}{c}
            u \\ \hline \sigma
        \end{array}
    \right\}
    =
    \left\{
        \begin{array}{c}
            f \\ \hline 0
        \end{array}
    \right\} ,
\end{equation}
which is of the form (\ref{dGm517}). These equations follow jointly as Euler-Lagrange equations of the Lagrangian
\begin{equation}
    L(u,\sigma)
    =
    w \sigma \cdot B u
    -
    f \cdot u
    -
    \frac{1}{2} w \mathbb{C}^{-1} \sigma \cdot \sigma ,
\end{equation}
corresponding to the $\inf$--$\sup$ problem
\begin{equation}
    \inf_u \sup_\sigma L(u,\sigma)
    =
    \inf_u \Big( \frac{1}{2} w\mathbb{C} Bu \cdot Bu - f \cdot u \Big) ,
\end{equation}
which identifies the effective problem as a zero--sum game.
\hfill$\square$
\end{example}

\begin{example}[Convex deviation function]
Suppose now that the deviation function is given by (\ref{kHp31i}) as in Example~\ref{u9t9PX}, the Nash equilibrium condition (\ref{D6kzmk}) specializes to
\begin{equation}
    2 w_e
    (DW^*_e(\sigma_e)-\epsilon_e)
    =
    0 ,
    \quad e=1,\dots,m ,
\end{equation}
and the effective constitutive relation is given locally by
\begin{equation}
    \hat{\sigma}_e(\epsilon_e) = DW_e(\epsilon_e) ,
    \quad e=1,\dots,m ,
\end{equation}
as expected. As in the preceding problem, the Nash equilibrium conditions (\ref{NP77aA}) follow jointly from the Lagrangian
\begin{equation}\label{0hVPq1}
    L(u,\sigma)
    =
    \sum_{e=1}^m  \left(w_e \sigma_e \cdot B_e u
    -
    W_e^*(\sigma_e) \right)
    -
    f \cdot u
 \end{equation}
corresponding to the $\inf$--$\sup$ problem
\begin{equation}
    \inf_u \sup_\sigma L(u,\sigma)
    =
    \inf_u \Big(
    \sum_{e=1}^m w_e W_e(B_e u) - f \cdot u \Big) ,
\end{equation}
which again identifies the effective problem as a zero--sum game.
\hfill$\square$
\end{example}

\subsection{Existence}\label{jtA0LI}

Solving (\ref{D6kzmk}) for the stresses, we obtain a relation of the form
\begin{equation}\label{nFb0r5}
    \sigma = \hat{\sigma}(\epsilon) ,
\end{equation}
which defines an {\sl effective}, or {\sl learned}, constitutive law. Conditions on $\varphi_D(\epsilon,\sigma)$ for the function $\hat{\sigma}(\epsilon)$ to exist and be continuous are given by the implicit function theorem \cite{Rudin:1976}. Eq.~(\ref{Gy5P7U}) then requires
%\marginpar{I think we should write this as $B^T w \hat{\sigma}(B u) = f$, because $w$ has $m$ components, and can be extended immediately to a linear operator on $\R^{md}$, whereas $B^T\sigma$ has $n$ components, I do no see how $w$ can act on this Besides, this is what \eqref{EQkN5s} needs. SC.}
%
\begin{equation}\label{wpL8eA}
    w B^T \hat{\sigma}(B u) = f  ,
\end{equation}
which defines an effective displacement problem. A first fundamental question is whether the problem (\ref{wpL8eA}) is well-posed in the sense of existence and uniqueness of solutions.

We note that problem (\ref{wpL8eA}) is not required to have a variational structure, e.~g., to derive from a minimum energy principle. Nevertheless, existence of solutions follows if the local material laws $\hat{\sigma}_e(\epsilon_e)$ are {\sl coercive}, in the sense of material stability, and the structure is likewise stable, in the sense of not allowing zero-strain mechanisms. Uniqueness of the solution follows if, in addition, the material law is {\sl strictly monotone}.

A general framework for existence and convergence of approximations is set forth by the following propositions, which are adapted from \cite{Evans:1998}, \S9.1,  to the present setting.

\begin{lem}[Zeros of a vector field] \label{T9WhYq} Suppose that a continuous function $v : \mathbb{R}^n \to \mathbb{R}^n$ satisfies
\begin{equation}
    v(u) \cdot u \geq 0,
    \quad
    \text{if} \;\; \|u\| = r,
\end{equation}
for some $r > 0$. Then, there exists a point $u^* \in \mathbb{R}^n$, $\|u^*\| \leq r$, such that $v(u^*) = 0$.
\end{lem}

The proof of the lemma is based on Brouwer's fixed point theorem and can be found in \cite{Evans:1998}, \S9.1.

\begin{prop}[Existence of solutions]\label{8CJa9S}
Let $w_e > 0$ and $\hat{\sigma}_e : \mathbb{R}^d \to \mathbb{R}^d$ continuous functions, $e=1,\dots,m$. Let $f \in \mathbb{R}^n$ and $B : \mathbb{R}^n \to \mathbb{R}^{m d}$.
Suppose that:
\begin{itemize}
    \item[i)]{Material stability}. There are $a > 0$, $b \geq 0$ such that, for all $\epsilon_e \in \mathbb{R}^d$,
\begin{equation}\label{eqsigmastable}
    \hat{\sigma}_e(\epsilon_e) \cdot \epsilon_e
    \geq
    a \| \epsilon_e \|_e^2 - b .
\end{equation}
    \item[ii)]{Structural stability}. There is $c > 0$ such that, for all $u \in \mathbb{R}^n$,
\begin{equation}\label{eqBstructurstable}
    \| u \|^2 \leq c \| Bu \|^2 .
\end{equation}
\end{itemize}
Then, problem (\ref{wpL8eA}) has a solution $u^* \in \mathbb{R}^n$ that satisfies the bound
\begin{equation}\label{kO17na}
    \| u^* \|^2
    \leq
    \Big(\frac{a}{c} - \epsilon \Big)^{-1}
    \Big( \frac{1}{4\epsilon} \| f \|^2 + b V \Big),
\end{equation}
for all $\epsilon < a/c $ and $V= \sum_{e=1}^m w_e$. Suppose, in addition:
\begin{itemize}
\item[iii)]{Strict monotonicity}. There is $\theta > 0$ such that, for all $\epsilon_e'$, $\epsilon_e'' \in \mathbb{R}^d$,
\begin{equation}\label{eqsigmastrictmonotone}
    \big( \hat{\sigma}_e(\epsilon_e') - \hat{\sigma}_e(\epsilon_e'') \big)
    \cdot
    ( \epsilon_e' - \epsilon_e'' )
    \geq
    \theta \, \| \epsilon_e' - \epsilon_e'' \|_e^2 .
\end{equation}
\end{itemize}
Then, the solution is unique.
\end{prop}

We note that condition (i) stipulates coercivity of the material laws, whereas condition (ii) requires the absence of mechanisms, i.~e., displacements that occur at zero strain. The proof of the proposition is illustrative of the roles played by material and structural stability and is therefore included next in full.

\begin{proof}
Let $(e_1,\dots,e_n)$ be the standard orthonormal basis of $\mathbb{R}^n$. Define the continuous function $v : \mathbb{R}^n \to \mathbb{R}^n$ by setting
\begin{equation}\label{EQkN5s}
    v_i(u)
    =
    w \hat{\sigma}(Bu) \cdot Be_i - f \cdot e_i ,
\end{equation}
for every $u \in \mathbb{R}^n$. From (i), we find
\begin{equation}
    v(u) \cdot u
    =
    w \hat{\sigma}(Bu) \cdot Bu - f \cdot u
    \geq
    a \| Bu \|^2 - b V - f \cdot u ,
\end{equation}
and by (ii),
\begin{equation}\label{ke1vBR}
    v(u) \cdot u + b V + f \cdot u
    \geq
    \frac{a}{c} \| u \|^2 .
\end{equation}
For every $\epsilon > 0$, we have the estimate
\begin{equation}
    | f \cdot u |
    \leq
    \epsilon \| u \|^2 + \frac{1}{4\epsilon} \| f \|^2 ,
\end{equation}
which, inserted into (\ref{ke1vBR}), gives
\begin{equation}
    v(u) \cdot u + b V + \epsilon \| u \|^2 + \frac{1}{4\epsilon} \| f \|^2
    \geq
    \frac{a}{c} \| u \|^2 .
\end{equation}
Rearranging terms,
\begin{equation}\label{EKk6Xo}
    v(u) \cdot u
    \geq
    \Big(\frac{a}{c} - \epsilon \Big)  \| u \|^2 - b V - \frac{1}{4\epsilon} \| f \|^2 ,
\end{equation}
Hence, $v(u) \cdot u \geq 0$ if $\| u \| = r$, provided that we select $\epsilon$ small enough and $r > 0$ sufficiently large. By Lemma~\ref{T9WhYq}, there is $u^* \in \mathbb{R}^n$ such that $v(u^*) = 0$, i.~e., $u^*$ is a solution of (\ref{wpL8eA}). In addition, from (\ref{EKk6Xo}) the solution satisfies the bound
\begin{equation}
    \Big(\frac{a}{c} - \epsilon \Big)  \| u^* \|^2
    \leq
    \frac{1}{4\epsilon} \| f \|^2 + b V ,
\end{equation}
which implies (\ref{kO17na}), as advertised. Assume, in addition, that (iii) holds. Suppose that there are two solutions $u'$, $u'' \in \mathbb{R}^n$. Then, for all $v \in \mathbb{R}^n$,
\begin{equation}
    w \hat{\sigma}(B u') \cdot Bv
    =
    w \hat{\sigma}(B u'') \cdot Bv
    =
    f \cdot v ,
\end{equation}
whence
\begin{equation}
    w
    \big(
        \hat{\sigma}(B u')
        -
        \hat{\sigma}(B u'')
    \big)
    \cdot Bv
    =
    0 .
\end{equation}
Setting $v = u' - u''$ and using (ii) and (iii)
\begin{equation}
    0
    =
    w
    \big(
        \hat{\sigma}(B u')
        -
        \hat{\sigma}(B u'')
    \big)
    \cdot
    B(u' - u'')
    \geq
    \theta \|Bu' - Bu''\|^2
    \geq
    \frac{\theta}{c^2} \| u' - u''\|^2 ,
\end{equation}
which requires $\| u' - u''\| = 0$ and $u' = u''$, as advertised.
\end{proof}

\subsection{Approximation}\label{kEV3cF}

Suppose now that the local material laws $\hat{\sigma}_e(\epsilon_e)$ in problem (\ref{wpL8eA}) are not known exactly but only approximately through a convergent sequence of approximate local material laws $\hat{\sigma}_{e,h} (\epsilon_e)$. The approximate local material laws set forth a sequence of approximating problems
%\marginpar{same as above: $B^Tw\sigma$, not $wB^T\sigma$. SC.}
\begin{equation}\label{FrW19q}
    w B^T \hat{\sigma}_h(B u) = f  ,
\end{equation}
where we write $\hat{\sigma}_h(\epsilon) = (\hat{\sigma}_{e,h}(\epsilon_e))_{e=1}^m$. We wish to ascertain conditions under which the solutions $u^*_h$ of problems (\ref{FrW19q}) converge to the solution $u$ of the limiting problem (\ref{wpL8eA}).

The following proposition sets forth conditions under which approximation of the local material laws results in convergent approximate solutions.

\begin{prop}[Approximation]\label{abAx6h}
Let $w_e > 0$ and $\hat{\sigma}_{e,h} : \mathbb{R}^d \to \mathbb{R}^d$ continuous functions, $e=1,\dots,m$, $h \in \mathbb{N}$. Let $f \in \mathbb{R}^n$ and $B : \mathbb{R}^n \to \mathbb{R}^{m d}$. Suppose that:
\begin{itemize}
\item[i)] The sequence $(\hat{\sigma}_{e,h})$ is uniformly stable in the sense of Prop~\ref{8CJa9S}(i), i.~e., there are $a > 0$, $b \geq 0$ independent of $h$ such that, for all $\epsilon_e \in \mathbb{R}^d$,
\begin{equation}
    \hat{\sigma}_{e,h}(\epsilon_e) \cdot \epsilon_e
    \geq
    a \| \epsilon_e \|_e^2 - b .
\end{equation}
\item[ii)] $B$ is stable in the sense of Prop~\ref{8CJa9S}(ii).
\item[iii)] The local material laws $\hat{\sigma}_{e,h}$ are continuous and converge to local material laws $\hat{\sigma}_{e}$ uniformly on compact sets.
\end{itemize}
Then, the solutions $u^*_h$ of problem (\ref{FrW19q}) converge, up to subsequences, to a solution of the limiting problem (\ref{wpL8eA}).
\end{prop}

We recall that, by the Arzel\`a-Ascoli theorem, (iii) is ensured if the sequences $(\hat{\sigma}_{e,h})$ are uniformly bounded and equi-continuous, in which case the limiting laws $\hat{\sigma}_{e}$ are also continuous. If the approximate local material laws $(\hat{\sigma}_{e,h})$ are differentiable, then equi-continuity is ensured if the derivatives of $(\hat{\sigma}_{e,h})$ are uniformly bounded.

\begin{proof}
By Prop.~\ref{8CJa9S}, the approximating problems (\ref{FrW19q}) have solutions $u^*_h$, not necessarily unique, satisfying $v_h(u^*_h)=0$, with $v_h$ defined from $\hat{\sigma}_{e,h}$ as in (\ref{EQkN5s}), and uniformly bounded as in (\ref{kO17na}). By this latter property, there is a subsequence (not renamed) that converges to some $u^* \in \mathbb{R}^n$. By uniform convergence,
\begin{equation}
    v(u^*) = \lim_{h\to\infty} v_h(u^*_h) = 0 ,
\end{equation}
with $v$ defined from $\hat{\sigma}_e$ as in (\ref{EQkN5s}), which shows that $u^*$ is in fact a solution of the limiting problem (\ref{wpL8eA}).
\end{proof}

The main conclusion afforded by the preceding propositions is that {\sl existence and convergence can be elucidated, for any stable structure and applied loading, locally at the material point level}, simply by examining the properties and convergence of the local material laws. We also remark that, by applying Cauchy's test, convergence of the approximating sequences of local material laws can be established without knowing the limiting local material laws explicitly. Such inferred convergence then guarantees that the approximating solutions in turn converge to the solutions of the (unknown) limiting problem.

Examples of approximation and convergence are presented in Section~\ref{cbb7Jh} for the important case of approximation by empirical point data.

\section{Approximation by empirical point data}
\label{u5RpCF}

An important property of the set--oriented Data--Driven games defined in the foregoing is that they remain applicable when the exact material data set, e.~g., the graph of the underlying constitutive relation, is replaced by an approximation $D$ thereof in the form of a point set, e.~g., measured empirically or computed from micromechanics. A central question is then to ascertain conditions ensuring the convergence of the Data-Driven solutions to the solution of the (unknown) limiting material law. In this section, we specifically consider two different scenarios: i) {\sl Uniformly convergent data}, in which the sampling error decreases as data is added to the material-data set in a uniform manner controlled by strict upper bounds, and ii) {\sl noisy data with outliers}, in which the data concentrates around the limiting material law in a weak or average sense that allows for the presence of outliers.

\subsection{Distance-based Data-Driven game}\label{M921rQ}

Suppose that the material behavior is local, i.~e., $D=D_{1} \times \cdots \times D_{m}$, with local material data sets $D_{e} = \{y_{e,i} \in Z_e \, : \, i=1,\dots,N_{e}\}$. Then,
\begin{equation}\label{eq:nonCoop:phiD:ele}
    \varphi_{D}(z)
    =
    \sum_{e=1}^m \varphi_{D,e}(z_e)
\end{equation}
with
\begin{equation}\label{z6bESY}
    \varphi_{D,e}(z_e)
    =
    \inf\{\Phi_e(y_{e},z_e) \, : \, y_{e} \in D_{e}\} ,
\end{equation}
and the Nash equilibrium condition for the stress player reduces to
\begin{equation}\label{kP26d1}
    \hat{\sigma}_e(\epsilon_e) = \eta_e,
    \quad \text{with} \;
    (\xi_e, \eta_e) \in
    {\rm argmin} \, \{\Phi_e(y_e, z_e) \,:\, y_e \in D_e \} .
\end{equation}
In the special case of a distance-based discrepancy function, (\ref{kP26d1}) reduces to
\begin{equation}\label{kP26d1b}
    \hat{\sigma}_e(\epsilon_e) = \eta_e,
    \quad \text{with} \;
    (\xi_e, \eta_e) \in D_{e}
    \; \text{and} \; \xi_e \; \text{closest to} \; \epsilon_e  .
\end{equation}
Thus, the effective constitutive law $\hat{\sigma}_e(\epsilon_e)$ {\sl learned} by the Data--Driven game consists of looking in the data set $D_{e}$ for the point $(\xi_e, \eta_e)$ such that $\xi_e$ is closest to $\epsilon_e$ and taking the corresponding stress $\eta_e$ as the value of $\hat{\sigma}_e(\epsilon_e)$.

\subsection{Max-ent regularization, clustering and smoothing}
\label{9kYy4s}

Remarkably, the game just defined for point-data results in a 'learned' stress-strain relation (\ref{kP26d1}), or (\ref{kP26d1b}), that is discontinuous and multiply-valued, e.~g., if the query strain is equidistant from more than one strain in the data set. This lack of continuity places the effective stress-strain relation (\ref{kP26d1}) out of scope of Prop.~\ref{abAx6h}, which requires continuity, and necessitates the use of specialized solvers such as dynamic relaxation, cf.~Sections~\ref{Nedi3P} and \ref{RP4nsY}, or smoothing over a carefully chosen strain range, cf.~Section~\ref{secuniformapprox}. An additional difficulty arises from outliers in the data, i.~e., points that deviate markedly from the general trend in the data. Indeed, uncontrolled outliers can bias the Data-Driven solution and forestall convergence.

These difficulties can be overcome by means of a {\sl max-ent} regularization of the problem \cite{Kirchdoerfer:2017}. As before, we suppose that the data set is a collection of local point--data sets, $D=D_{1} \times \cdots \times D_{m}$, with $D_{e} = \{y_{e,i} \in Z_e \, : \, i=1,\dots,N_{e}\}$. We begin by noting that the local deviation function (\ref{z6bESY}) can equivalently be expressed as
\begin{equation}\label{4vCTOE}
\begin{split}
    &
    \varphi_{D,e}(z_e)
    = \\ &
    \inf
    \Big\{
        \sum_{i=1}^{N_e} p_{e,i} \Phi_e(y_{e,i},z_e)
        \, : \,
        \sum_{i=1}^{N_e} p_{e,i} = 1 ,
        \; p_{e,i} \geq 0, \; i = 1,\dots,N_e
    \Big\} .
\end{split}
\end{equation}
The weights $\{p_{e,i}\}_{i=1}^{N_e}$ quantify how well a point $z_e$ of local phase space $Z_e$ is represented by a point $y_{e,i}$ in the corresponding local material data set $D_e$, or, conversely, the relevance of a point $y_{e,i}$ in the local material data set to a given point $z_e$ in local phase space. Evidently, the minimizing weights concentrate on the material points that minimize their deviation from $z_e$, whereupon (\ref{z6bESY}) is recovered.

The measure-theoretical representation (\ref{4vCTOE}) can now be conveniently regularized through the addition of a small entropy term, with the result
\begin{equation}\label{O8yq5k}
\begin{split}
    &
    \varphi_{D,e}(z_e;\beta_e)
    = \\ &
    \inf
    \Big\{
        \sum_{i=1}^{N_e}
        p_{e,i}
        \Big(
            \Phi_e(y_{e,i},z_e)
            +
            \beta_e^{-1}
            \log p_{e,i}
        \Big)
        \, : \,
        \sum_{i=1}^{N_e} p_{e,i} = 1
    \Big\} ,
\end{split}
\end{equation}
with $\beta_e\to +\infty$. The regularization term can indeed be interpreted as the negative of Shannon's information-theoretical entropy \cite{Shannon:1948}, hence the term {\sl maximum-entropy}, or 'max-ent', regularization. The entropy term ensures that the distribution of weights is as unbiased as possible. The regularized functional thus has the structure of a free energy, with the first term representing the internal energy of the system. In this interpretation, the problem (\ref{O8yq5k}) expresses the principle of minimum free energy, $\beta_e$ plays the role of a local {\sl reciprocal temperature} and the limit $\beta_e\to+\infty$ is the corresponding zero-temperature, or {\sl athermal}, limit.

We additionally recall that the entropy term measures the {\sl Kullback-Leibler discrepancy} \cite{Kullback:1951, Kullback:1959} between the empirical and weighted measures
\begin{equation}\label{wg58FX}
    \mu_{D,e} = \sum_{i=1}^{N_e} \delta_{y_{e,i}}
    \quad \text{and} \quad
    \nu_{D,e} = \sum_{i=1}^{N_e} p_{e,i} \delta_{y_{e,i}} .
\end{equation}
By this interpretation, the entropy term in (\ref{O8yq5k}) aims to minimize the discrepancy between the two measures (\ref{wg58FX}), i.~e., to render the weights $p_{e,i}$ as uniform as possible. The local deviation function $\Phi_e(y_{e,i},z_e)$ in turn ensures that points $y_{e,i}$ in the local data set $D_e$ that are distant from the query point $z_e$ are accorded less weight than nearby points. Evidently, the optimal weights now follow as the result of a competition between the deviation function and entropy, with $\beta_e^{-1}$ playing the role of a local Pareto weight.

Conveniently, the minimizing weights and the minimum of the functional (\ref{O8yq5k}) follow explicitly as
\begin{subequations}
\begin{align}
    & \label{xKSL3U}
    p_{e,i}^*(z_e;\beta_e)
    =
    \frac{1}{Z_e(z_e;\beta_e)} {\rm e}^{- \beta_e \Phi_e(y_{e,i},z_e)},
    \\ &
    Z_e(z_e;\beta_e)
    =
    \sum_{i=1}^{N_e}
    {\rm e}^{- \beta_e \Phi_e(y_{e,i},z_e)} ,
    \\ &
    \varphi_{D,e}^*(z_e;\beta_e)
    =
    -
    \beta_e^{-1} \log Z_e(z_e;\beta_e) ,
\end{align}
\end{subequations}
which are a discrete Gibbs distribution, partition function and equilibrium free energy, respectively. In addition, the Nash equilibrium condition (\ref{D6kzmk}) corresponds to minimizing the local {\sl free energy} to with respect to the stress $\sigma_e$ at fixed strain $\epsilon_e$. Appealing to the optimality of the weights, we compute
\begin{equation}\label{0VlMsK}
    \frac{\partial\varphi_{D,e}^*}{\partial\sigma_e}(z_e;\beta_e)
    =
    \sum_{i=1}^{N_e}
    p_{e,i}^*(z_e;\beta_e)
    \frac{\partial\Phi_{D,e}}{\partial\sigma_e}(y_{e,i},z_e)
    =
    0 ,
\end{equation}
where we write $z_e = (\epsilon_e,\sigma_e)$.

It follows from (\ref{xKSL3U}) that points in the data set $D_e$ that deviate from $z_e$ much more than ${\beta_e^{-1/2}}$ have negligible weight, whereas, conversely, the local behavior at $z_e$ is dominated by the local cluster of data points in the ${\beta_e^{-1/2}}$--neighborhood of $z_e$, hence the term {\sl clustering}. In particular, outliers, or points outside that neighborhood, have negligible weight.

\begin{example}[Minimum--distance deviation]\label{y60Z7n2}
For a local discrepancy function of the form
\begin{equation}\label{2t16Lj}
    \Phi_e(y_e,z_e)
    =
    \|y_e-z_e\|_e^2 ,
\end{equation}
corresponding to a global discrepancy function of the form (\ref{2Cu1UG}), the Nash equilibrium condition (\ref{0VlMsK}) reduces to
\begin{equation}\label{aBn2Dp}
    \sigma_e
    =
    \sum_{i=1}^{N_e}
    p_{e,i}^*(z_e;\beta_e) \sigma_{e,i} ,
\end{equation}
where we write $y_{e,i} = (\epsilon_{e,i},\sigma_{e,i}) \in D_e$. Evidently, condition (\ref{aBn2Dp}) stipulates that the local stress $\sigma_e$ be the mean of its local cluster in the data set, hence the connection to $k$-means clustering \cite{Du:1999}. It should be noted that condition (\ref{aBn2Dp}) is {\sl implicit} by virtue of the dependence of the optimal weights $p_{e,i}^*$ on $\sigma_e$. However, we can render the stress computation {\sl explicit} if we specifically assume local discrepancy functions of the form
\begin{equation}\label{Zl4CeS}
    \Phi_e(y_e,z_e)
    =
    \| \xi_e - \epsilon_e \|_e^2 ,
\end{equation}
where we write $y_e=(\xi_e,\eta_e)$ and $z_e=(\epsilon_e,\sigma_e)$. Then, (\ref{aBn2Dp}) specializes to
\begin{equation}\label{npUr0e}
    \hat{\sigma}_{e}(\epsilon_e)
    =
    \sum_{i=1}^{N_{e}}
    p_{e,i}^*(\epsilon_e;\beta_e) \sigma_{e,i} ,
\end{equation}
with
\begin{equation}
    p_{e,i}^*(\epsilon_e;\beta_e)
    =
    \frac{1}{Z_{e}(\epsilon_e;\beta_e)} {\rm e}^{- \beta_e \|\epsilon_{e,i}-\epsilon_e\|_e^2},
\end{equation}
and
\begin{equation}
\label{eqdefzeh}
    Z_{e}(\epsilon_e;\beta)
    =
    \sum_{i=1}^{N_{e}}
     {\rm e}^{- \beta \|\epsilon_{e,i}-\epsilon_e\|_e^2},
\end{equation}
which is now explicit. We note that the stress $\sigma_e=\hat{\sigma}_{e}(\epsilon_e)$ is given by (\ref{npUr0e}) as an average of the stresses $\sigma_{e,i}$ sampled at strains $\epsilon_{e,i}$ in a neighborhood of $\epsilon_e$ of size $\sim \beta_e^{-1/2}$, with Gaussian weights depending on the distance $\|\epsilon_{e,i}-\epsilon_e\|_e$.
\hfill$\square$
\end{example}

\subsection{Analysis of convergence}
\label{BR7thb}

\begin{figure}[ht]
\begin{center}
	\begin{subfigure}{0.45\textwidth}\caption{} \includegraphics[width=0.99\linewidth]{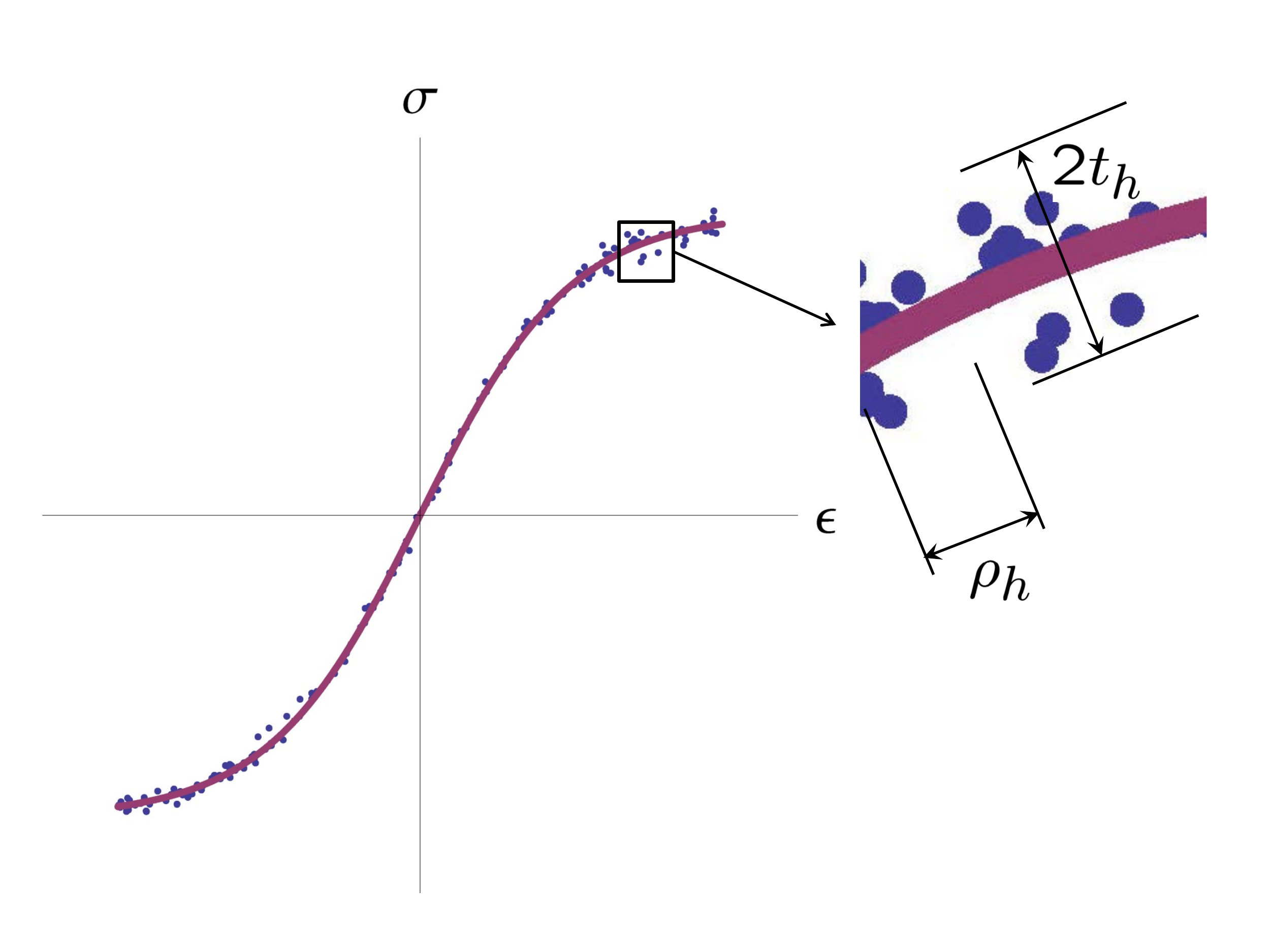}
	\end{subfigure}
	\begin{subfigure}{0.45\textwidth}\caption{} \includegraphics[width=0.99\linewidth]{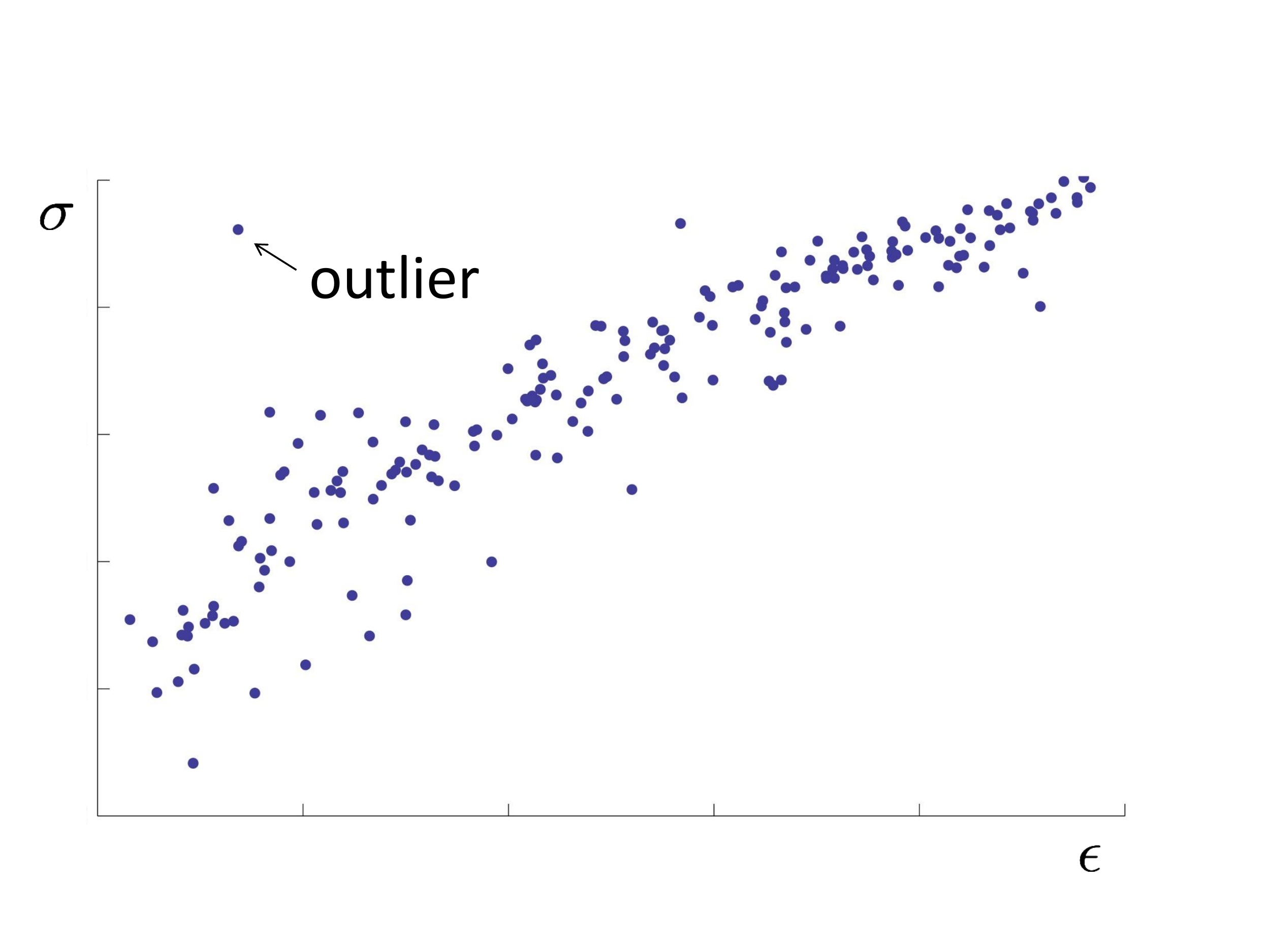}
	\end{subfigure}
    \caption{Two data-convergence scenarios. a) Uniform convergence: The data is contained within an increasingly narrow error band around a limiting stress-strain graph and samples the band uniformly and with increasing density. b) Noisy data with outliers: The data concentrates increasingly around a limiting stress-strain graph but exhibits scatter and outliers.} \label{rx4YN6}
\end{center}
\end{figure}

The convergence analysis presented in Section~\ref{jtA0LI} can be adapted to the case of approximation by point data sets in order to ascertain general conditions for convergence with respect to the data. For definiteness, we contemplate two data-convergence scenarios, shown schematically in Fig~\ref{rx4YN6}. The first scenario, which we refer to as {\sl uniform convergence}, arises when the data is contained within an increasingly narrow error band around the limiting stress-strain graph and samples the band uniformly and with increasing density over an increasingly larger data-coverage region of stress-strain space, Fig~\ref{rx4YN6}a. The second scenario concerns the case of noisy data with outliers, where the data concentrates increasingly around the limiting stress-strain graph but exhibits random scatter with outliers. For each of these two data-convergence scenarios, we proceed to discuss heuristically conditions on the data under which the Data-Driven solutions converge to a solution of the limiting problem. Rigorous mathematical statements and proofs are given in Appendix A. %~\ref{9tpjYN}.

\subsubsection{Uniformly converging data} \label{dQi63n}

We assume that $\hat{\sigma}_e : \mathbb{R}^d \to \mathbb{R}^d$ are continuous functions which obey material stability, in the sense of \eqref{eqsigmastable}. We denote by
\begin{equation}\label{eqdefgraph}
    G_e
    :=
    \{
        (\epsilon_e,\hat\sigma_e(\epsilon_e)): \epsilon_e\in \R^d
    \}
    \subseteq
    Z_e
\end{equation}
the graph of $\hat\sigma_e$. For every $h\in\N$, we consider a discrete set of points $D_{e,h}\subseteq Z_e$, which can be seen as a discrete approximation of $G_e$.  Starting from these data, we define an approximate stress function $\hat\sigma_{e,h}$ without regularization as in Section~\ref{M921rQ}. It bears emphasis that the functions $\hat\sigma_{e,h}$ thus defined are {\sl ansatz}-free and follow directly from the data.

Conditions on the data ensuring convergence of the Data-Driven solutions are set forth in the following theorem, which we enunciate heuristically. A rigorous statement and proof is given in Appendix~\ref{secuniformapprox}.

\begin{prop}[Uniform approximation, heuristics]\label{JB8qX5} Let $w_e > 0$ and let $\hat{\sigma}_e : \mathbb{R}^d \to \mathbb{R}^d$ be continuous functions which obey material stability, in the sense of \eqref{eqsigmastable}, and which are sufficiently regular. Let $B : \mathbb{R}^n \to \mathbb{R}^{m d}$ obey structural stability, in the sense of \eqref{eqBstructurstable}. Assume that the sequence $(D_{e,h})$ defines a locally uniform approximation, in the sense that:
\begin{enumerate}
    \item The data sets $D_{e,h}$ are contained within an increasingly narrow error-band around the limiting graph $G_e$.
    \item There is a bounded region covered by the data set $D_{e,h}$ such that, for every point on the limiting graph $G_e$, there is an increasingly close point in the data set $D_{e,h}$.
    \item The regions covered by the data sets increase in size.
\end{enumerate}
Then, there are Data-Driven approximate solutions $u_h$ that converge to a solution $u$ of the continuous problem \eqref{wpL8eA}. If the functions $\hat\sigma_e$ are strictly monotone, then the solution $u$ of the limiting problem is unique.
\end{prop}

Note that the functions $\hat\sigma_{e,h}$ defined in Section~\ref{M921rQ} are not continuous. Therefore the approximate problems do not, strictly speaking, have solutions. However, this difficulty can be overcome by recourse to mollification, or smoothing, as discussed in Appendix~\ref{secuniformapprox}.

We can make the heuristic assumptions in Prop.~\ref{JB8qX5} quantitative by introducing a sequence $\paramunifh\to0$ such that the distance of any point in $D_{e,h}$ to $G_e$ is no larger than $\paramunifh$, cf.~Fig.~\ref{rx4YN6}a. In addition, we introduce sequences $\paramfineh\to0$ and $R_h\to\infty$ such that, for any strain $\epsilon_e$ not larger than $R_h$, there is at least one data point which approximates the pair $(\epsilon_e, \hat\sigma_e(\epsilon_e))$ with an error no larger than $\paramfineh$, cf.~Fig.~\ref{rx4YN6}a. These assumptions are similar to the fine and uniform approximation properties formulated in \cite[Th.~3.3]{conti:2018}. We show in Proposition~\ref{propuniform} that, if $\hat\sigma_e$ is sufficiently regular and suitable relations between the sequences $\paramunifh$, $\paramfineh$ and $R_h$ hold, then there are approximate solutions to the problem with stress function $\hat\sigma_{e,h}$, constructed solely from the data $D_{e,h}$, that converge to solutions of the limiting problem with stress function $\hat\sigma_e$. Numerical examples illustrating this mode of convergence are presented in Section~\ref{cbb7Jh}.

\subsubsection{Noisy data with outliers}
\label{kuP0lq}

As above, we still assume that $\hat{\sigma}_e : \mathbb{R}^d \to \mathbb{R}^d$ are continuous functions which obey material stability, in the sense of \eqref{eqsigmastable}, and denote by   $G_e$ the graph of $\hat\sigma_e$, in the sense of \eqref{eqdefgraph}. For every $h\in\N$, we consider a discrete set of data points $D_{e,h}\subseteq Z_e$, which can again be seen as a discrete approximation to $G_e$.  Starting from these data, we define an approximate stress function $\hat\sigma_{e,h}$ with regularization as in (\ref{npUr0e}).  We remark that, since $D_{e,h}$ is a nonempty finite set the weights are all positive and depend continuously on $\epsilon_e$. In particular, the functions $\sigma_{e,h}$ defined as in (\ref{npUr0e}) are continuous.

The following proposition ensures that, if $D_{e,h}$ is a good approximation of the graph of $\hat\sigma_e$, then the solutions of the problems defined by the functions $\hat\sigma_{e,h}$ converge to solutions of the limiting problem defined by the functions $\hat\sigma_{e}$. A rigorous statement and proof are given in Appendix~\ref{secappappoutl}.

\begin{prop}[Noisy data with outliers, heuristics]\label{G2zOb4}
Let $w_e > 0$ and let $\hat{\sigma}_e : \mathbb{R}^d \to \mathbb{R}^d$ be continuous functions which obey material stability, in the sense of \eqref{eqsigmastable}, and are sufficiently regular. Let $B : \mathbb{R}^n \to \mathbb{R}^{m d}$ obey structural stability, in the sense of \eqref{eqBstructurstable}. For every $h\in\N$, consider a finite set of data points $D_{e,h}\subseteq Z_e$. Assume:
\begin{enumerate}
    \item There are sufficiently few outliers.
    \item The data do not cluster in strain space and the covering of strain space is approximately uniform.
    \item The region of stress-strain space covered by the data is bounded and increasingly large.
    \item For the sampled strains, there are incresingly many data points in stress-strain space that are close to the limit stress-strain graph.
\end{enumerate}
Then, the Data-Driven solutions $u_h$ converge to a solution $u$ of the continuous problem \eqref{wpL8eA}. If the functions $\hat\sigma_e$ are strictly monotone, then the solution $u$ of the limiting problem is unique.
\end{prop}

As before, we can make the heuristic assumptions in Prop.~\ref{G2zOb4} quantitative. First, as $D_{e,h}$  is finite, it is necessarily bounded with bound $R_h\to\infty$. We denote by $\paramfineh$ the resolution in strain space. We assume that there are many data points that correctly characterize strains which are not too large. Specifically, we assume  that the number of data points with strain in a neighborhood of size $\paramfineh$ of every strain $\epsilon_e$,  with $\|\epsilon_e\|_e\le R_h$, is contained between $N_h$ and $C_h$, with $1\ll N_h\ll C_h$. In addition, we require that there not be too many outliers. We say that a pair $(\xi_e,\eta_e)\in D_{e,h}$ is an {\sl outlier} if the measured stress $\eta_e$ differs from the `true' stress $\hat\sigma(\xi_e)$ corresponding to the measured strain $\xi_e$ by more than $\paramunifh$. Then, we assume that the number of outliers in the vicinity of every strain $\epsilon_e$ is bounded by $M_h$. The proposition then ensures that, if appropriate relations between these quantities hold, there is a sequence $\beta_{e,h}\to\infty$ such that the data-driven solution converges to the `true' solution.

\subsubsection{Numerical examples}\label{cbb7Jh}

\begin{figure}[ht]
\begin{center}
\includegraphics[width=0.65\linewidth]{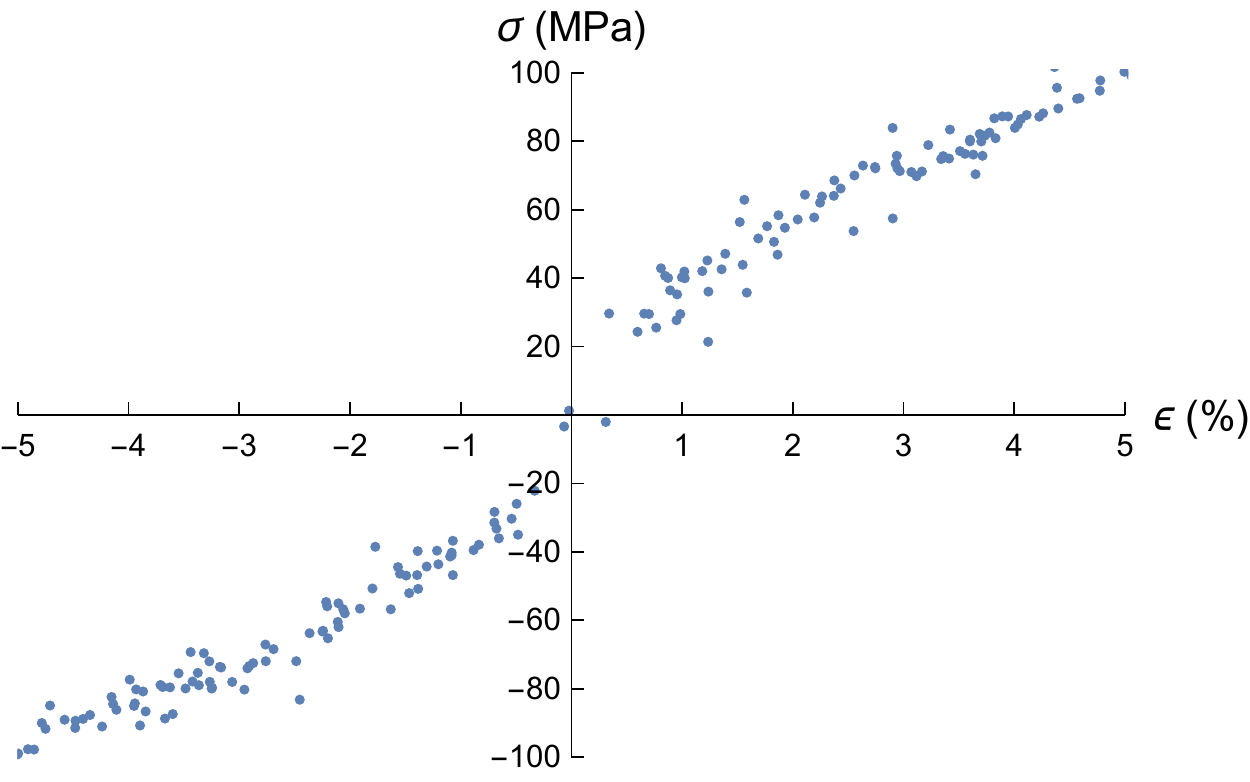}
    \caption{Two-hundred data-point set. The strains are uniformly distributed and the data points are obtained by adding random Gaussian noise to the limiting stress-strain curve. } \label{1Z6rln}
\end{center}
\end{figure}

The examples that follow illustrate the type of convergence that can be expected when the effective material laws are inferred from point material data, including max-ent regularization. We specifically consider the simple case of a uniaxial material obeying the stress-strain relation
\begin{equation}\label{J9D7pa}
    \epsilon
    =
    \frac{\sigma}{\mathbb{C}_0} +
    \Big( \frac{\sigma}{\mathbb{C}_1} \Big)^2 .
\end{equation}
This material behavior is tested and a sample point-data set of $N$ points is generated. The test apparatus is imprecise and the data exhibits additive Gaussian noise evaluated using the Box-Muller transformation \cite{BoxMuller1958}. A representative sample corresponding to $\mathbb{C}_0 = 200$MPa, $\mathbb{C}_1 = 44.7214$MPa, strain standard deviation $0.05$\%, stress standard deviation $1$Mpa and $N = 200$, is shown in Fig.~\ref{1Z6rln}.

\begin{figure}[ht]
\begin{center}
	\begin{subfigure}{0.45\textwidth}\caption{} \includegraphics[width=0.99\linewidth]{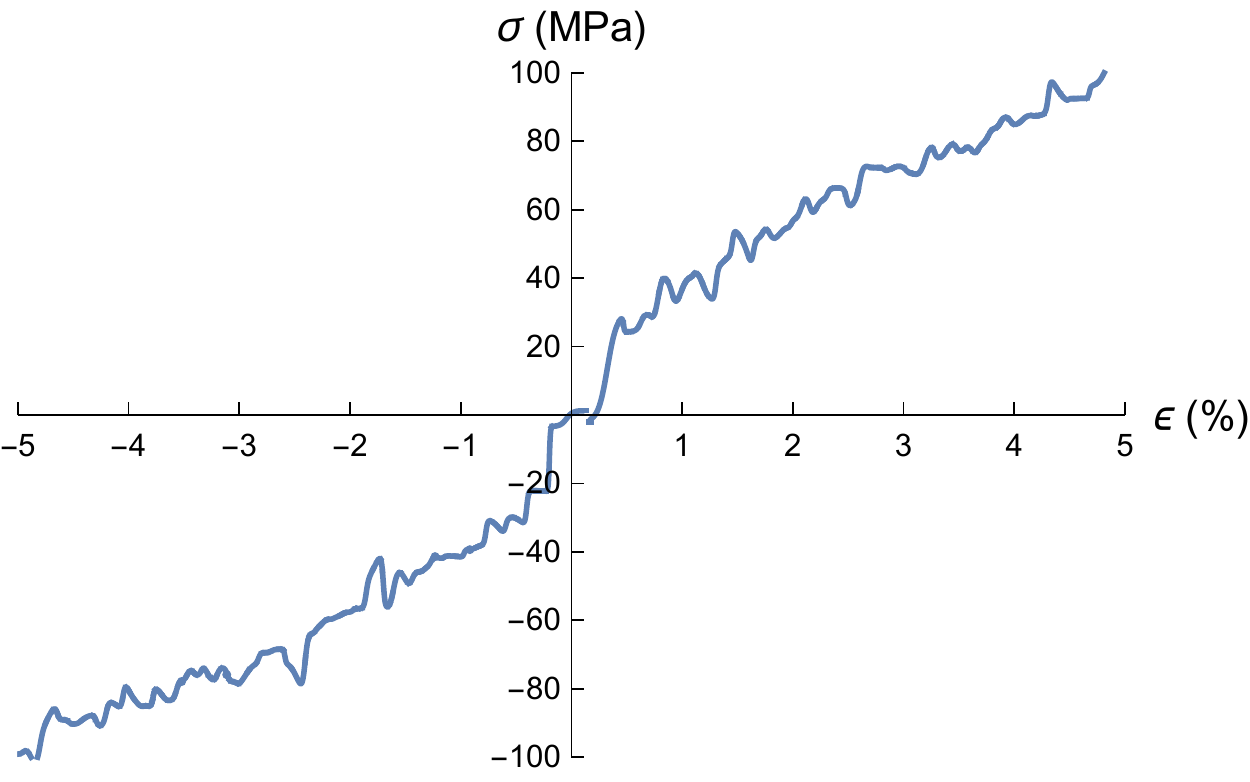}
	\end{subfigure}
	\begin{subfigure}{0.45\textwidth}\caption{} \includegraphics[width=0.99\linewidth]{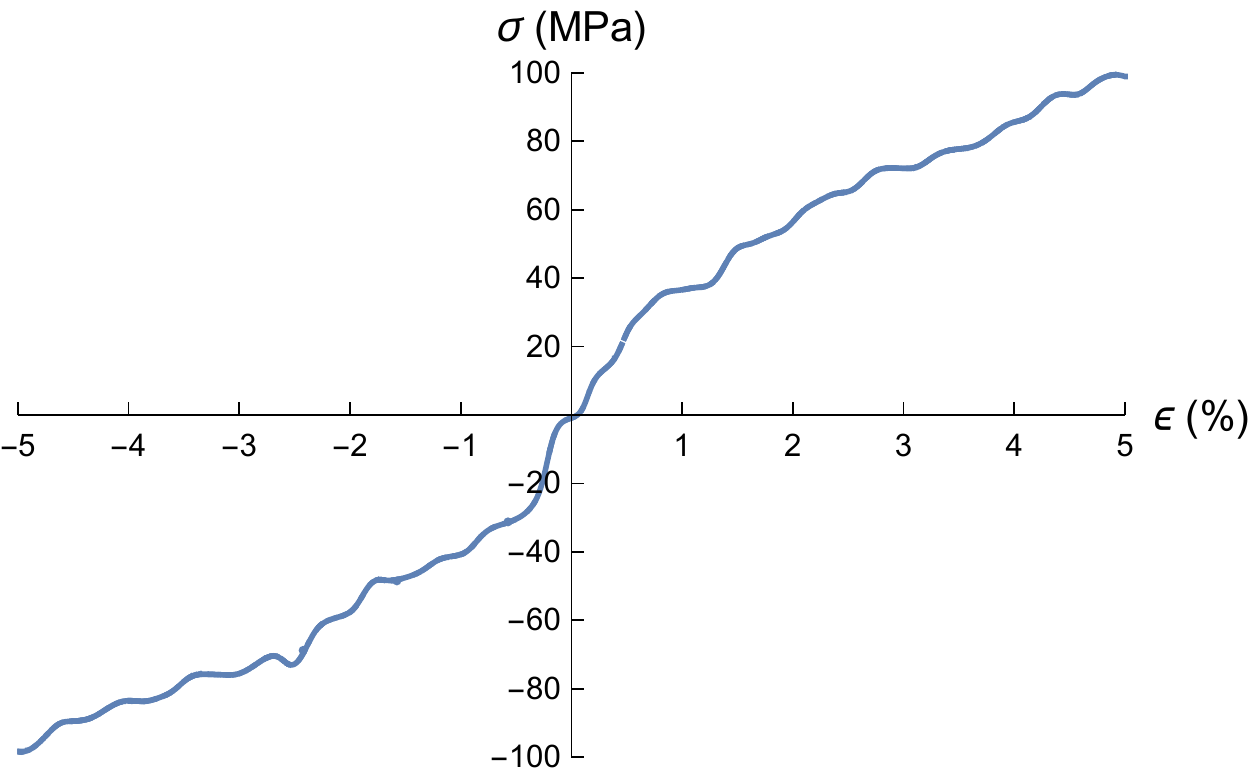}
	\end{subfigure}
	\begin{subfigure}{0.45\textwidth}\caption{} \includegraphics[width=0.99\linewidth]{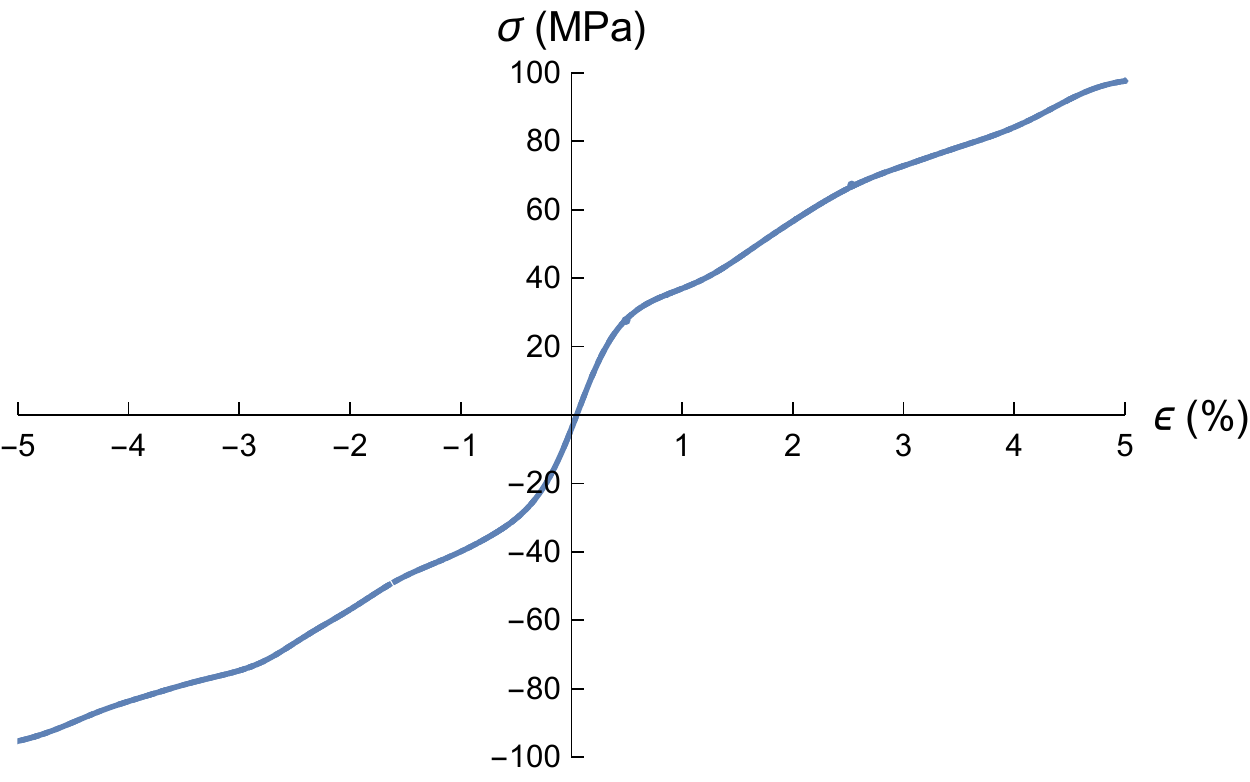}
	\end{subfigure}
	\begin{subfigure}{0.45\textwidth}\caption{} \includegraphics[width=0.99\linewidth]{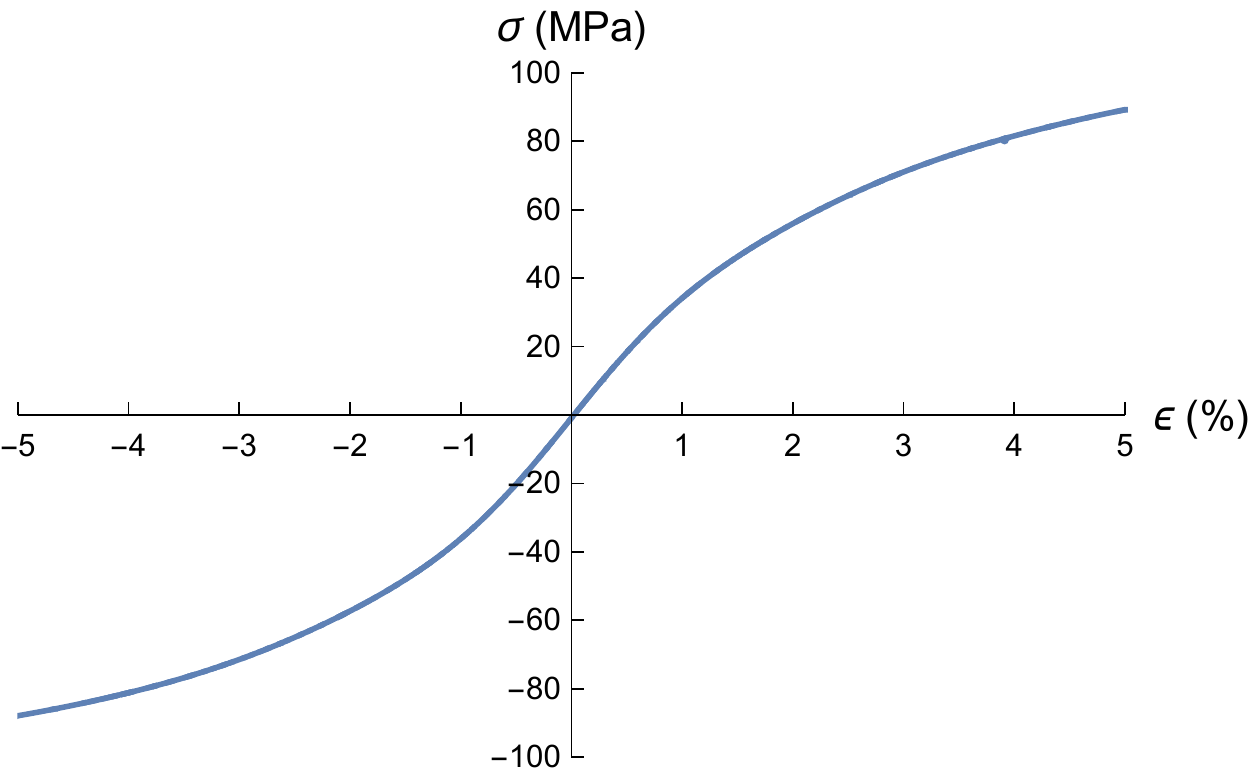}
	\end{subfigure}
    \caption{Regularized effective stress-strain curves. a) $\beta = 1600$; b) $\beta = 160$; c) $\beta = 16$; and d) $\beta = 1.6$.} \label{atMH3z}
\end{center}
\end{figure}

For simplicity, we choose a strain-controlled discrepancy function (\ref{Zl4CeS}) and compute stresses as in (\ref{npUr0e}). Fig.~\ref{atMH3z} shows four regularized effective stress-strain curves obtained from the data of Fig.~\ref{1Z6rln} with $\beta = 1600$, $160$, $16$ and $1.6$. As may be seen from the figures, the max-ent regularization has the effect of increasing the smoothness of the effective stress-strain curve, which becomes monotonically increasing for sufficiently small $\beta$.

\begin{figure}[htp]
\begin{center}
	\begin{subfigure}{0.45\textwidth}\caption{} \includegraphics[width=0.8\linewidth]{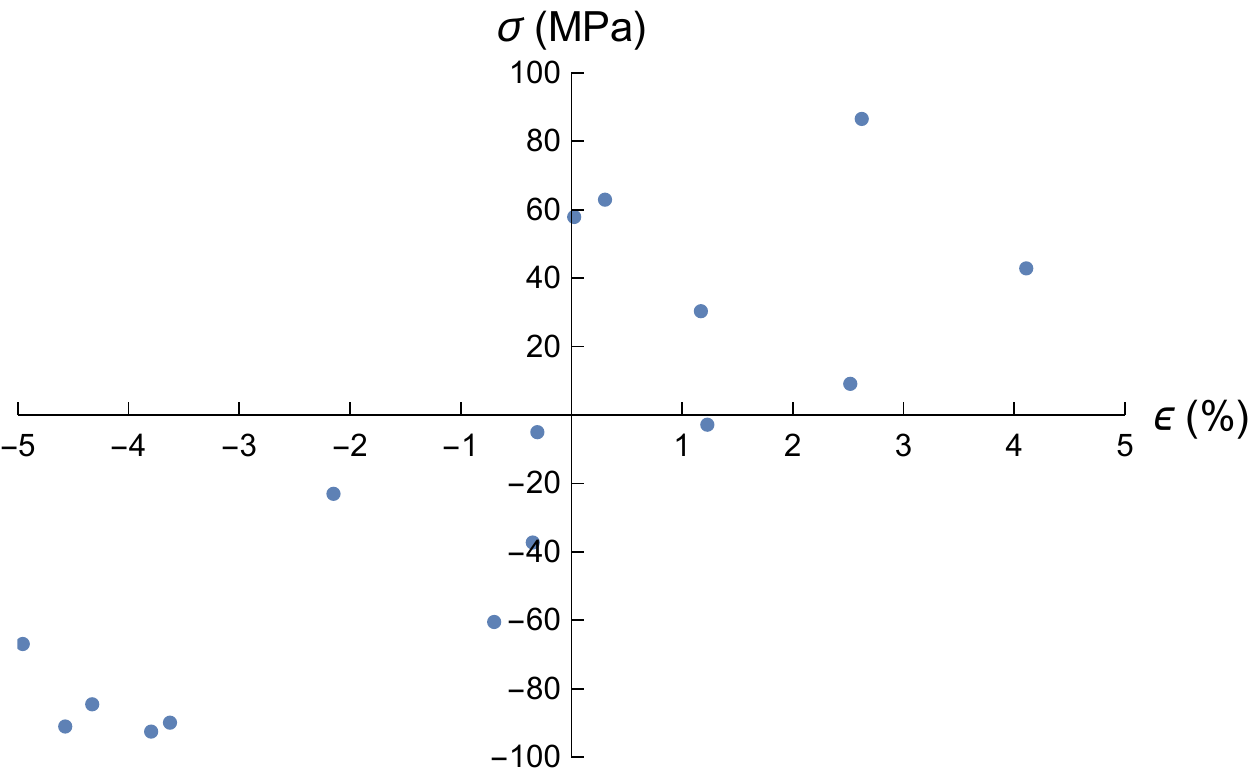}
	\end{subfigure}
	\begin{subfigure}{0.45\textwidth}\caption{} \includegraphics[width=0.8\linewidth]{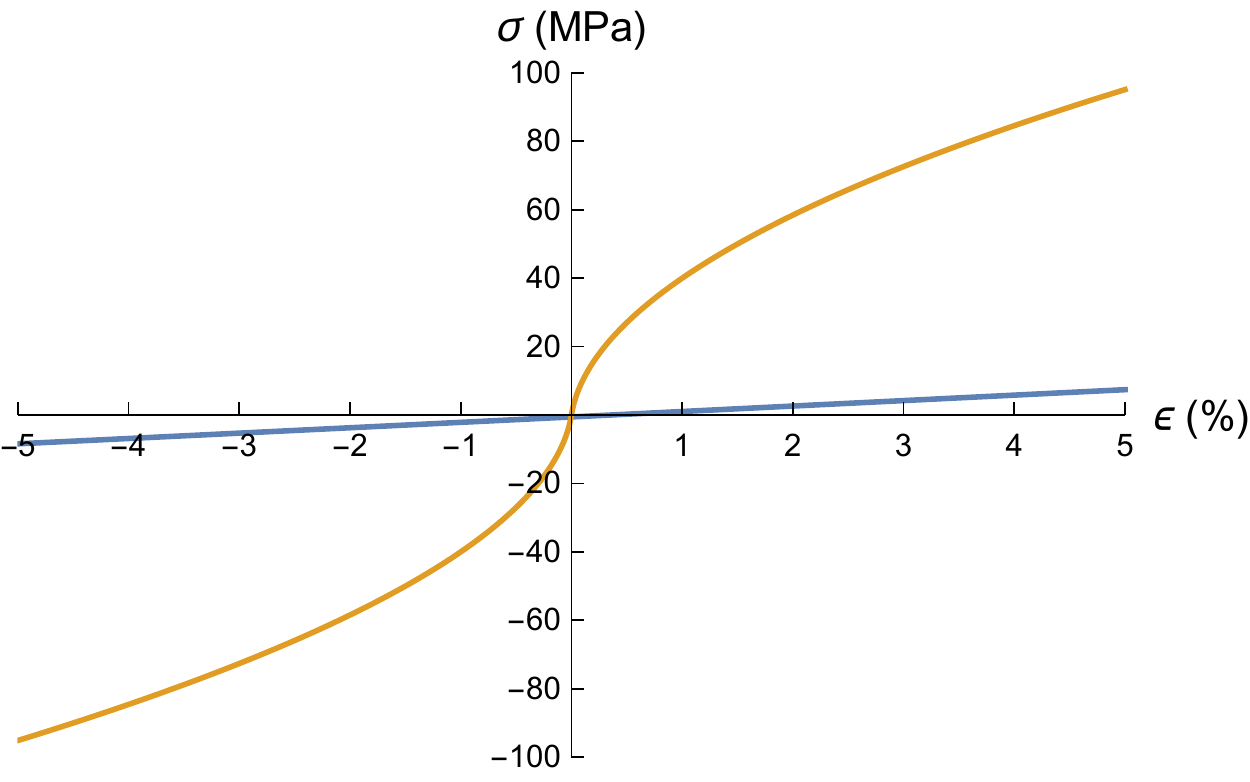}
	\end{subfigure}
	\begin{subfigure}{0.45\textwidth}\caption{} \includegraphics[width=0.8\linewidth]{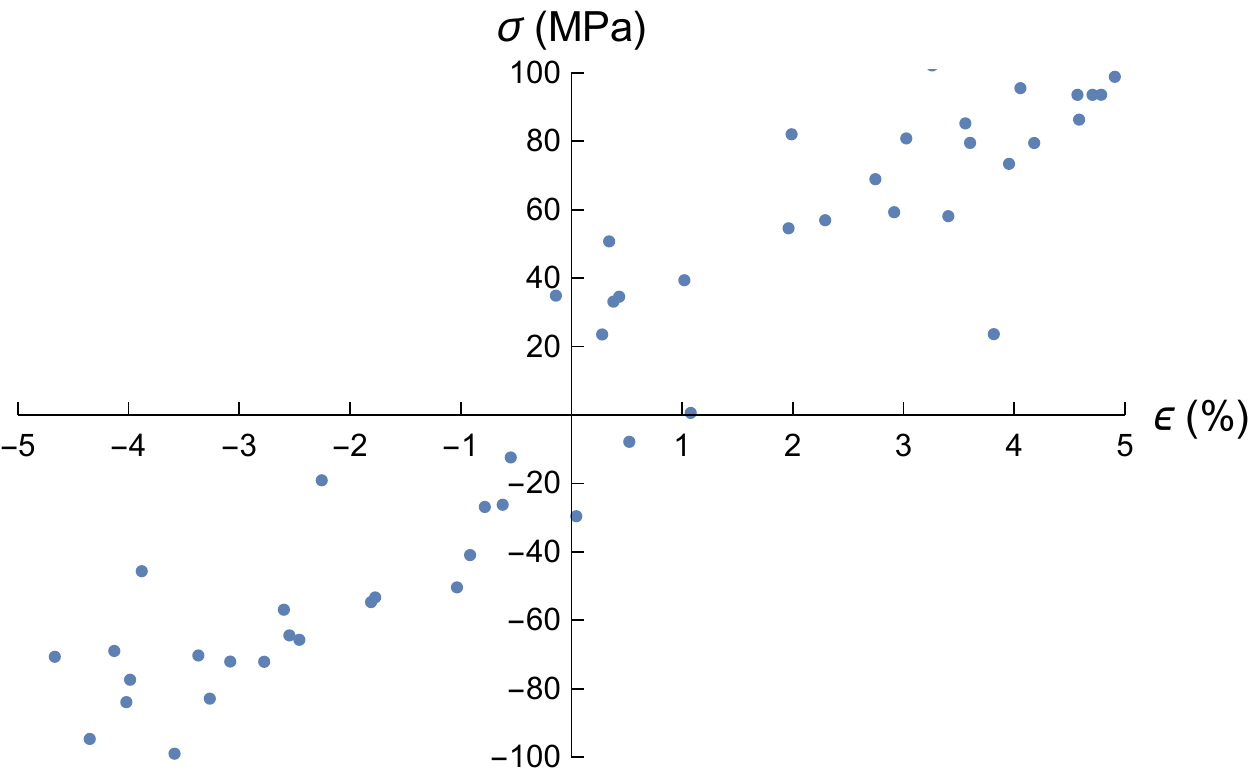}
	\end{subfigure}
	\begin{subfigure}{0.45\textwidth}\caption{} \includegraphics[width=0.8\linewidth]{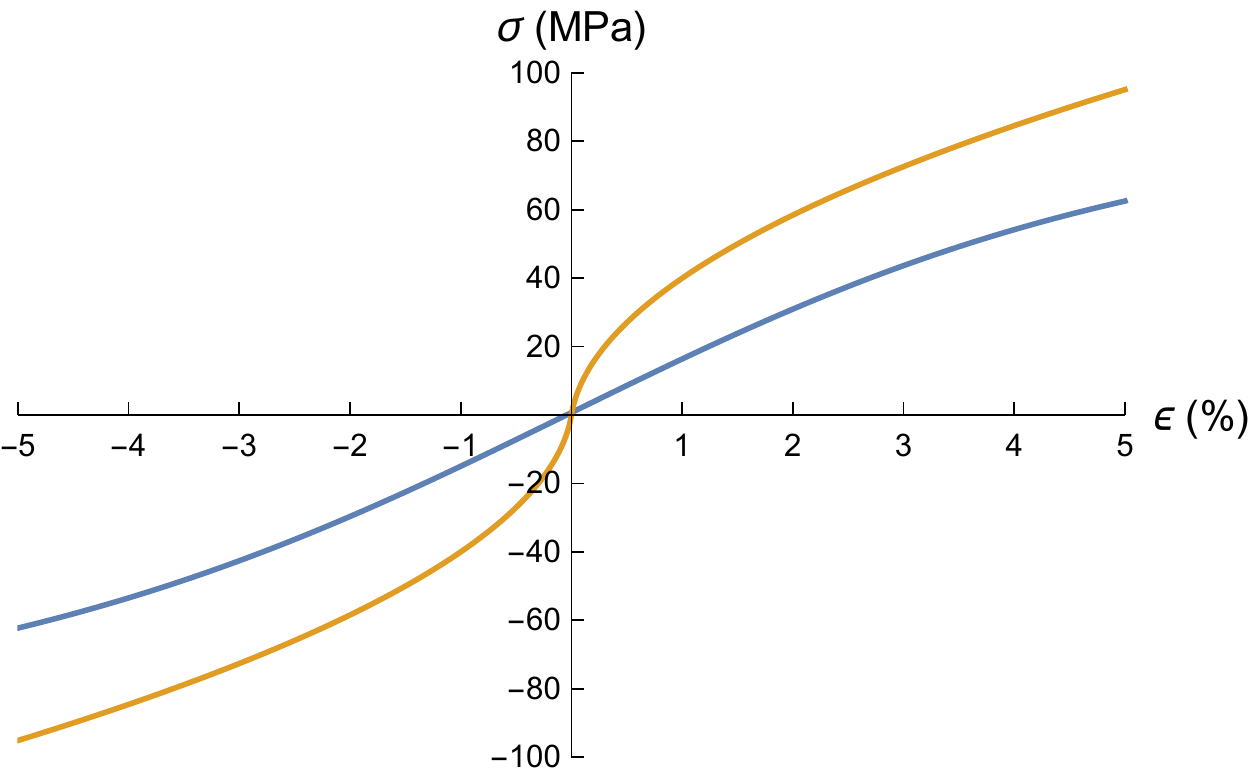}
	\end{subfigure}
	\begin{subfigure}{0.45\textwidth}\caption{} \includegraphics[width=0.8\linewidth]{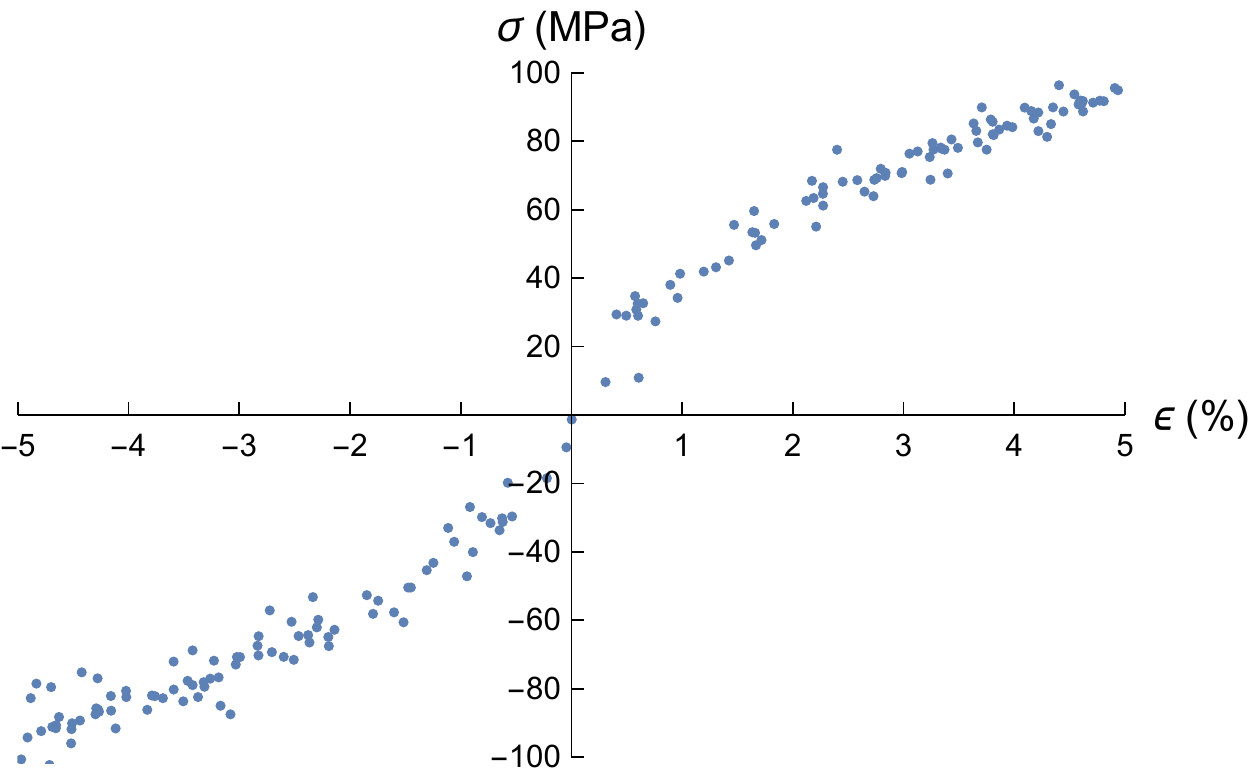}
	\end{subfigure}
	\begin{subfigure}{0.45\textwidth}\caption{} \includegraphics[width=0.8\linewidth]{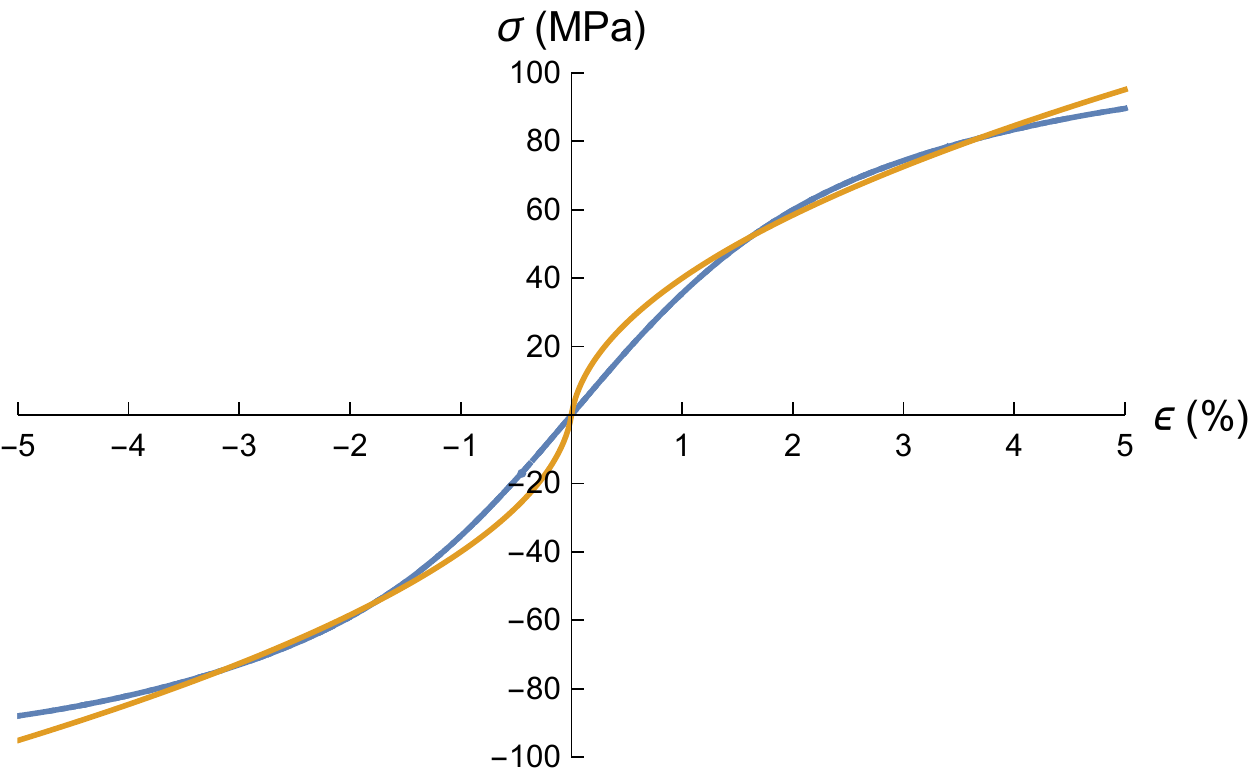}
	\end{subfigure}
	\begin{subfigure}{0.45\textwidth}\caption{} \includegraphics[width=0.8\linewidth]{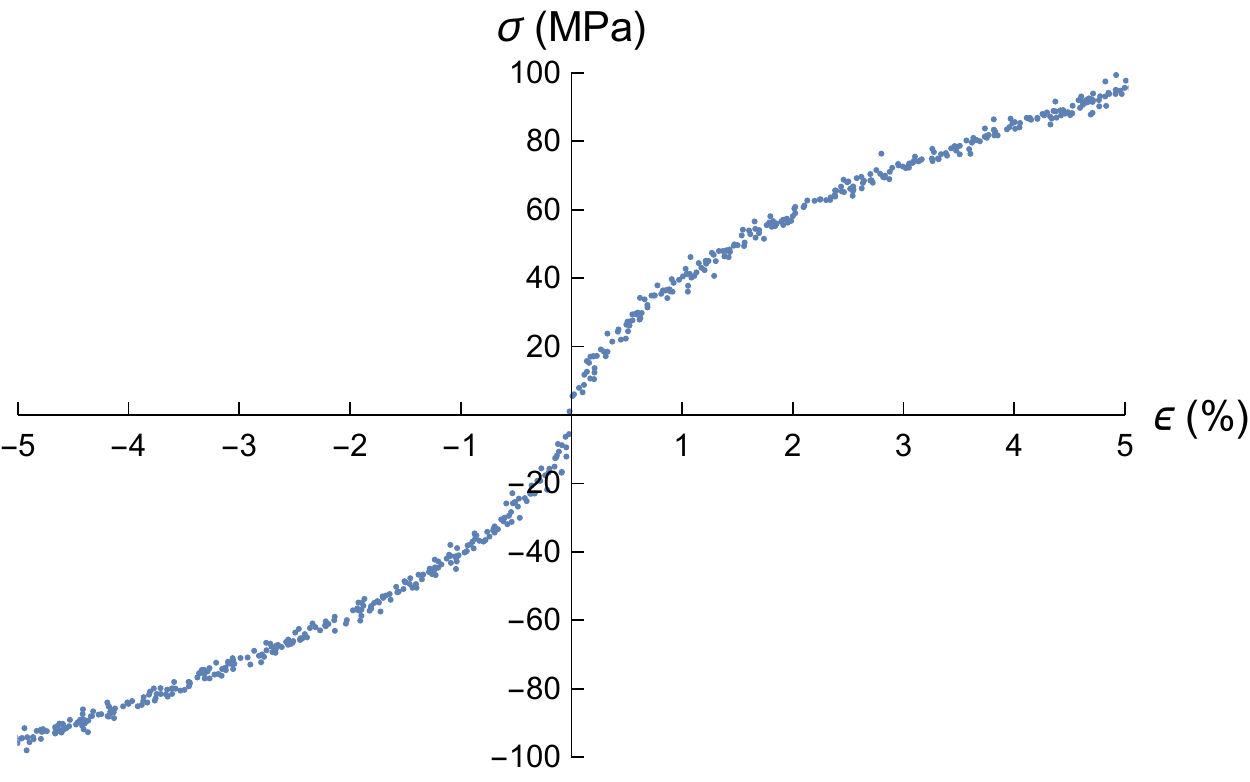}
	\end{subfigure}
	\begin{subfigure}{0.45\textwidth}\caption{} \includegraphics[width=0.8\linewidth]{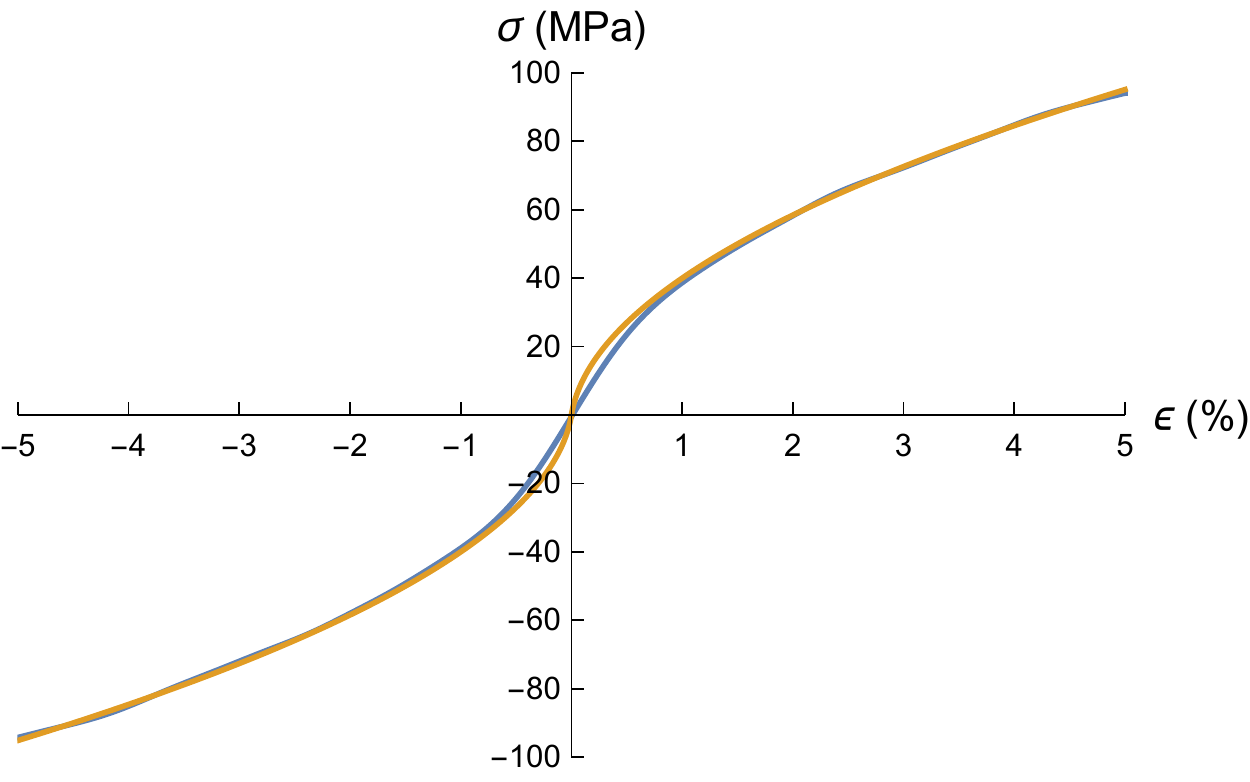}
	\end{subfigure}
	\begin{subfigure}{0.45\textwidth}\caption{} \includegraphics[width=0.8\linewidth]{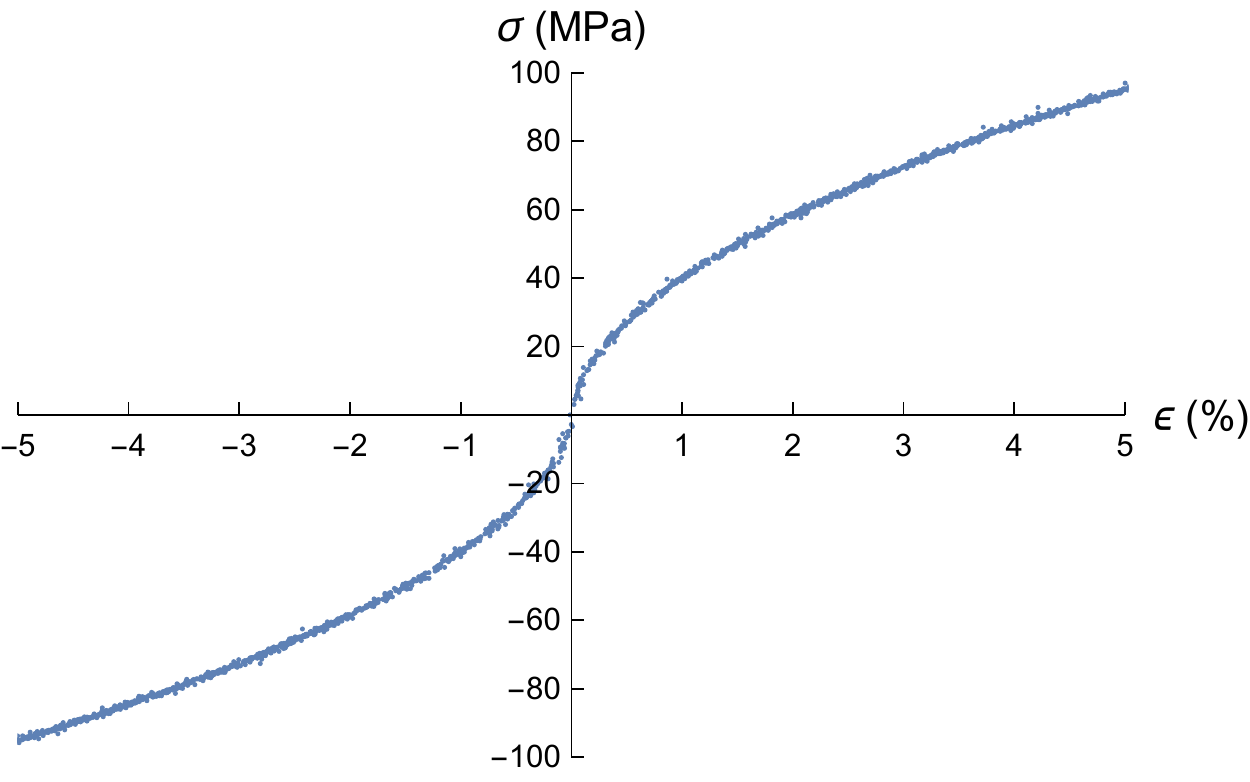}
	\end{subfigure}
	\begin{subfigure}{0.45\textwidth}\caption{} \includegraphics[width=0.8\linewidth]{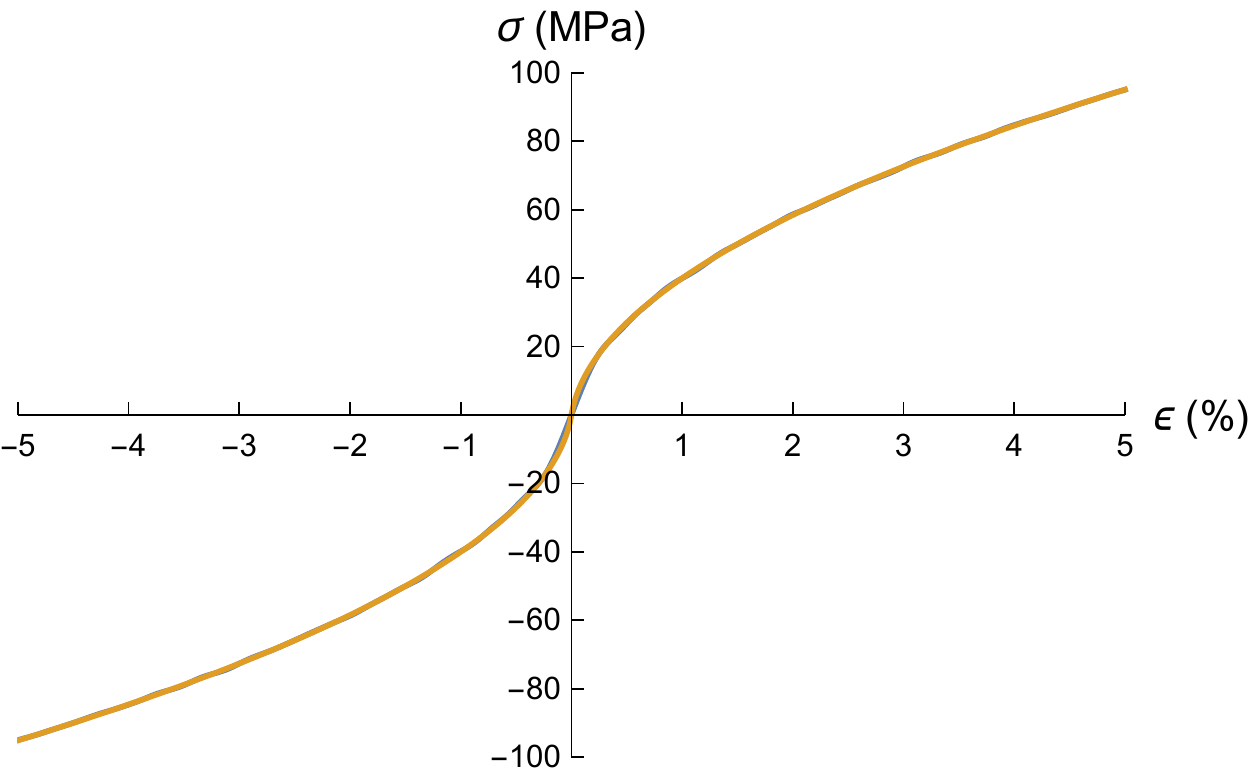}
	\end{subfigure}
    \caption{Example of convergence of the effective material laws with respect to noisy data. Left column: Sampled material data sets; Right column: Regularized effective material laws (blue) and exact material law (orange). a,b) $N=10$; c,d) $N=32$; e,f) $N=100$; g,h) $N=316$; i,j) $N=1000$. } \label{AXp6f3}
\end{center}
\end{figure}

Fig.~\ref{AXp6f3} illustrates convergence of the effective material law with respect to the data. Guided by Prop.~\ref{G2zOb4}, we consider data sets that improve steadily by virtue of an increasing sample size and a simultaneous decrease in the size of the noise. This type of data convergence may be expected to result, e.~g., from experimental campaigns in which both the size of the samples and the precision of the measurements increase steadily. Specifically, in the example shown in Fig.~\ref{AXp6f3} the strain standard deviation is assumed to decrease with the number $N$ of data points as $(10/N)$\%, and the stress standard deviation as $(200/N)$ MPa, along the sequence of material data sets. The annealing sequence is chosen so that $1/\sqrt{\beta}$ is commensurate with the strain standard deviation. Specifically, we choose $\beta = 10^3 (N/5)^2$. As sequence of material data samples of increasing size $N = 10$, $32$, $100$, $316$ and $1000$ are shown in the left column of Fig.~\ref{AXp6f3}. The right column shows the corresponding effective material laws and the limiting material law (\ref{J9D7pa}), or 'ground truth', for comparison. The expected uniform convergence of the effective material laws to the limiting material law over the strain interval under consideration is evident from the figure.

\section{Implementation and examples of application}
\label{Nd4BgR}

We conclude with examples of application that illustrate the suitability of Data-Driven game-theoretical approaches for implementation and application within a standard finite element framework. This property notwithstanding, we emphasize that all the calculations that follow are carried out directly from material data and that at no time the data is modelled in the sense of fitting to {\sl ad hoc} predetermined classes of functions, such as neural networks.

\subsection{Evaluation of effective material laws}
\label{T5RuNG}

We recall that unregularized Data-Driven games result in effective local material laws of the form (\ref{kP26d1}). The local evaluation of the effective constitutive relation is summarized in Algorithm~\ref{p4ZuLW}. The effective problem (\ref{wpL8eA}) can then be solved for the displacements, e.~g., by means of a matrix-free iterative solver such as dynamic relaxation \cite{OakleyKnight1995}, cf.~\S~\ref{Nedi3P} and ~\ref{RP4nsY}.

\begin{algorithm}[H]
\caption{Game-theoretical Data--Driven material law -- Unregularized}
\label{p4ZuLW}
\begin{algorithmic}
\REQUIRE
\STATE i) Local stress-strain point-data set $D_{e}$.
\STATE ii) Local strain $\epsilon_e$.
\STATE {\bf Then}:
\STATE i) Find $(\xi_e,\eta_e)$ in $D_e$ such that $\xi_e$ is closest to $\epsilon_e$.
\STATE ii) Return $\sigma_e = \eta_e$.
\end{algorithmic}
\end{algorithm}

Regularized Data-Driven games, by contrast, result in effective local material laws of the general implicit form (\ref{0VlMsK}), for a general discrepancy function. For the particular choice of a distance discrepancy function,  (\ref{0VlMsK}) reduces to the implicit form (\ref{aBn2Dp}), if the discrepancy function combines stress and strain states, or to the explicit form (\ref{npUr0e}) if the discrepancy function refers to strain states only.

We recall that, when the system is regularized properly, the effective local material laws $\hat{\sigma}_e(\epsilon_e)$ are smooth and have well-defined material tangents $D\hat{\sigma}_e(\epsilon_e)$. An expression for the tangents can be derived by differentiation of (\ref{0VlMsK}) taking into account the implicit dependence of the weights $p^*_{e,i}$ on the local state $z_e=(\epsilon_e,\sigma_e)$, with the result
\begin{equation}\label{zA9siX}
    \frac{\partial p^*_{e,i}}{\partial\epsilon_e} =
    -\beta_e p^*_{e,i}\Big( \frac{\partial\Phi_{e,i}}{\partial\epsilon_e}
          -\sum_{j=1}^{N_e} p^*_{e,j} \frac{\partial\Phi_{e,j}}{\partial\epsilon_e} \Big) ,
\end{equation}
where we write $\Phi_{e,i}=\Phi_e(y_{e,i},z_e)$. In the particular case of a distance local discrepancy function (\ref{2t16Lj}), the tangents follow by differentiation of (\ref{aBn2Dp}) and (\ref{zA9siX}) specializes to
\begin{equation}
    \frac{\partial p^*_{e,i}}{\partial\epsilon_e}
    =
    -
    {2\beta_e} p^*_{e,i}
    \Big[
        \mathbb{C}_e:(\bar{\epsilon}_e-\epsilon_{e,i})
        +
        (\sigma_e-\sigma_{e,i}):\mathbb{C}_e^{-1}:\frac{\partial\sigma_e}{\partial\epsilon_e}
    \Big] ,
\end{equation}
with $\bar{\epsilon}_e = \sum_{i=1}^{N_e} p_{e,i}^*\epsilon_{e,i}$, whereupon the tangents follow as
\begin{multline}\label{61XLcU}
    \frac{\partial\sigma_e}{\partial\epsilon_e}
    =
    \Big[
        \mathbb{I}
        -
        \sum_{i=1}^{N_e}{2\beta_e} p_{e,i}^*
        \big(
            \sigma_{e,i}\otimes[\mathbb{C}_e^{-1}:(\sigma_{e,i}-\sigma_e)]
        \big)
    \Big]^{-1} \\
    \Big[
        \sum_{j=1}^{N_e} {2\beta_e} p_{e,j}^* \sigma_{e,j}
        \otimes
        \big(
            \mathbb{C}_e:(\epsilon_{e,j}-\bar{\epsilon}_e)
        \big)
    \Big] ,
\end{multline}
where $\mathbb{I}$ denotes the $4$th order identity tensor. Finally, for discrepancy functions of the strain-distance type (\ref{Zl4CeS}), the tangents follow by differentiation of (\ref{npUr0e}, whereupon (\ref{61XLcU}) further specializes to the simple form
\begin{equation}\label{fEtr8X}
    \frac{\partial\sigma_e}{\partial\epsilon_e}
    =
    \sum_{j=1}^{N_e}
    {
        {2\beta_e} p_{e,j}^*  \sigma_{e,j}\otimes
        \big(
            \mathbb{C}_e:(\epsilon_{e,j}-\bar{\epsilon}_e)
        \big)
    } ,
\end{equation}
which can be evaluated explicity from the strains.

We note that, unlike the cooperative case, the tangent moduli $D\hat{\sigma}_e(\epsilon_e)$ resulting from the non-cooperative Data-Driven games are non-symmetric in general, which attests, to the non-variational global structure of said games. In the cooperative mood, simple estimates additionally show that the tangent moduli are positive definite for ${\beta_e}$ large enough \cite{Kirchdoerfer:2017}. In the non-cooperative mood, the effective material law $\hat{\sigma}_e(\epsilon_e)$ is expected to be monotonic for sufficiently large $\beta_e$, cf.~Prop.~\ref{propoutliers}, which in turn requires the symmetric part of $D\hat{\sigma}_e(\epsilon_e)$ to be positive definite.

\begin{algorithm}[H]
\caption{Game-theoretical Data--Driven material law -- Regularized}
\label{r2NBY6}
\begin{algorithmic}
\REQUIRE
\STATE i) Local regularization parameter $\beta_{e}$.
\STATE ii) Local stress-strain point-data set $D_{e}$.
\STATE iii) Local strain $\epsilon_e$.
\STATE {\bf Then}:
\STATE i) Solve (\ref{0VlMsK}), or (\ref{aBn2Dp}), for stress $\sigma_e$; or evaluate $\sigma_e$ from (\ref{npUr0e}).
\STATE ii) Evaluate tangents $D\sigma_e$ from (\ref{zA9siX}), or (\ref{61XLcU}), or (\ref{fEtr8X}).
\STATE iii) Return $\sigma_e$, $D\sigma_e$.
\end{algorithmic}
\end{algorithm}

The evaluation of the effective regularized constitutive laws and attendant tangents is summarized in Algorithm~\ref{r2NBY6}. The effective problem (\ref{wpL8eA}) can then be solved for the displacements, e.~g., by means of standard solvers such as Newton-Raphson or nonlinear Krylov-GMRES, cf.~\S~\ref{ssec:cube}.

\subsection{Illustrative examples}

We conclude this section with three simple examples of application intended to demonstrate how the proposed non-cooperative game-theoretical Data-Driven paradigm can be implemented and deployed, with or without regularization, within a standard finite-element framework. We note that, a detailed convergence analysis of the global results is not necessary since convergence can be ascertained entirely at the local level, cf.~Section~\ref{cbb7Jh}.

\subsubsection{Rotating rod}
\label{Nedi3P}

\begin{figure}[H]
\begin{center}
\begin{subfigure}{0.49\textwidth}\caption{}
\includegraphics[width=0.99\linewidth]{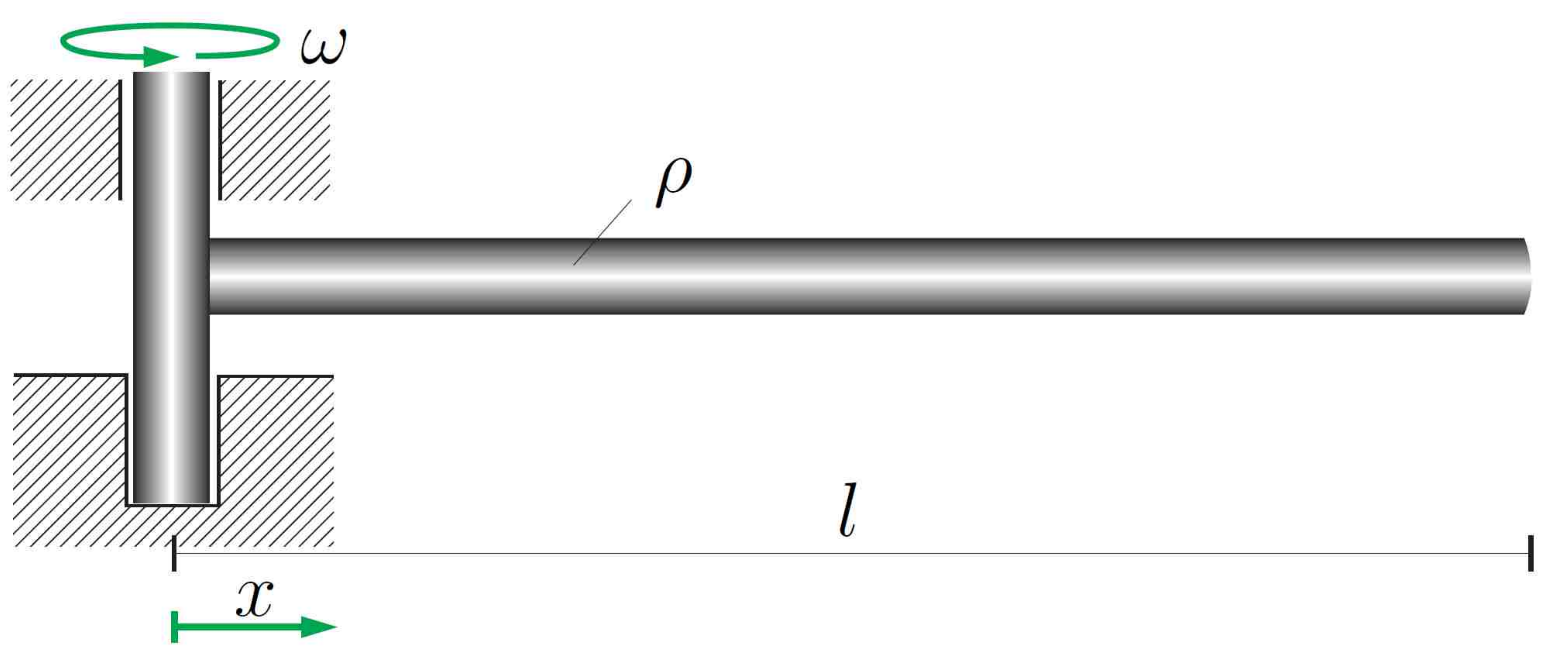}
\bigskip\smallskip
\end{subfigure}
\begin{subfigure}{0.49\textwidth}\caption{}
\includegraphics[width=0.99\linewidth]{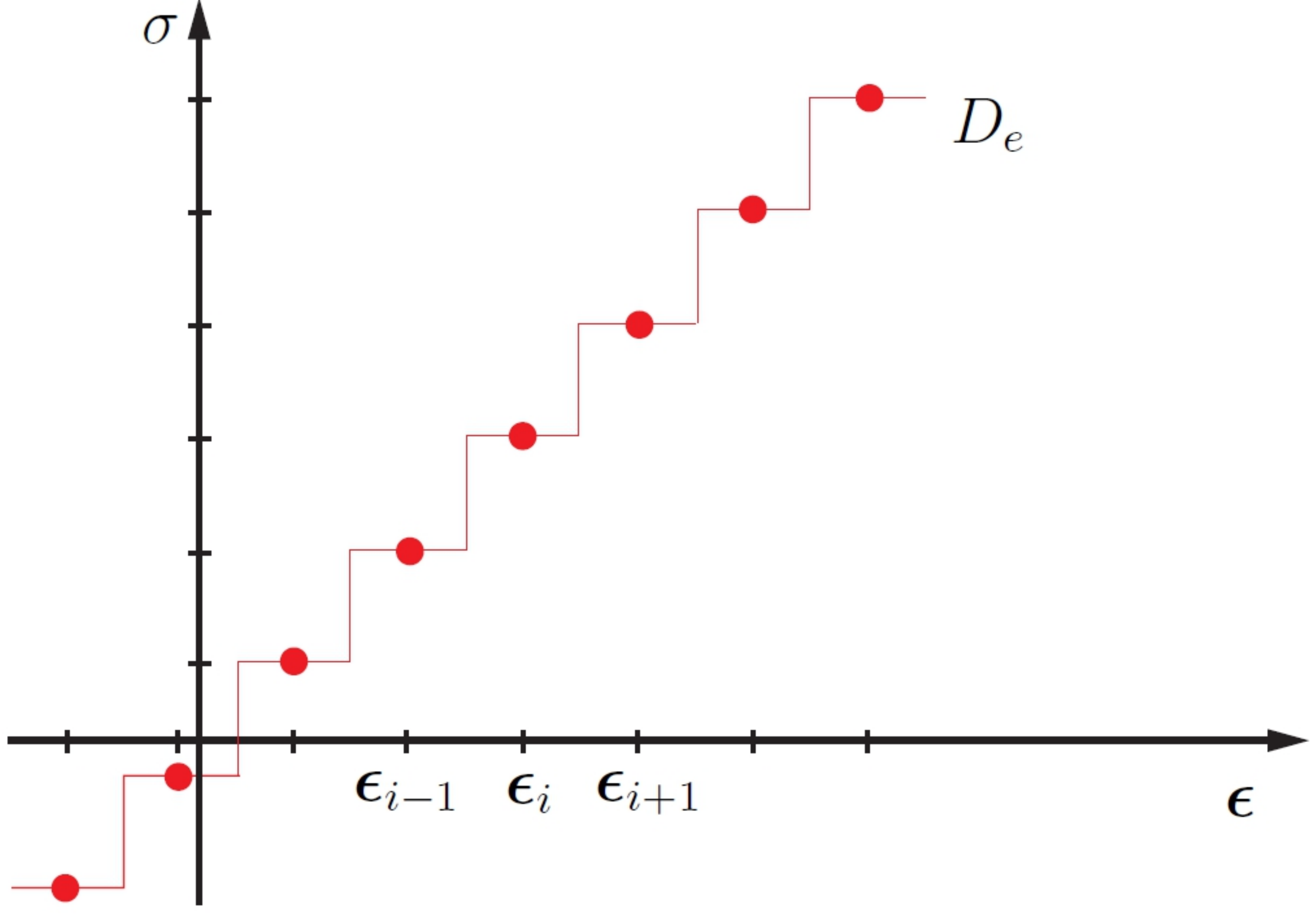}
\end{subfigure}
\caption{A) Linear-elastic rod
%of  length $l$, cross section $A$, Hooke's modulus $\mathbb{C}$ and mass density $\rho$
rotating with angular velocity $\omega$. B) Local material point-data set $D_e$ obtained by sampling Hooke's law at equidistant strains. }\label{fig:spinningRod}
\end{center}
\end{figure}

A first simple illustration of the non-cooperative Data-Driven paradigm is provided by the problem of a rod rotating around one end with angular velocity of $\omega$, Fig.~\ref{fig:spinningRod}a. The spinning motion is simply accounted for through a constant centrifugal body force of magnitude $\rho \omega^2$. The rod behaves in uniaxial stress and the material obeys Hooke's law. In calculations, the rod is discretized into equal linear finite elements along its length.

The local material point--data sets $D_e$ are obtained by sampling without noise the material law at equidistant strains over the expected range covered by the solution, Fig.~\ref{fig:spinningRod}b. The figure also shows the effective piecewise constant, or stepwise, local material law that results from a minimum strain distance game with no regularization, cf.~Section~\ref{M921rQ} and Example~\ref{y60Z7n2}. In order to render the effective stress-strain relation single-valued, at strains equidistant from the sampling strains we choose the smallest, in absolute value, of the two possible stresses. Owing to the lack of smoothness of the resulting material law, solutions must be obtained using matrix-free solvers. In calculations, we specifically use dynamic relaxation \cite{OakleyKnight1995}.

\begin{figure}
\begin{center}
\includegraphics[width=0.79\linewidth]{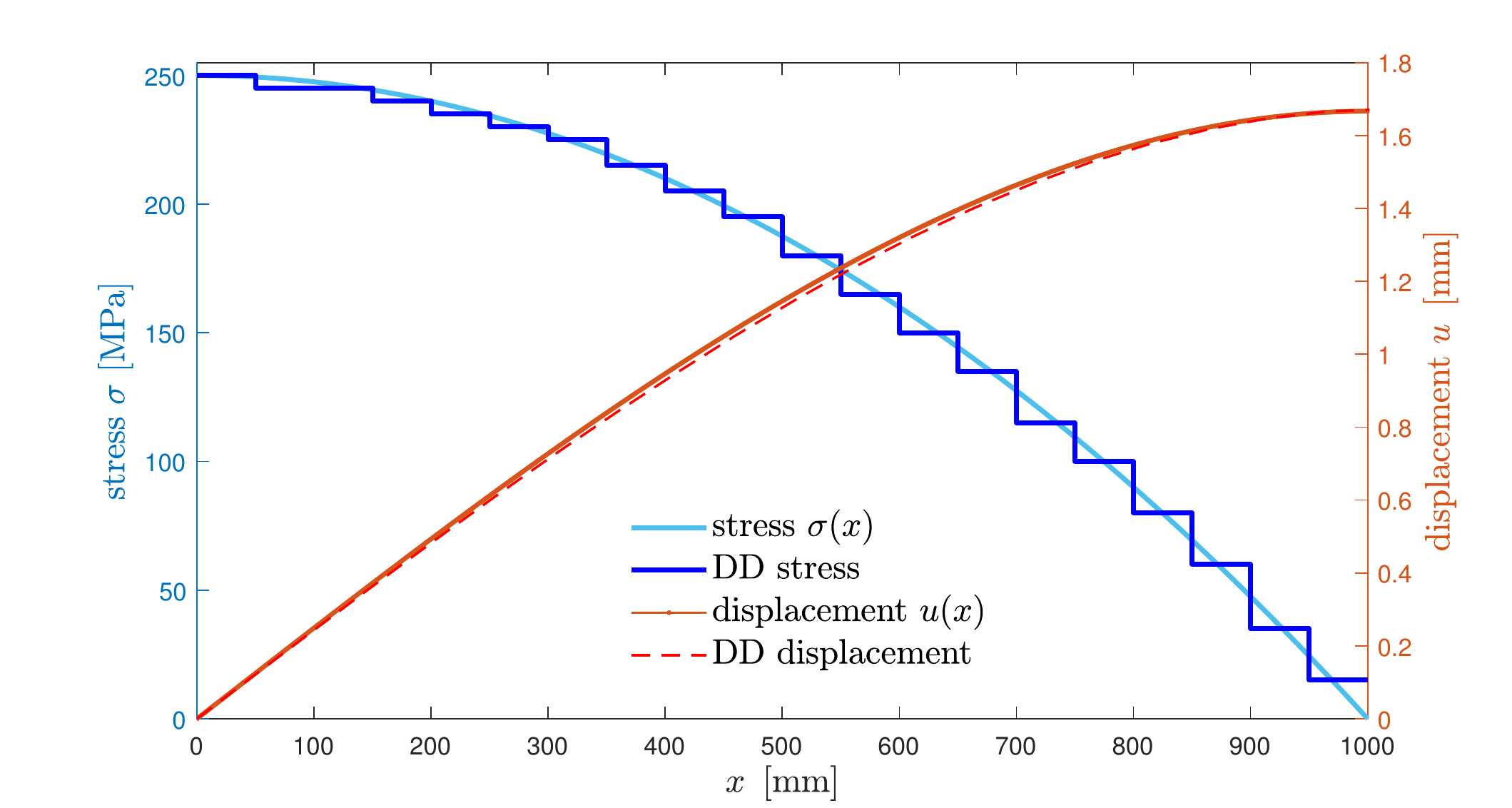}
\caption{Rotating rod problem discretized into $20$ equal linear finite elements. Non-cooperative Data-Driven stress and displacements for
$l=1\,\mathrm{m}$, $A=1$ $\mathrm{mm^2}$, $\mathbb{C}=100\,000\,$MPa, $\rho=20\,\mathrm{g/cm^3}$,
$\omega=$ $0.5 \mathrm{s^{-1}}$, and a local material point-data set $D_e$ covering the range $0$--$250$ MPa with $51$ data points.}\label{fig:stabFliehkraft_sigma_u}
\end{center}
\end{figure}

\begin{figure}[h]
\begin{center}
\includegraphics[width=0.75\linewidth]{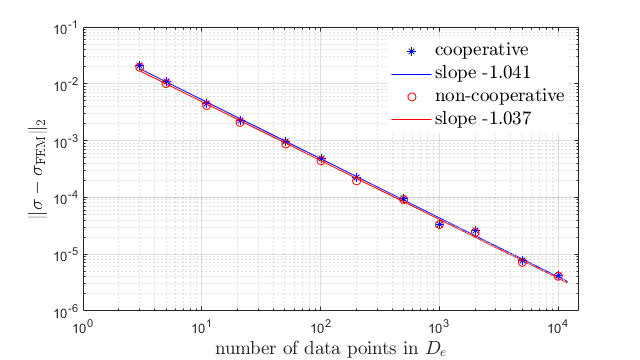}
\caption{Rotating rod problem discretized into $100$ equal linear finite elements. Convergence of the $L^2$--error in the stress field with increasing number of points in the local material data set $D_e$.}\label{fig:rodSpinning_error}
\end{center}
\end{figure}

Figure \ref{fig:stabFliehkraft_sigma_u} shows stress and displacement fields computed for a rod of length $l=1\,\mathrm{m}$, cross section $A=1$ $\mathrm{mm^2}$, Hooke's modulus $\mathbb{C}=100\,000$ MPa, and mass density $\rho=$ $20\,\mathrm{g/cm^3}$ rotating with angular velocity of $\omega=$ $0.5 \mathrm{s^{-1}}$, cf.~Fig.~\ref{fig:spinningRod}a, and a local material point-data set $D_e$ covering the range $0$--$250$ MPa with $51$ data points, cf.~Fig.~\ref{fig:spinningRod}b. In the calculations, the rod is discretized into $20$ equal linear finite elements. As may be seen from the figure, the Data-Driven solution approximates the limiting displacement and stress fields closely, despite the coarseness of the discretization and material sampling. The goodness of the approximation is all the more remarkable given the lack of regularization in the calculations and the attendant discontinuous character of the effective material law, and owes partly to the stabilizing effect of inertia \cite{Bulin:2022}.

Figure \ref{fig:rodSpinning_error} finally shows the $L^2$--error of the cooperative and non-cooperative Data-Driven stress fields as a function of the number of material data points. The results are computed using $100$ equal linear finite elements in the discretization of the rod. Remarkably, both the cooperative and non-cooperative Data-Driven games afford a nearly identical linear range of convergence with respect to material data. Evidently, the convergence of the discrete rod problem with respect to the data observed in the calculations is expected in view of Prop.~\ref{propuniform}, but the calculations additionally supply a precise convergence rate that renders the analysis of convergence quantitative.

\subsubsection{Perforated plate}
\label{RP4nsY}

\begin{figure}[h]
\begin{center}
\includegraphics[width=\linewidth]{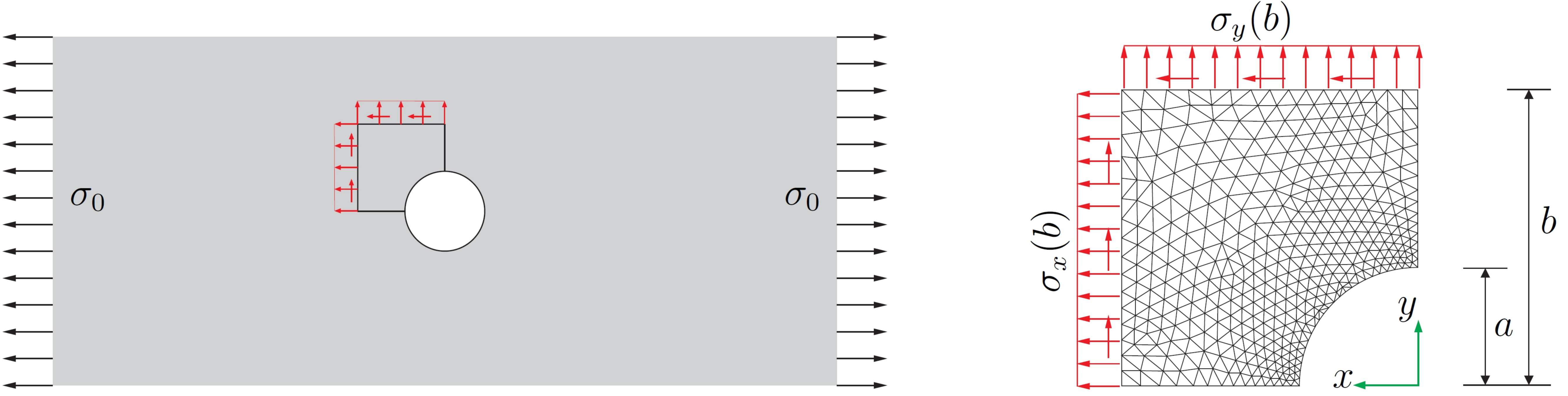}
\caption{a) Plate with circular perforation of radius $a$ under uniform remote tension $\sigma_0$. b) Computational finite element cell of size $b$ loaded by tractions computed from the exact analytical and discretized into linear triangular elements.}\label{fig:lochscheibeModell}
\end{center}
\end{figure}

\begin{figure}
\begin{center}
\includegraphics[width=0.75\linewidth]{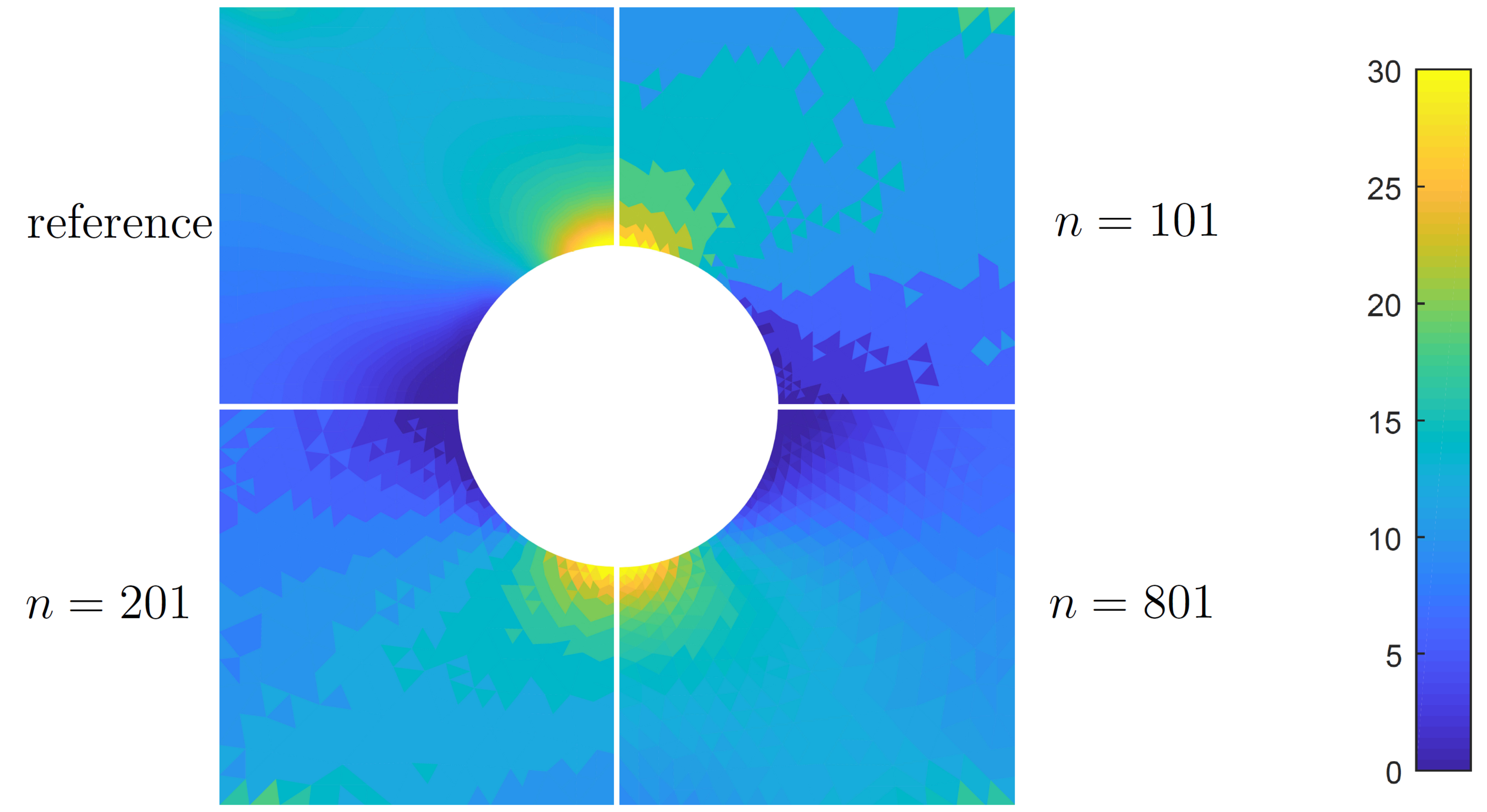}
\caption{Perforated plate problem. Comparison of exact (top left) and Data-Driven horizontal stress component $\sigma_{x}$ for $101^3$, $201^3$ and $801^3$ points in the local material data set.}\label{fig:plate3DataSets}
\end{center}
\end{figure}

As a two-dimensional example, we consider an isotropic plate with a circular perforation of radius $a=200\,$mm deforming under the action of remote tension in the horizontal direction, Fig.~\ref{fig:lochscheibeModell}a.
%The  material point–data sets $D_e$ follow from a linear-elastic material with $E=1000\,$MPa and $\nu=0$.
A quadrant of a square computational cell of size $b=500$ mm is discretized into $720$ linear triangular elements, Fig.~\ref{fig:lochscheibeModell}b.
For linear-elastic material, the exact solution of the problem is known analytically \cite{Sokolnikoff:1956} and the corresponding tractions are applied to the boundary of the computational cell. Local material data sets $D_e$ are generated by sampling without noise Hooke's law with Lam\'e constants $\lambda=0$ and $\mu=500\,$MPa over a regular Cartesian grid spanning a stress range of $\pm 0.1$ to $\pm 100\,$MPa. Choosing a strain-distance discrepancy function without regularization, cf.~Section~\ref{M921rQ} and Example~\ref{y60Z7n2}, the data define a non-cooperative Data-Driven problem that is solved by dynamic relaxation \cite{OakleyKnight1995}. As in the preceding spinning-rod example, the unregularized effective stress functions $\hat{\sigma}_e(\epsilon_e)$ are piecewise constant over the Voronoi cells of the sampled strains.

Fig.~\ref{fig:plate3DataSets} compares the Data-Driven horizontal stress component $\sigma_{x}$ for three material data sets of size $101^3$, $201^3$ and $801^3$, respectively. The exact analytical solution is also shown for comparison. We note that the Data-Driven stresses in the figure are shown exactly as computed at the element level and are not in any way smoothed, which accounts for the patchy look of the figure. Despite the non-informative nature of the material data, a general trend towards convergence of the Data-Driven solutions is evident from the figure, as expected from Prop.~\ref{propuniform}.

\subsubsection{Torsion of a cube}
\label{ssec:cube}

As a final three-dimensional example, we consider the torsion of a cube whose bottom face is clamped while a planar rotation is imposed on the top face. The cube is meshed using $4635$ linear tetrahedral elements. Horizontal ($x-y$ plane) displacements of nodes located on the top face are prescribed corresponding to a rotation of 10$^{-3}\pi$ rad around the central axis. Vertical displacements are left free. We further consider a material governed by the non-linear elastic relation
\begin{equation}\label{eq:NLelast}
    \sigma
    =
    K\left(1+\Tr[\epsilon]^2\right)\Tr[\epsilon]\delta
    +
    G(1+\Dev[\epsilon]:\Dev[\epsilon]) \Dev[\epsilon] .
\end{equation}
Data sets were generated by sampling this relation with Young modulus of $10^{11}$ Pa and a Poisson ratio of 0.35 over the intervals $[-0.004,0.004]$ for $\epsilon_{xx}$, $\gamma_{xy}$, $\epsilon_{yy}$, $\epsilon_{zz}$, and $[-0.008,0.008]$ for $\gamma_{xz}$ and $\gamma_{yz}$, cf.~Table \ref{tab:3d-sampling}, yielding data sets containing $75625$, $405769$ and $1476225$ pairs of stress and strain tensors, respectively.

\begin{table}[hbt]
    \centering
    \begin{tabular}{|r|c|c|c|c|c|c|}
      \hline
         \# sampling points & $\epsilon_{xx}$ & $\gamma_{xy}$ & $\epsilon_{yy}$ & $\epsilon_{zz}$ & $\gamma_{xz}$ & $\gamma_{yz}$ \\ \hline
         75625 ($D_1$) & 5 & 5 & 5 & 5 & 11 & 11 \\
         405769 ($D_2$) & 7 & 7 & 7 & 7 & 13 & 13 \\
         1476225 ($D_3$) & 9 & 9 & 9 & 9 & 15 & 15 \\
      \hline
    \end{tabular}
    \caption{3D strain tensor sampling schemes.}
    \label{tab:3d-sampling}
\end{table}

\begin{figure}
    \centering
    \includegraphics[width=0.475\linewidth]{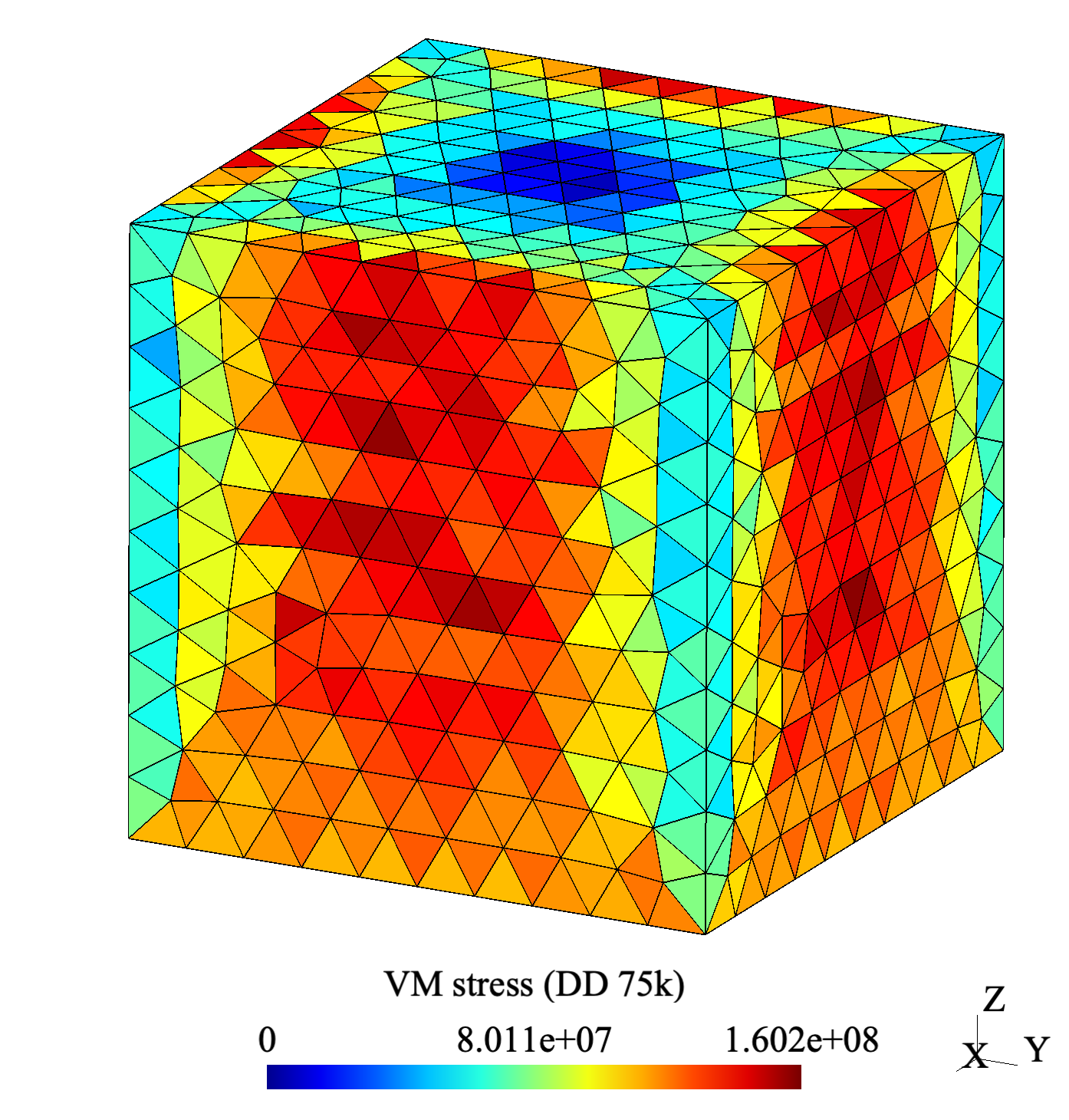}
    \includegraphics[width=0.475\linewidth]{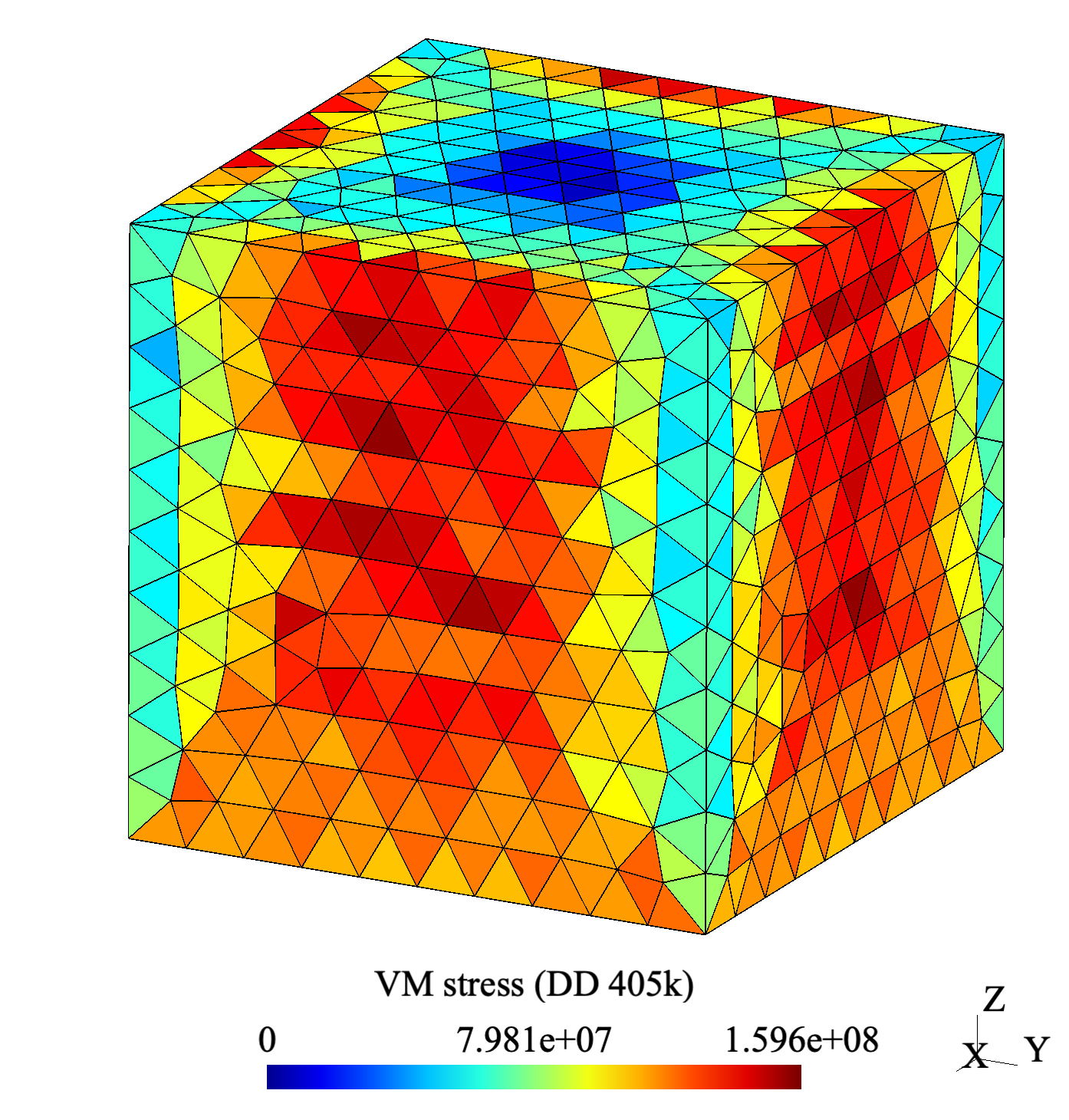}\\
    \includegraphics[width=0.475\linewidth]{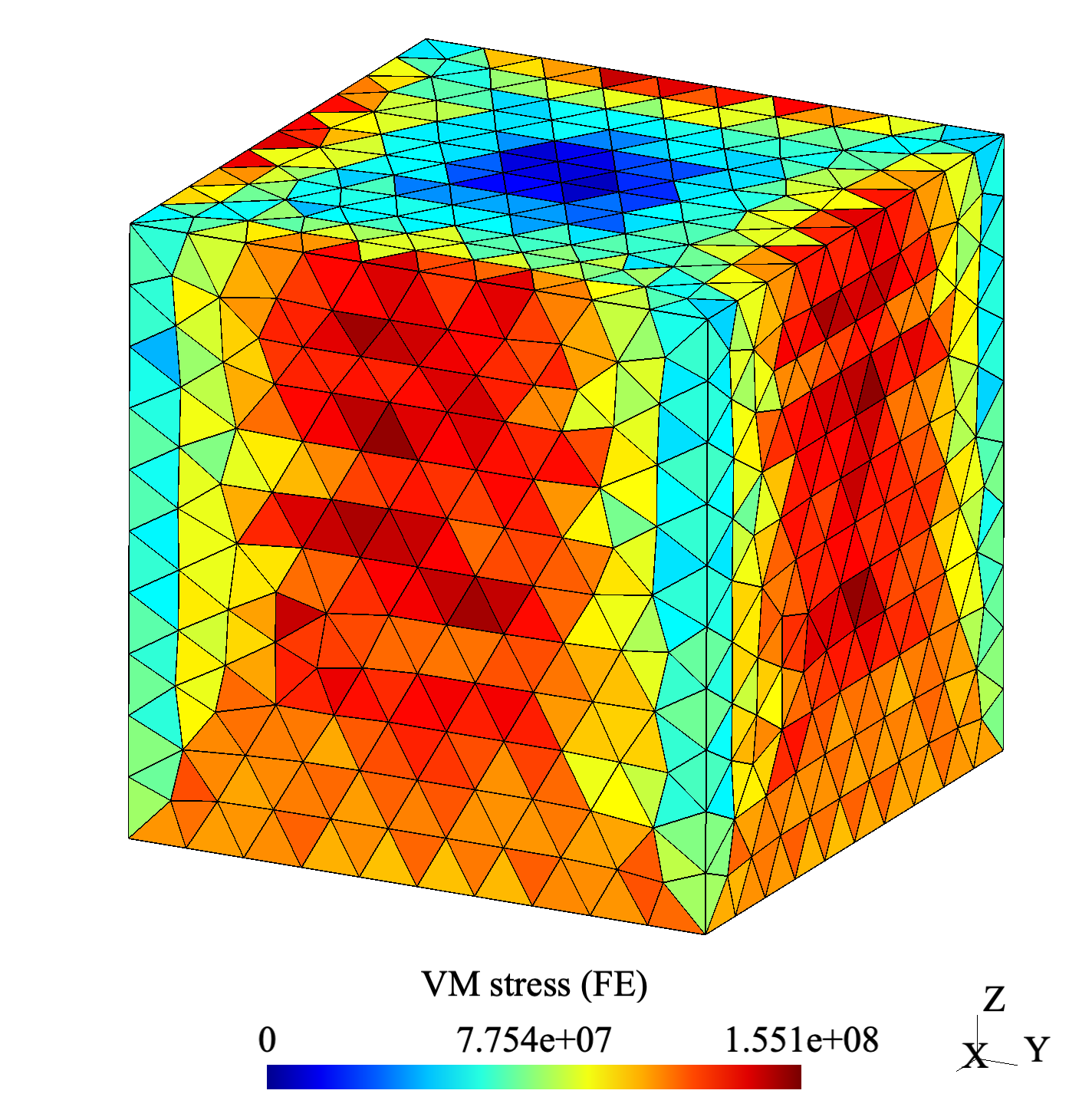}
    \includegraphics[width=0.475\linewidth]{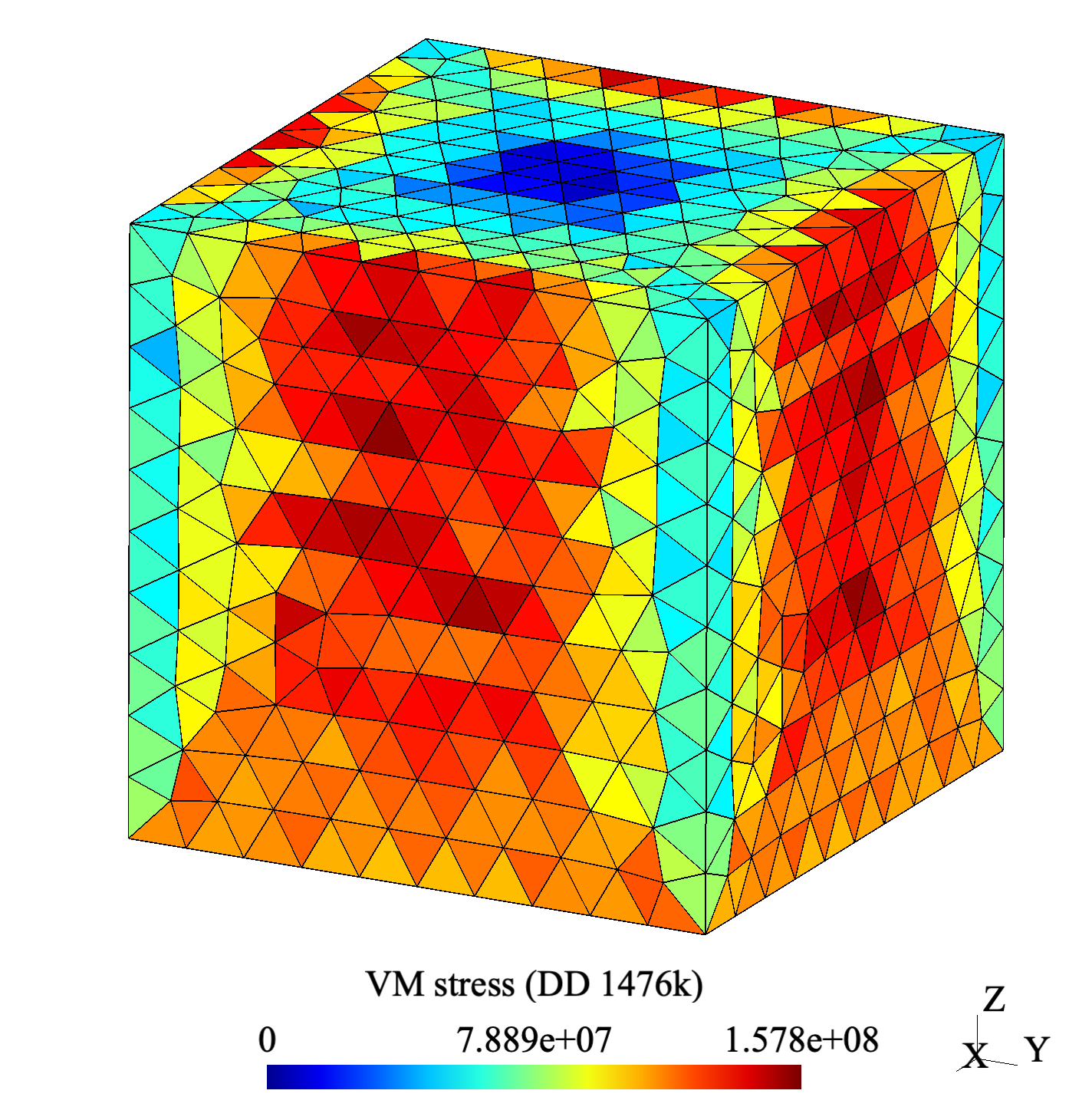}
    \caption{Torsion of a non-linear elastic cube, equivalent von-Mises strain. Clockwise: Data-Driven solutions for material data sets of sizes $75625$, $405769$ and $1476225$, and reference finite-element solution.}
    \label{fig:cube3d-VM}
\end{figure}

These data sets were used to solve the problem using a regularized Data-Driven formulation, cf.~Section~\ref{T5RuNG}. The smoothness afforded by the regularization allows to use non-linear solvers such as Newton-Raphson or iterative gradient descent algorithms. In calculations, we use a Krylov-GMRES algorithm from SciPy 1.10.1 \cite{SciPy-NMeth}), initially preconditioned by the tangent computed at zero strain. The solution is then obtained in 4 or 5 Krylov iterations. Note that, in order to obtain this performance, the regularization parameter $\beta$ needs to be carefully chosen. In the calculations, we specifically use values of $\beta$ ranging from $1.2\times 10^{-6}$ to $2.4\times 10^{-6}$, increasing with data density, in combination with a distance defined by an isotropic elasticity tensor computed from a Young modulus of 10$^{12}$ Pa and a Poisson ration of 0.3 (chose different from the linearized elastic response at zero strain to avoid any bias).

\begin{figure}
    \centering
    \includegraphics[width=0.475\linewidth]{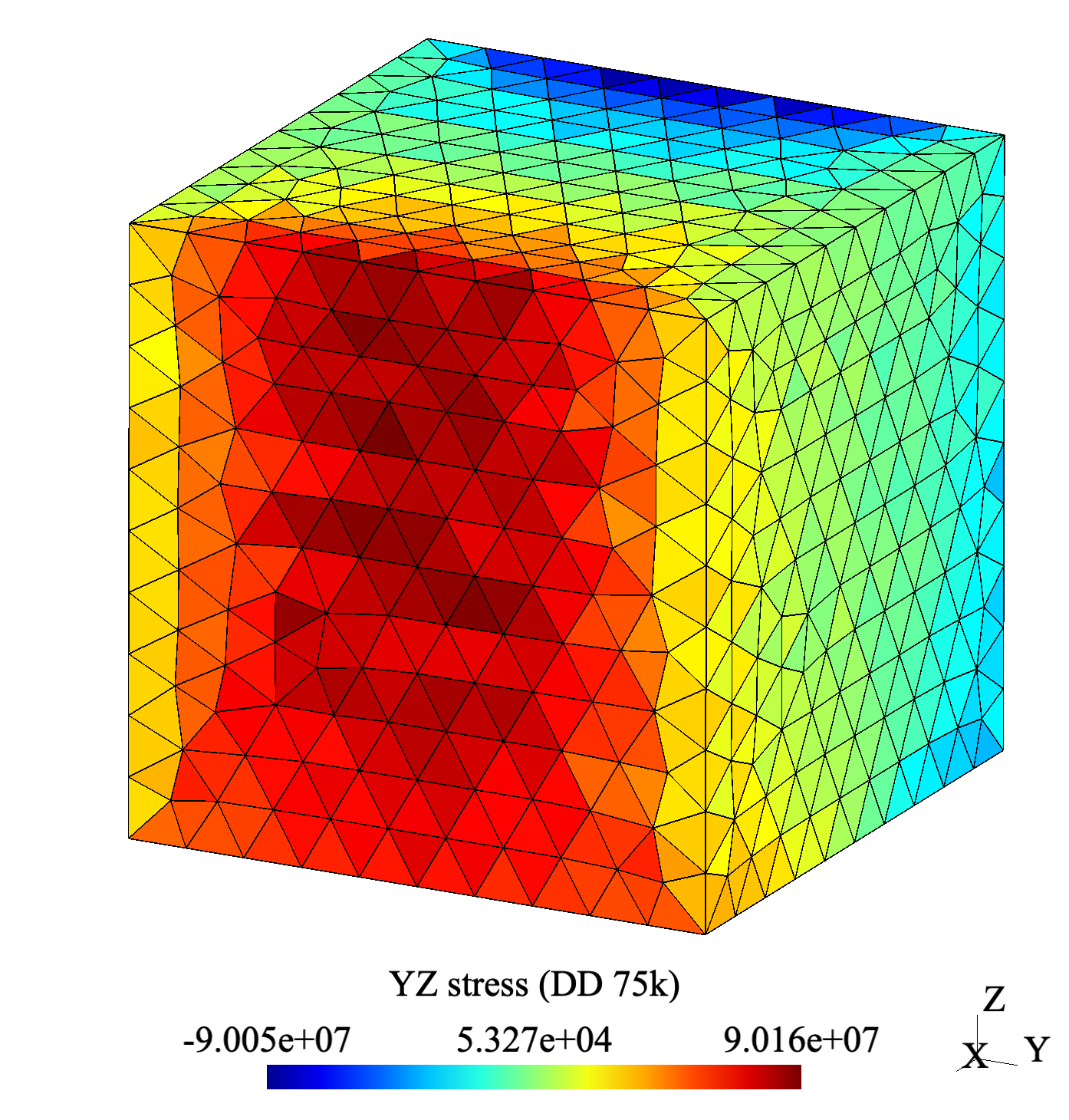}
    \includegraphics[width=0.475\linewidth]{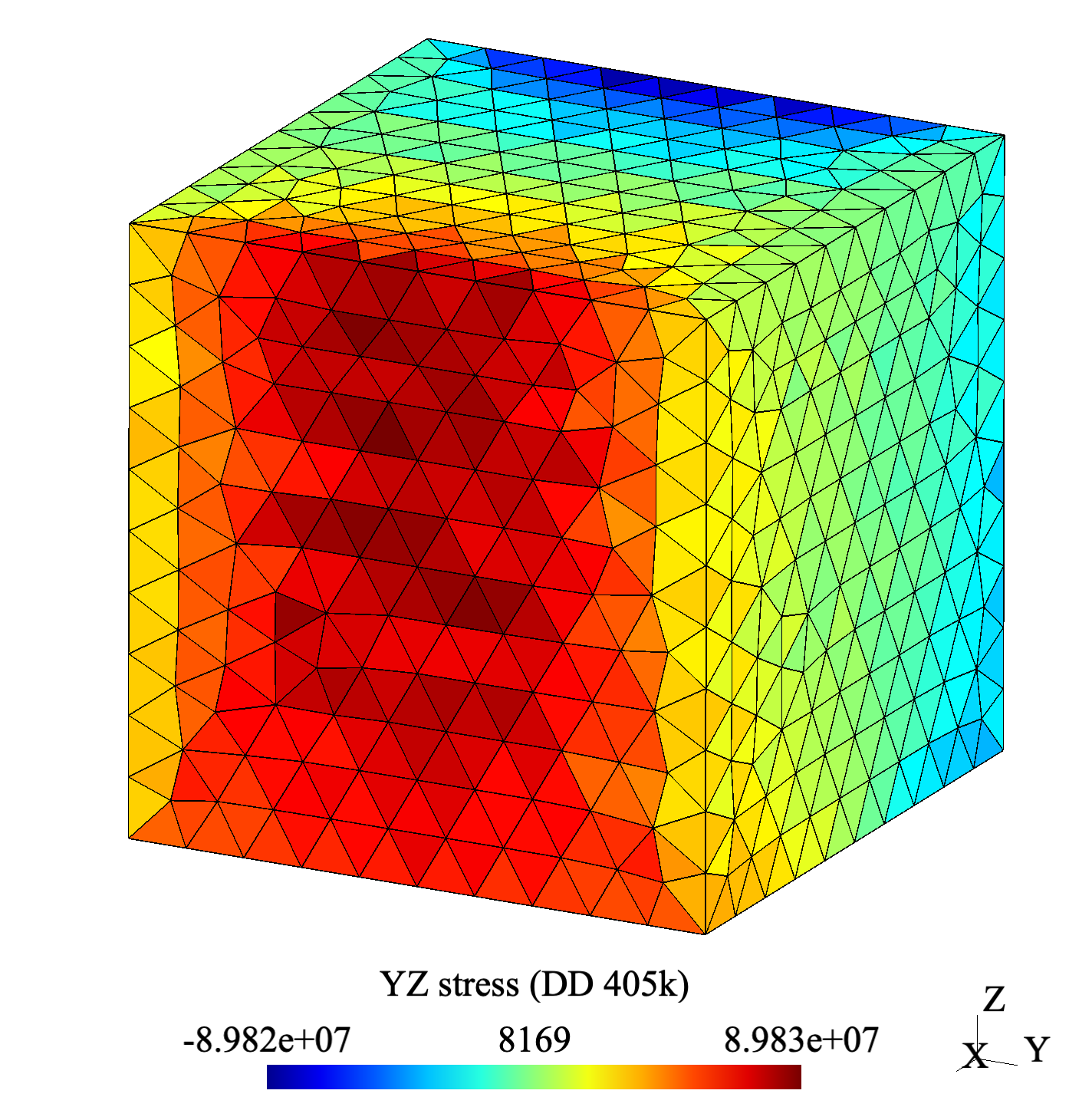}\\
    \includegraphics[width=0.475\linewidth]{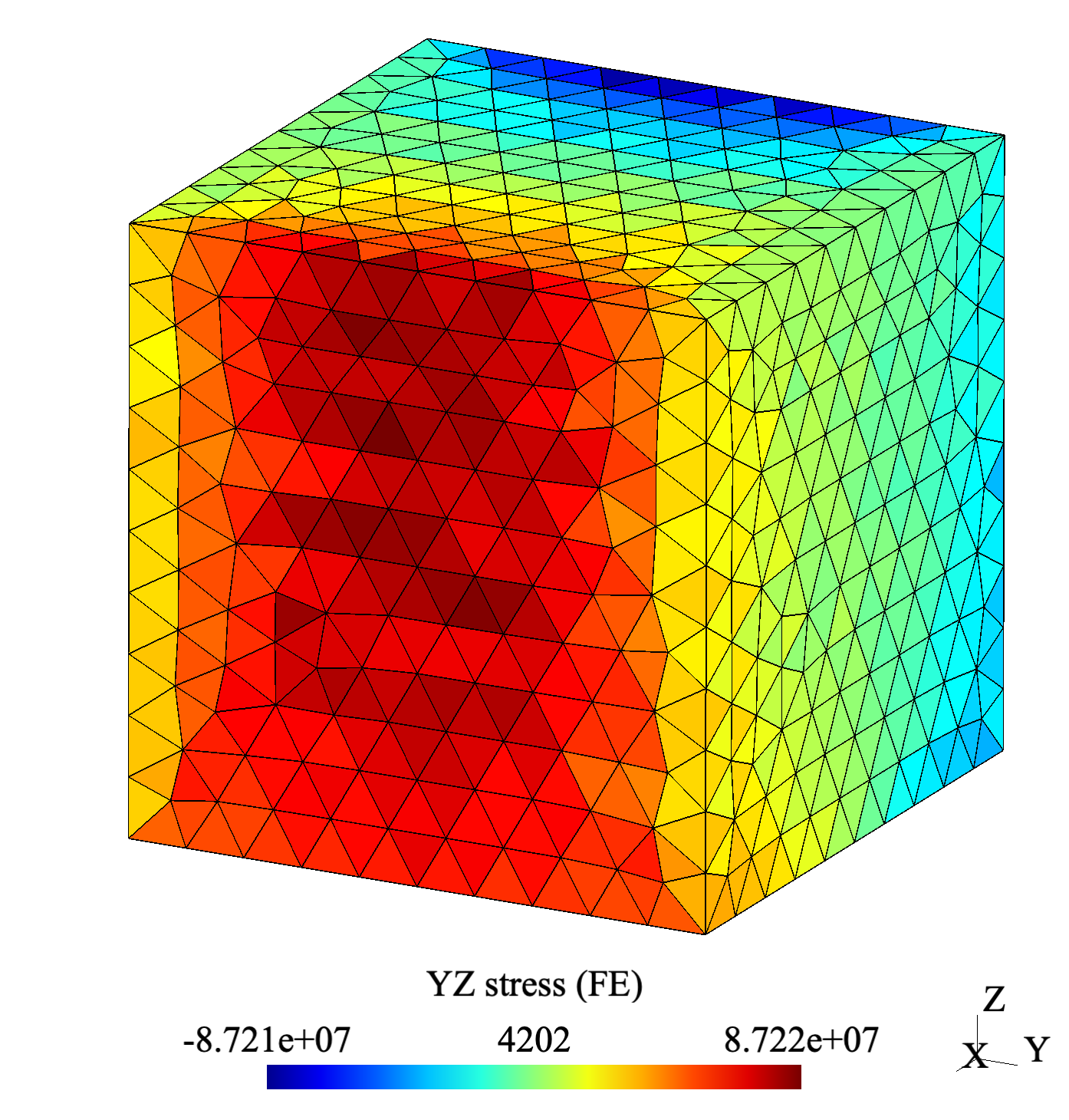}
    \includegraphics[width=0.475\linewidth]{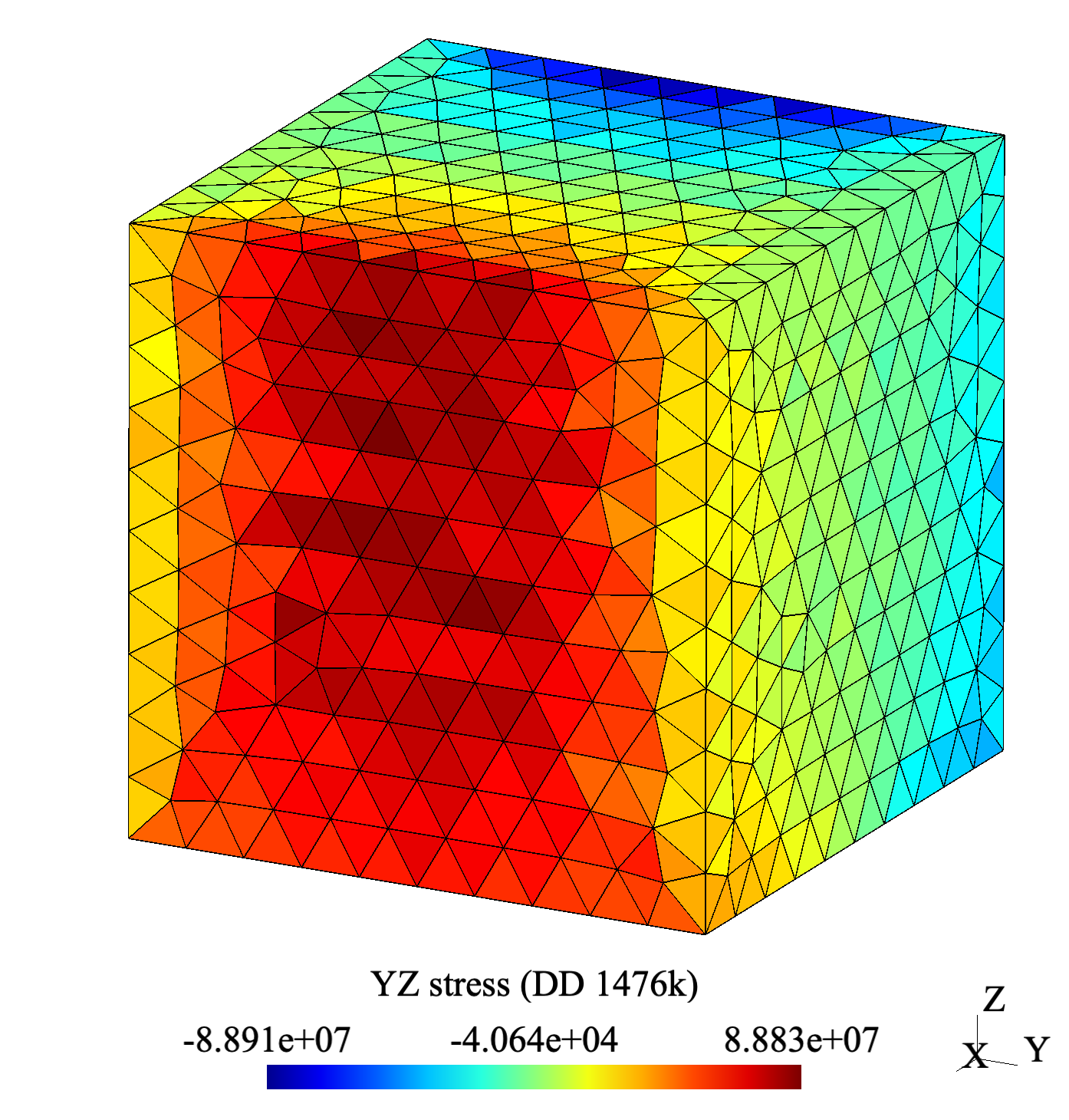}
    \caption{Torsion of a non-linear elastic cube, $\sigma_{yz}$ stress. Clockwise: Data-Driven solutions for material data sets of sizes $75625$, $405769$ and $1476225$, and reference finite-element solution.}
    \label{fig:cube3d-YZ}
\end{figure}

The resulting Data-Driven von Mises shear strain and $yz$ shear stresses for the three material data sets are shown in Figs.~\ref{fig:cube3d-VM} and \ref{fig:cube3d-YZ}, respectively. A standard non-linear finite element solution computed from the limiting material law (\ref{eq:NLelast}) is also shown in the figures for comparison. A general trends towards convergence towards the reference solution can be observed in the figures, with stress errors under 1\% for the largest data set.

\bigskip

\section{Summary and concluding remarks}
\label{I41hqu}

We have resorted to game-theoretical concepts to formulate Data-Driven methods for solid mechanics in which stress and strain are adversarial players and pursue different--and competing--objectives: the objective of the stress player is to minimize the discrepancy to a material data set that characterizes material behavior; the objective of the strain player is to ensure the admissibility of the mechanical state, in the sense of satisfying compatibility and equilibrium. The main properties of the proposed non-cooperative Data-Driven games are:

\begin{enumerate}
\item Unlike the cooperative Data-Driven games proposed in the past \cite{Kirchdoerfer:2016, Kirchdoerfer:2017, conti:2018, Prume:2023, ContiHoffmannOrtiz2023}, the new non-cooperative Data-Driven games identify an effective material law from the material data.
\item The data from which the effective material law is identified is {\sl fundamental}, i.~e., a set of observed material states in stress-strain, or {\sl phase}, space. In particular, the material data is not subordinate to---or parameterized in terms of---a particular, {\sl a priori} surmised, material model.
\item The effective material law follows directly from the material data, {\sl all} the material data, and {\sl nothing but} the material data, together with material-independent notions of distance and proximity in stress-strain space.
\item At no time during the calculations---or during the evaluation of the effective material law---is the data set replaced by an {\sl ad hoc} parameterized funcion, e.~g., a neural network, fitted to the data, which inevitably results in biasing and loss of information.
\item The proposed non-cooperative Data-Driven games reduce to conventional displacement boundary-value problems in which the material is characterized by the effective material law.
\item The displacement boundary-value problems set forth by the proposed non-cooperative Data-Driven games can be implemented within a standard finite-element framework in which the effective material law is evaluated at each material or quadrature point, e.~g., through a user-supplied material law.
\item For structurally stable systems, the convergence of the Data-Driven solutions with respect the material data can be analyzed and ascertained entirely at the local material point level in terms of the convergence of the effective material law.
\end{enumerate}

We have presented rigorous analyses that elucidate sufficient conditions for convergence of the Data-Driven solutions with respect to the material data in two scenarios: i) {\sl Uniformly convergent data}, in which the sampling error decreases as data is added to the material-data set in a uniform manner controlled by strict upper bounds, and ii) {\sl noisy data with outliers}, in which the data concentrates around the limiting material law in a weak or average sense that allows for the presence of outliers. These results supply conditions under which the proposed non-cooperative Data-Driven games and well-posed in the sense of convergence of the solutions with respect to the material data to a limiting material law.

We have also presented selected examples of implementation and application that demonstrate how the proposed non-cooperative Data-Driven games can be conveniently combined with the finite-element method for solving specific boundary-value problems. When the material data is sampled from an underlying 'true' material law, presumed unknown, we find that the Data-Driven solutions exhibit accuracy comparable to that of a reference finite-element solution formulated in terms of the underlying material law, even for relatively sparse and non-adaptive material sets.

In closing we note that, while not based on regression, e.~g., to neural networks, the proposed approach may be considered a form of unsupervised, set-oriented, {\sl machine learning}, in which the objective is to learn the structure of the material data set. A number of issues pertaining to this viewpoint arise immediately, among which are the following.

\underline{Data structures and fast searches}.
Evidently, for large material-data sets the speed of the search algorithms used to interrogate the material data sets while evaluating the effective material law is a computational bottleneck and requires careful attention in order to ensure adequate performance. For regularized Data-Driven problems, the rapid decay of exponential functions appearing in \eqref{xKSL3U} can be exploited to limit evaluations to nearest-neighbours. These local neighborhoods can be searched for in large data sets by means of efficient techniques such as kd-tree data structures or approximate nearest neighbours (ANN) algorithms. A review of recent developments and an assessment of a number of search algorithms in the context of Data-Drivem methods can be found in \cite{Eggersmann:2021}.

\underline{Data sampling and adaptive learning}.
A central issue in data science concerns the efficient sampling of the data. In general, the dimensionality stress-strain or phase space is too large to allow for unstructured uniform sampling of the material behavior and more finely-tuned techniques need to be adopted in which the sampling is adapted to particular problems. One such technique is {\sl active} or {\sl adaptive learning} \cite{Settles:2012}. In this setting, a coarse material data set covers the regions of phase space which are relevant to a specific problem of interest. In subsequent iterations, the material data set is augmented with additional data sampled in regions covered by the solution to achieve increased accuracy. In addition, non-informative data far from the phase-space set covered by the solution are excluded from the data set in order to speed up subsequent iterations. Further details of implementation can be found in \cite{Korzeniowski:2021a} and applications of adaptive learning to Data-Driven problems can be found in \cite{Korzeniowski:2021b, Bulin:2022, Gorgogianni:2023}.

\section*{Acknowledgements}

% SC and MO are grateful for support from the Deutsche Forschungsgemeinschaft (DFG) under priority program  SPP 2256.
KW gratefully acknowledges support from the DFG through project WE2525/14-1 as part of the priority program SPP 2256 (no. 422730790). %SPP 1886 %(no. 273721697).
MO gratefully acknowledges additional support from the DFG through project 390685813--GZ 2047/1--HCM. LS gratefully acknowledges the financial support of the French Agence Nationale de la Recherche (ANR) through project ANR-19-CE46-0012 within the French-German Collaboration for Joint Projects in Natural, Life and Engineering (NLE) Sciences, as well as NExT ISite program of Nantes Universit\'e through International Research Project (IRP) iDDrEAM.

\appendix
\section{Approximation via point sets}
\label{9tpjYN}
\renewcommand{\thesection}{\Alph{section}}

In this appendix, we provide rigorous statements and proofs of the propositions enunciated heuristically in Section~\ref{BR7thb}.

\subsection{Uniform approximation}
\label{secuniformapprox}

We revisit the scenario considered in Section~\ref{dQi63n}.  In order to make it possible to work with finite sets $D_{e,h}$, we use a cutoff procedure with extrapolation to infinity. We first define a map $f_{e,h}:\R^d\to\R^d$, and then regularize it to obtain $\hat\sigma_{e,h}$. Fix a large number $R_h>0$,  which will diverge in the limit. Working with the distance-based discrepancy function, for bounded strains $\|\epsilon_e\|_e<R_h$ we consider the stress corresponding to the closest pair in $D_{e,h}$, in the sense that $s_{e,h}(\epsilon_e)=\eta_e$, with $\eta_e$ defined as in \eqref{kP26d1b}. For large strains $\|\epsilon_e\|_e\ge R_h$, we instead set
\begin{equation}\label{eqdeffehtail}
    s_{e,h}(\epsilon_e)=\C_e \epsilon_e
\end{equation}
where $\C_e\in \R^{d\times d}$ is an arbitrary positive definite matrix. The result does not depend on the choice of $\C_e$. We then select a small regularization distance $\delta_h>0$, and define
\begin{equation}\label{eqdefhatsigmaeh}
    \hat\sigma_{e,h}(\epsilon_e)= (\psi_{\delta_h}\ast s_{e,h})(\epsilon_e)
\end{equation}
where $\psi_\delta\in C^\infty_c(B_\delta;[0,\infty))$ is an even mollification kernel  on scale $\delta$.
This definition ensures continuity of $\hat\sigma_{e,h}$ (whereas the function $f_{e,h}$ is not continuous, even for small strains). Further, it only requires sampling the data set in a bounded region, and therefore a good approximation can be obtained with finite sets $D_{e,h}$.
We shall show that if $D_{e,h}$ is a good approximation of the graph of $\hat\sigma_e$ then the solutions converge.

We quantify the goodness of the approximation via three separate parameters. First, $t_h\to0$ gives the scale of the uniformity of the approximation, in the sense that no points contains an error larger than $t_h$. Second,
$R_h\to\infty$ gives the size of the region which is covered by the approximation. Third,
$\rho_h\to0$ quantifies the fineness of the approximation, in the sense that any strain smaller than $R_h$ has been tested, up to a precision $\rho_h$. Finally, $\delta_h\to0$ is the scale of the mollification in \eqref{eqdefhatsigmaeh}.

\begin{prop}[Uniform approximation]\label{propuniform}
Let $w_e > 0$ and let $\hat{\sigma}_e : \mathbb{R}^d \to \mathbb{R}^d$ be continuous functions
which obey material stability, in the sense of
\eqref{eqsigmastable}, and which are locally Lipschitz continuous, in the sense that for any $R>0$ there is $L_R>0$ such that
\begin{equation}\label{eqhatsigmaloclip}
    \|\hat\sigma_e(\epsilon)-\hat\sigma_e(\epsilon')\|_e
    \le
    L_R\|\epsilon-\epsilon'\|_e
    \text{ for all }
    \epsilon,\epsilon'\text{ with } \|\epsilon\|_e<R, \|\epsilon'\|_e<R.
\end{equation}
% \SC{Alternative: $\hat\sigma_e\in C^1$ and $\|D\hat\sigma_e\|_{B_R}\le L_R$}.
Let $B : \mathbb{R}^n \to \mathbb{R}^{m d}$ obey structural stability, in the sense of \eqref{eqBstructurstable}.

For every $h\in\N$, consider discrete set of points $D_{e,h}\subseteq Z_e$.
Assume that each $D_{e,h}$ is a locally uniform approximation, in the sense that there are sequences $\paramunifh\to0$, $\paramfineh\to0$, $\delta_h\to0$ and $R_h\to\infty$
such that
\begin{enumerate}
 \item\label{propuniformge} $D_{e,h}\subseteq (G_e)_{\paramunifh}$,
 meaning that for any $y_e\in D_{e,h}$ there is $z_e\in G_e$ with $\|z_e-y_e\|_e<\paramunifh$, where $G_e$ is the graph of $\hat\sigma_e$, defined as in \eqref{eqdefgraph};
 \item\label{propuniformeps} for every $\epsilon_e\in Z_e$ with $\|\epsilon_e\|_e\le 2R_h$ there is  $y_e\in D_{e,h}$
 with $\|(\epsilon_e,\hat\sigma_e(\epsilon_e))-y_e\|_e<\paramfineh$.
\item\label{propuniformlim}
$\lim_{h\to\infty} L_{2R_h}(\paramunifh+\paramfineh+\delta_h)=0$.
 \end{enumerate}
Let $\hat\sigma_{e,h}$ be defined as in \eqref{eqdefhatsigmaeh} above, and let $\hat\sigma_h=(\hat\sigma_{1,h},\dots,\hat\sigma_{N,h})$.

Then for sufficiently large $h$, for any
$f \in \mathbb{R}^n$ the problem
\begin{equation}\label{wpL8eAh}
    w B^T \hat{\sigma}_h(B u_h) = f
\end{equation}
has a solution $u_h$, and  up to subsequences the solutions $u_h$ converge to a solution $u$ of the continuous problem \eqref{wpL8eA}. If the $\hat\sigma_e$ are strictly monotone in the sense of~\eqref{eqsigmastrictmonotone}
then the solution $u$ of the limiting problem is unique and the entire sequence $u_h$ converges to $u$.
\end{prop}

\begin{proof}
We first show that, for sufficiently large $h$,
the functions $\hat\sigma_{e,h}$ obey the material stability condition \eqref{eqsigmastable} uniformly. As they are continuous  by construction, and structural stability holds, existence of solutions $u_h$ for each $h$ will then follow from Proposition~\ref{8CJa9S}.

Fix $\epsilon_e$, and consider some $\epsilon_e'$ with
$\|\epsilon_e'-\epsilon_e\|_e<\delta_h$. We distinguish two cases. Assume first
$\|\epsilon_e'\|_e\ge R_h$, then
\begin{equation}\label{eqfehoutside}
\begin{split}
 s_{e,h}(\epsilon_e')\cdot \epsilon_e
 = &\C_e\epsilon_e'\cdot \epsilon_e
 \ge \C_e\epsilon_e\cdot \epsilon_e -\|\C_e\| \,\|\epsilon_e\|_e\,\|\epsilon_e'-\epsilon_e\|_e\\
 \ge&  \hat a\|\epsilon_e\|_e^2 -\|\C_e\| \,\|\epsilon_e\|_e\delta_h\ge
\frac12  \hat a\|\epsilon_e\|_e^2 -\frac12 \hat a
 \end{split}
\end{equation}
where in the last step we assumed that $h$ is sufficiently large that $\|\C_e\| \delta_h\le \hat a$ and used $2x\le x^2+1$.

Assume now that $\|\epsilon_e'\|_e<R_h$.
By definition of $s_{e,h}$, there is $\xi_e'$
such that $(\xi_e', s_{e,h}(\epsilon_e'))\in D_{e,h}$ and $\xi_e'$ is (one of the) points which minimizes the distance to $\epsilon_e'$.

By \ref{propuniformeps}  there is $(\xi_e,\eta_e) \in D_{e,h}$
such that $\|(\xi_e,\eta_e) -(\epsilon_e,\hat\sigma_e(\epsilon_e))\|_e<\paramfineh$,
so that by minimality
\begin{equation}
\|\xi_e'-\epsilon_e'\|_e\le \|\xi_e-\epsilon_e'\|_e
\le \|\xi_e-\epsilon_e\|_e+\|\epsilon_e-\epsilon_e'\|_e<\paramfineh+\delta_h.
\end{equation}
By \ref{propuniformge}, there is $\xi_e''$ such that
\begin{equation}\label{eqrhohas}
 \| (\xi_e'',\hat\sigma_e(\xi_e''))-(\xi_e', s_{e,h}(\epsilon_e'))\|_e
 <\paramunifh.
\end{equation}
Combining the two,
\begin{equation}
\|\xi_e''-\epsilon_e'\|_e\le
 \|\xi_e''-\xi_e'\|_e
 +\|\xi_e'-\epsilon_e'\|_e\le
 \paramunifh+\paramfineh+\delta_h
\end{equation}
and
\begin{equation}
 \|\xi_e''-\epsilon_e\|_e\le
 \|\xi_e''-\epsilon_e'\|_e
 +\|\epsilon_e'-\epsilon_e\|_e\le
 \paramunifh+\paramfineh+2\delta_h.
\end{equation}
In particular, $\|\xi_e''\|_e\le R_h+\paramunifh+\paramfineh+\delta_h\le 2R_h$,
and the same for $\epsilon_e$ (for $h$ sufficiently large).
Therefore, using \eqref{eqrhohas} and \eqref{eqhatsigmaloclip},
\begin{equation}\label{eqdistfehhatsig}
\begin{split}
\|s_{e,h}(\epsilon_e')-\hat\sigma_e(\epsilon_e)\|_e\le&
 \|s_{e,h}(\epsilon_e')-\hat\sigma_e(\xi_e'')\|_e
+ \|\hat\sigma_e(\xi_e'')-\hat\sigma_e(\epsilon_e)\|_e\\
\le&
\paramunifh
+ L_{2R_h} \|\xi_e''-\epsilon_e\|_e
\le
(1+L_{2R_h})(\paramunifh+\paramfineh+2\delta_h).
\end{split}
\end{equation}
In particular, this implies
\begin{equation}\begin{split}
s_{e,h}(\epsilon_e')
\cdot\epsilon_e\ge &
\hat\sigma_e(\epsilon_e)\cdot\epsilon_e
-(1+L_{2R_h})(\paramunifh +\paramfineh+2\delta_h) \|\epsilon_e\|_e\\
\ge & a \|\epsilon_e\|^2-b-\sqrt{ab} \|\epsilon_e\|_e\\
\ge &
\frac12 a \|\epsilon_e\|_e^2-2b.
\end{split}
\end{equation}
where in the second step we used
\eqref{eqsigmastable}, and
assumed that $h$ is sufficiently large that
$(1+L_{2R_h})(\paramunifh +\paramfineh+2\delta_h)\le \sqrt{ab}$ (this is possible by assumption \ref{propuniformlim}).
Recalling \eqref{eqfehoutside}, the same estimate holds for $\|\epsilon_e'\|_e\ge R_h$, possibly with different values of the constants $a$ and $b$. Therefore it holds for all $\epsilon_e'$ with $\|\epsilon_e-\epsilon_e'\|_e<\delta_h$, and by convexity of the mollification we conclude
\begin{equation}
 \hat\sigma_{e,h}(\epsilon_e)\cdot\epsilon_e \ge a'\|\epsilon_e\|_e^2-b'
\end{equation}
for all $\epsilon_e$, with $a':=\min\{\frac12 a,\frac12\hat a\}$ and
$b':=\max\{2b,\frac12\hat a\}$.
 This proves material stability and hence existence of solutions for each fixed $h$.

Finally, from \eqref{eqdistfehhatsig} we obtain locally uniform convergence.
In particular, fix $M>0$. For $h$ sufficiently large we have $M+\delta_h<R_h$. Therefore
\eqref{eqdistfehhatsig} holds for all pairs
$\epsilon_e$, $\epsilon_e'$ with
$\|\epsilon_e\|_e<M$ and
$\|\epsilon_e-\epsilon_e'\|_e<\delta_h$. Again, by convexity
\begin{equation}
\begin{split}
 \|\hat\sigma_{e,h}(\epsilon_e)-\hat\sigma_e(\epsilon_e)\|_e
\le
(1+L_{2R_h})(\paramunifh+\paramfineh+2\delta_h)
\end{split}
\end{equation}
for all
$\epsilon_e$  with
$\|\epsilon_e\|_e<M$.
By \ref{propuniformlim}, we obtain uniform convergence. The rest follows from Proposition~\ref{abAx6h}.
\end{proof}

\subsection{Noisy data with outliers}
\label{secappappoutl}

We revisit the scenario considered in Section~\ref{cbb7Jh}.  As in the preceding section, in order to make it possible to work with finite sets $D_{e,h}$, we use a cutoff at infinity. Precisely, we fix a function $\psi\in C^1_c(B_1;[0,1])$ with $\psi=1$ on $B_{1/2}$ and set
\begin{equation}\label{eqdefhatsigmahweigts}
 \hat\sigma_{e,h}(\epsilon_e)
 := \psi\left(\frac1{R_h}\|\epsilon_e\|_e\right)
 \sigma^*_{e,h}(\epsilon_e)
 + \left(1-\psi\left(\frac1{R_h}\|\epsilon_e\|_e\right)\right)
 \C_e \epsilon_e
\end{equation}
where $\C_e$ is a fixed positive definite matrix as in \eqref{eqdeffehtail}
and
\begin{equation}\label{eqdefsigmaehapp}
 \sigma_{e,h}^*(\epsilon_e)
    =
    \sum_{i=1}^{N_e}
    p_{e,h,i}^*(\epsilon_e;\beta_h) \sigma_{e,h,i} .
\end{equation}
Also in this case, we shall show that the limit does not depend on the choice of $\C_e$. The proof uses some ideas from \cite{ContiHoffmannOrtiz2023, ContiHoffmannOrtizb}.

\begin{prop}[Discrete approximation with outliers]\label{propoutliers}
Let $w_e > 0$ and let $\hat{\sigma}_e : \mathbb{R}^d \to \mathbb{R}^d$ be continuous functions which obey material stability, in the sense of \eqref{eqsigmastable}, and which are locally Lipschitz, in the sense that for any $R>0$ there is $L_R>0$ such that \eqref{eqhatsigmaloclip} holds. Let $B : \mathbb{R}^n \to \mathbb{R}^{m d}$ obey structural stability, in the sense of \eqref{eqBstructurstable}.

For every $h\in\N$, consider a finite set of data points $D_{e,h}\subseteq Z_e$.
Assume that $D_{e,h}$ is a good approximation with few outliers, in the sense that there are sequences $\paramunifh$, $\paramfineh$, $\delta_h\to0$ and $C_h$, $N_h$, $M_h$, $R_h$, $R_h^*\to\infty$
such that
\begin{enumerate}
 \item\label{propoutliersout}
 For every $\epsilon_e$ with $\|\epsilon_e\|_e<R_h$ one has
 \begin{equation}
  \# \{ (\xi_e,\eta_e)\in D_{e,h}:
\xi_e\in  \epsilon_e + [0,\paramfineh )^d, \|\eta_e-\hat\sigma_e(\epsilon_e)\|_e>\paramunifh \} < M_h;
 \end{equation}
 \item\label{propoutliersunif}
 For every $\epsilon_e$  one has
 \begin{equation}
  \# \{ (\xi_e,\eta_e)\in D_{e,h}:
\xi_e\in  \epsilon_e + [0,\paramfineh )^d\} < C_h;
 \end{equation}
 \item\label{propoutliersbound} For any $y_e\in D_{e,h}$ one has
 $\|y_e\|_e\le R_h^*$;
 \item\label{propoutliersdenst}
  For every $\epsilon_e$ with $\|\epsilon_e\|_e<2R_h$ one has
 \begin{equation}
  \# \{ (\xi_e,\eta_e)\in D_{e,h}: \xi_e
  \in \epsilon_e + [0,\paramfineh )^d \} \ge  N_h;
 \end{equation}
\item\label{propoutlierslimits}
$\lim_{h\to\infty} \beta_h\paramfineh^2=0$,
$\lim_{h\to\infty} L_{R_h+R_h^*}C_h/(\beta_h^{1/2}N_h)=0$, and
$\lim_{h\to\infty} (2+L_{R_h+R_h^*})R_h^*M_h/N_h=0$.
\end{enumerate}
Let $\hat\sigma_{e,h}$ be defined as in \eqref{eqdefhatsigmahweigts} above, and let $\hat\sigma_h=(\hat\sigma_{1,h},\dots,\hat\sigma_{N,h})$.

Then for sufficiently large $h$, for any
$f \in \mathbb{R}^n$ the problem
\begin{equation}\label{wpL8eAh2}
    w B^T \hat{\sigma}_h(B u_h) = f
\end{equation}
has a solution $u_h$, and  up to subsequences the solutions $u_h$ converge to a solution $u$ of the continuous problem \eqref{wpL8eA}. If the functions $\hat\sigma_e$ are strictly monotone in the sense of~\eqref{eqsigmastrictmonotone}
then the solution $u$ of the limiting problem is unique and the entire sequence $u_h$ converges to $u$.
\end{prop}

\newcommand\I{\mathrm{I}}
\newcommand\II{\mathrm{II}}
\newcommand\III{\mathrm{III}}
\begin{proof}
Fix a possible strain $\epsilon_e\in\R^d$.
Assume first that $\|\epsilon_e\|_e\le R_h$.
From \eqref{eqdefsigmaehapp},
using $\sum_i p_{e,h,i}^*(\epsilon_e;\beta_h)=1$,
\begin{equation}\begin{split}
 \left\|\sigma_{e,h}^*(\epsilon_e)
 -\hat\sigma_e(\epsilon_e)\right\|_e
    =&\left\|
    \sum_{i=1}^{N_{e,h}}
    p_{e,h,i}^*(\epsilon_e;\beta_h) (\sigma_{e,h,i}-\hat\sigma_e(\epsilon_e)) \right\|_e\\
    \le &
    \sum_{i=1}^{N_{e,h}}
    p_{e,h,i}^*(\epsilon_e;\beta_h) \left\|\sigma_{e,h,i}-\hat\sigma_e(\epsilon_e)\right\|_e.
%     \\
%     =&
%    \sum_{p\in P} \sum_{i:\epsilon_{e,h,i}\in Q_p }
%     p_{e,h,i}^*(\epsilon_e;\beta_h) (\sigma_{e,h,i}-\hat\sigma_e(\epsilon_e))
    \end{split}
\end{equation}
With another triangular inequality, \eqref{eqhatsigmaloclip}, and \ref{propoutliersbound}, writing $L_h:=L_{R_h+R_h^*}$ for brevity,
\begin{equation}\begin{split}
\|\sigma_{e,h,i}-\hat\sigma_e(\epsilon_e)\|_e
\le&
\|\sigma_{e,h,i}-\hat\sigma_e(\epsilon_{e,h,i})\|_e+
\|\hat\sigma_e(\epsilon_{e,h,i})-\hat\sigma_e(\epsilon_e)\|_e\\
\le&
\|\sigma_{e,h,i}-\hat\sigma_e(\epsilon_{e,h,i})\|_e+
L_h\|\epsilon_{e,h,i}-\epsilon_e\|_e.
\end{split}
\end{equation}
We separate the two terms, and for the first one treat separately the values of $i$ for which it is larger than $\paramunifh$. We obtain, recalling that
$\sum_i p_{e,h,i}^*(\epsilon_e;\beta_h)=1$,
\begin{equation}\label{sigmanimsigma}\begin{split}
\| \sigma_{e,h}^*(\epsilon_e)
 -\hat\sigma_e(\epsilon_e)\|_e\le
&\paramunifh+\I+\II,
    \end{split}
\end{equation}
where
\begin{equation}
\label{eqdefIuII}
\begin{split}
\I:= &
     \sum_{i: \|
\sigma_{e,h,i}-
\hat\sigma_e(\epsilon_{e,h,i})\|_e> \paramunifh}
    p_{e,h,i}^*(\epsilon_e;\beta_h)
\|\sigma_{e,h,i}-\hat\sigma_e(\epsilon_{e,h,i})\|_e,\\
\II:=
&      L_h
    \sum_{i=1}^{N_{e,h}}
    p_{e,h,i}^*(\epsilon_e;\beta_h)
       \|\epsilon_{e,h,i}-\epsilon_e\|_e.
    \end{split}
\end{equation}
At this point we localize.
Let $P:=\paramfineh \Z^d$ and
$Q_p:=p+[0,\paramfineh)^d$. Obviously the cubes $\{Q_p\}_{p\in P}$ form a disjoint cover of $\R^d$.
For any $p\in P$, any $\epsilon_e$,
and any $q$, $q'\in Q_p$
we have
\begin{equation}
| \|q-\epsilon_e\|_e-
\|q'-\epsilon_e\|_e|\le \|q'-q\|_e\le
\mathrm{diam} (Q_p)=
\paramfineh k_e,
\end{equation}
where $k_e:=\mathrm{diam}([0,1]^d)$ is a constant that depends only on $d$ and on the norm $\|\cdot\|_e$. Therefore
\begin{equation}
 \frac12 \|q-\epsilon_e\|_e^2
-k_e^2\paramfineh^2\le
\|q'-\epsilon_e\|_e^2\le
2\|q-\epsilon_e\|_e^2+2k_e^2\paramfineh^2.
\end{equation}
Averaging over $Q_p$, we obtain that  for any $q'\in Q_p$ and any $\beta>0$
\begin{equation}\label{eqdisccont1}
\frac{{\rm e}^{-2k_e^2\beta \paramfineh^2 }}{|Q_p|}
\int_{Q_p}
{\rm e}^{- 2\beta \|q-\epsilon_e\|_e^2} dq
\le
{\rm e}^{-\beta \|q'-\epsilon_e\|_e^2}
\le
\frac{{\rm e}^{k_e^2\beta  \paramfineh^2 }}{|Q_p|}
\int_{Q_p}
{\rm e}^{-\frac12\beta \|q-\epsilon_e\|_e^2} dq.
\end{equation}
Therefore, using the definition \eqref{eqdefzeh},
then \eqref{eqdisccont1} and
\ref{propoutliersunif}
give
\begin{equation}\label{eqzetaupper}
\begin{split}
 Z_{e,h}(\epsilon_e;\beta)&=
 \sum_{p\in P}
  \sum_{i: \epsilon_{e,h,i}\in Q_p}
{\rm e}^{- \beta \|\epsilon_{e,h,i}-\epsilon_e\|_e^2}\\
&\le  \frac{{\rm e}^{ k_e^2 \beta \paramfineh^2}}{|Q_p|}
\sum_{p\in P}
\sum_{i: \epsilon_{e,h,i}\in Q_p}
\int_{Q_p}
{\rm e}^{- \frac12\beta \|q-\epsilon_e\|_e^2}dq
\\
&\le
\frac{C_h{\rm e}^{  k_e^2\beta \paramfineh^2 }}{|Q_p|}
\sum_{p\in P}
\int_{Q_p}
{\rm e}^{- \frac12\beta \|q-\epsilon_e\|_e^2}dq
\\
&=
\frac{C_hc_e}{|Q_p|}
{\rm e}^{  k_e^2\beta \paramfineh^2}
(\beta/2)^{-d/2},
\end{split}\end{equation}
where in the last step we used that by scaling there is a constant $c_e>0$ such that
\begin{equation}\label{eqconstgauss}
 \int_{\R^d} {\rm e}^{-\beta\|q\|_e^2} dq = c_e \beta^{-d/2},
\end{equation}
for any $\beta>0$.

A similar computation using \ref{propoutliersdenst} instead of \ref{propoutliersunif}  leads to
\begin{equation}\label{eqzetalower}
\begin{split}
Z_{e,h}(\epsilon_e;\beta_h)=&
 \sum_{p\in P}
  \sum_{i: \epsilon_{e,h,i}\in Q_p}
{\rm e}^{- \beta_h \|\epsilon_{e,h,i}-\epsilon_e\|_e^2}\\
\ge&
\frac{N_h{\rm e}^{-2 d\beta_h \paramfineh^2}}{|Q_p|}
 \sum_{p\in P, |p|< 2R_h-\paramfineh k_e}
  \int_{Q_p} {\rm e}^{- 2\beta_h \|q-\epsilon_e\|_e^2} dq
\\
\ge&
\frac{c_e' N_h{\rm e}^{-2 k_e^2\beta_h \paramfineh^2}}{|Q_p|}(2\beta_h)^{-d/2},
\end{split}\end{equation}
where in the last step we used that
$\|\epsilon_e\|\le R_h$ and
that there is $c_e'>0$ such that
\begin{equation}
 \int_{B_R} {\rm e}^{-\beta\|q\|_e^2} dq \ge c_e'
\beta^{-d/2}
\end{equation}
whenever $R\ge1$ and $\beta\ge1$.

We next estimate the two terms introduced in \eqref{eqdefIuII}.
We write
\begin{equation}
 \II= L_h \frac{1}{Z_{e,h}(\epsilon_e;\beta_h)} \sum_{i=1}^{N_{e,h}}
{\rm e}^{- \beta_h \|\epsilon_{e,h,i}-\epsilon_e\|_e^2} \|\epsilon_{e,h,i}-\epsilon_e\|_e.
\end{equation}
We use that $t{\rm e}^{-t^2}\le {\rm e}^{-t^2/2}$ for all $t\ge 0$ to obtain
\begin{equation}\label{eqestII}
\begin{split}
 \II\le& L_h \beta_h^{-1/2}\frac{1}{Z_{e,h}(\epsilon_e;\beta_h)}
 \sum_{i=1}^{N_{e,h}}
{\rm e}^{- \frac12 \beta_h \|\epsilon_{e,h,i}-\epsilon_e\|^2_e}\\
=&  L_h \beta_h^{-1/2}\frac{Z_{e,h}(\epsilon_e;\frac12\beta_h)}{Z_{e,h}(\epsilon_e;\beta_h)}\\
\le &   L_h \beta_h^{-1/2}\frac{C_hc_e 4^d {\rm e}^{3k_e^2 \beta_h\paramfineh^2} }{N_hc_e'},
\end{split}\end{equation}
where in the last step we estimated the partition function with \eqref{eqzetalower} and \eqref{eqzetaupper}.
We remark that the last bound depends only on $h$ and is uniform in $\epsilon_e$.

For the term $\I$ we argue similarly. We use the uniform bound on the stresses in \eqref{propoutliersbound}, the Lipschitz condition,
and choose $h$ sufficiently large that
$\|\hat\sigma_e(0)\|_e\le R_h^*$ to write
\begin{equation}\begin{split}
 \|\sigma_{e,h,i}-\hat\sigma_e(\epsilon_{e,h,i})\|_e
 \le&
 \|\sigma_{e,h,i}\|_e+\|\hat\sigma_e(\epsilon_{e,h,i})-\hat\sigma_e(0)\|_e+\|\hat\sigma_e(0)\|_e\\
 \le& R_h^*+L_hR_h^*+\|\hat\sigma_e(0)\|_e
 \le (2+L_h)R_h^*.
 \end{split}
 \end{equation}
With the bound on the number of outliers in \ref{propoutliersout} and \eqref{eqdisccont1}
 we obtain, with a computation similar to \eqref{eqestII},
\begin{equation}\begin{split}
 \I\le &(2+L_h)R_h^*\frac{1}{Z_{e,h}(\epsilon_e;\beta_h)}\sum_{p\in P}
  \sum_{i: \epsilon_{e,h,i}\in Q_p,
  \atop  {\|\sigma_{e,h,i}-\hat\sigma_e(\epsilon_{e,h,i})
  \|_e> \paramunifh}
}
{\rm e}^{- \beta_h \|\epsilon_{e,h,i}-\epsilon_e\|_e^2}\\
\le&  \frac{(2+L_h)R_h^*M_h{\rm e}^{k_e^2\beta_h\paramfineh^2}}{|Q_p|} \frac{1}{Z_{e,h}(\epsilon_e;\beta_h)}
\sum_{p\in P}
\int_{Q_p}
{\rm e}^{- \frac12 \beta_h \|q-\epsilon_e\|_e^2}dq.
\end{split}\end{equation}
With \eqref{eqconstgauss} and
\eqref{eqzetalower} we conclude
\begin{equation}\label{eqestI}\begin{split}
 \I
\le& (2+L_h)R_h^*
\frac{c_eM_h 2^{d}{\rm e}^{3k_e^2\beta_h\paramfineh^2} }{c_e'N_h}.
\end{split}\end{equation}
Collecting terms,
\eqref{sigmanimsigma}, \eqref{eqestII} and
\eqref{eqestI} give
\begin{equation}
\| \sigma_{e,h}^*(\epsilon_e)
 -\hat\sigma_e(\epsilon_e)\|_e\le
\paramunifh+c_e''
(1+L_h)R_h^*
\frac{ M_h {\rm e}^{3k_e^2\beta_h\paramfineh^2} }{N_h}+c_e''' \frac{L_h C_h {\rm e}^{3k_e^2\beta_h\paramfineh^2} }{\beta_h ^{1/2}N_h},
\end{equation}
with $c_e''$ and $c_e'''$ depending only on the dimension $d$ and the norm $\|\cdot\|_e$.
Recalling the assumptions in \ref{propoutlierslimits},
\begin{equation}
\omega_h:=\sup_{\epsilon_e: \|\epsilon_e\|_e<R_h} \| \sigma_{e,h}^*(\epsilon_e)
 -\hat\sigma_e(\epsilon_e)\|_e
\end{equation}
tends to zero as $h\to\infty$.
In particular, from \eqref{eqsigmastable},
\begin{equation}
 \sigma_{e,h}^*(\epsilon_e)\cdot\epsilon_e
 \ge
 \hat\sigma_{e}(\epsilon_e)\cdot\epsilon_e
 -\omega_h\|\epsilon_e\|_e
 \ge a\|\epsilon_e\|_e^2-b-\omega_h\|\epsilon_e\|_e
 \ge \frac12 a\|\epsilon_e\|_e^2-2b
\end{equation}
provided $h$ is sufficiently large to have $\omega_h\le\sqrt{ab}$.
As $\C_e$ was chosen positive definite, $\C_e\epsilon_e\cdot\epsilon_e\ge \hat a\|\epsilon_e\|^2$, and using  \eqref{eqdefhatsigmahweigts}
and $\psi\in[0,1]$ we conclude that
$\hat\sigma_{e,h}$ is coercive. Existence of solutions
follows then from Proposition~\ref{8CJa9S}.

Finally, the conditions $R_h\to\infty$ and $\omega_h\to0$ imply uniform convergence of
$\sigma_{e,h}^*$ to $\hat\sigma_e$, and with $R_h\to\infty$ and
 \eqref{eqdefhatsigmahweigts}
the same holds for
$\hat\sigma_{e,h}$.
By Proposition~\ref{abAx6h} the solutions converge.
 \end{proof}


\end{document}